\documentclass[11pt]{article}
\usepackage{graphicx}
\usepackage{makecell}
\usepackage{amsfonts,epsfig}
\usepackage{amsmath}
\usepackage{amsbsy}
\usepackage{amssymb}
\usepackage{dsfont}
\usepackage{url}
\usepackage[hidelinks]{hyperref}
\usepackage{fancyhdr}
\usepackage{indentfirst}
\usepackage{enumerate}
\usepackage{amsthm}
\usepackage{color}
\usepackage{comment}
\usepackage[authoryear]{natbib}
\usepackage{bm}
\usepackage[compact]{titlesec}
\newcommand{\dd}{\mathrm{d}}
\def\d{\mathrm{d}}
\def\laweq{\buildrel \d \over =}
\def\seq{\buildrel S \over \sim }

\newcommand{\VaR}{\mathrm{VaR}}

\newcommand{\E}{\mathbb{E}}
\renewcommand{\L}{\mathcal{L}}
\newcommand{\R}{\mathbb{R}}
\newcommand{\N}{\mathbb{N}}
\newcommand{\p}{\mathbb{P}}

\newcommand{\id}{\mathds{1}}

\renewcommand{\(}{\left(}

\renewcommand{\ge}{\geqslant}
\renewcommand{\le}{\leqslant}
\renewcommand{\geq}{\geqslant}
\renewcommand{\leq}{\leqslant}
\renewcommand{\epsilon}{\varepsilon}

\newcommand{\esssup}{\mathrm{ess\mbox{-}sup}}
\newcommand{\essinf}{\mathrm{ess\mbox{-}inf}}

\usepackage[onehalfspacing]{setspace}

\theoremstyle{plain}
\newtheorem{theorem}{Theorem}
\newtheorem{corollary}{Corollary}
\newtheorem{lemma}{Lemma}
\newtheorem{proposition}{Proposition}
\theoremstyle{definition}
\newtheorem{definition}{Definition}
\newtheorem{example}{Example}

\theoremstyle{remark}
\newtheorem{remark}{Remark}
\usepackage[misc]{ifsym}

\usepackage{tikz}
\usetikzlibrary{arrows.meta, positioning}
\usetikzlibrary{positioning}
\usepackage{extarrows} 

\usepackage{subcaption}
\usepackage{booktabs}\usepackage{caption}
\newcommand{\tabincell}[2]{\begin{tabular}{@{}#1@{}}#2\end{tabular}} 
\usepackage[flushleft]{threeparttable}
\usepackage{floatrow}
\usepackage{float} \floatstyle{plaintop} \restylefloat{table} 
\usepackage[figuresright]{rotating} \usepackage{array}
\newcommand{\I}{\mathcal I}
\usepackage{tabularx,lscape,multirow}
\usepackage{epigraph}
\newcommand{\opName}[1]{{\operatorname{#1}}}
\begin{document}

\title{Rating Criteria, Risk Measures, and Risk Consistency}
\author{Nan Guo\thanks{Nankai-Taikang College of Insurance and Actuarial Science, Nankai University. \texttt{guon@nankai.edu.cn}}
\and Ruodu Wang\thanks{Department of Statistics and Actuarial Science, University of Waterloo. \texttt{wang@uwaterloo.ca}}
\and Chenxi Xia\thanks{Department of Statistics and Actuarial Science, University of Waterloo. \texttt{c2xia@uwaterloo.ca}}
\and Jingping Yang\thanks{School of Mathematical Sciences, Peking University. \texttt{yangjp@math.pku.edu.cn}}}
\date{}
\maketitle

\begin{abstract}
Categorical ratings are widely used by investors to screen financial products, such as defaultable securities and investment funds.
We model rating criteria as discrete transforms of risk measures, and introduce and study several natural notions of risk consistency that promote prudent investment decisions.
Three closely related notions of risk consistency are considered: with respect to risk aversion, the asset pooling effect, and the benefit of portfolio diversification. These notions can be formulated under multiple probability measures. 
For both the monetary and positive homogeneous (MPH) and state-dependent expected utility (SEU) classes, we show how these properties translate between rating criteria and their corresponding risk measures, and establish a hierarchical structure among the properties.
These findings lead to full characterizations for the Choquet subclass of MPH criteria and for SEU criteria. Illustrated by case studies on collateralized loan obligations and investment funds, some classes of risk-consistent scenario-based criteria can serve as useful alternatives to the probability of default and expected loss criteria used in practice for rating financial products.

\end{abstract}

\section{Introduction}

As a quantitative tool for investors to screen financial products, categorical ratings play an important role in modern financial systems.
Two prominent criteria for rating defaultable securities are the probability of default (PD) used primarily by Standard and Poor's and Fitch, 
and the expected loss (EL) criterion used primarily by Moody's; see \cite{SP2019CDO} and \cite{moodys2023cso}. These criteria are typically implemented by calculation under various probability measures, called scenarios,  
and we refer to those as the scenario-based PD and EL.\footnote{When rating structured products, \cite{SP2019CDO} considered how ``a given portfolio of corporate credits would suffer in various rating scenarios consistent with [their] rating definitions.'' (para.~8), and \cite{moodys2023cso} applied ``a three-point probability distribution to the inter-industry correlation to represent three separate states of low, medium and 
high correlations.'' (p.~4).}  Ratings are also used to guide investment-fund selection. Morningstar, for example, ranks funds by a risk-adjusted-return methodology grounded in expected utility theory; see \cite{M23}. 

The presence of various rating methodologies calls for discussions of their
respective advantages.
On the credit-rating side, existing studies have mainly focused on the comparison between PD and EL.
 \cite{brennan2009tranching} showed that under ratings-based pricing, 
structured securitization creates a much larger increase in total deal value
when using PD compared to EL. 
\cite{HW12} argued that PD leads to ratings arbitrage that could misinform investors, whereas EL does not.  
More recently, \cite{GKWW25} proposed a self-consistency axiom for rating criteria, which rules out profits from excessive tranching and showed that it characterizes Choquet rating criteria. This class includes scenario-based EL but excludes PD-type criteria.
The above works suggest that EL is theoretically preferable to PD as a rating criterion.\footnote{On the other hand, PD may have practical advantages, such as being easy to compute, implement, or interpret.}  

The extensive literature on risk measures suggests that EL is best understood as a simple but limited benchmark (see \citealp{MFE15} and \citealp{FS16} for overviews). It is insensitive to volatility, which has been central to financial decision making since \cite{M52}; it treats upside and downside movements in financial asset returns symmetrically; and, because of its linearity, it does not capture diversification effects. Thus, while scenario-based EL incorporates economic scenarios and satisfies the self-consistency requirement, it still does not fully reflect risk aversion.

On the fund-rating side, the situation is almost the opposite. Morningstar's measure is motivated by risk aversion and evaluates funds through risk-adjusted returns. However, it remains law invariant and therefore evaluates a fund through the unconditional distribution of returns. \cite{S98} argued that neither Morningstar's measure nor the Sharpe ratio is an efficient tool for choosing mutual funds within peer groups for a multi-fund portfolio, because neither incorporates information about correlations. This suggests that fund ratings should account for state-contingent information.

Motivated by these complementary drawbacks of existing rating criteria, this paper develops a framework for rating criteria that combines scenario relevance with risk-aversion and diversification considerations.
We start in Section \ref{sec:prelim} by introducing two broad classes of risk measures. 
The first class consists of monetary and positive homogeneous (MPH) risk measures \citep{ADEH99}. It includes the important 
  subclass  of  Choquet risk measures (\citealp{S86, S89}; \citealp{EMWW21}), and in particular, the Value-at-Risk (VaR) and the Expected Shortfall (ES).
The second class consists of state-dependent expected utility (SEU) functionals \citep{WZ99}, motivated by utility-based fund performance evaluation such as Morningstar's methodology. This class naturally incorporates economic states into fund ratings. 
The two applications also require different natural loss domains: $[0,1]$ for defaultable securities and $(-\infty,1]$ for investment funds, since they may generate positive returns.

We then define rating criteria as discrete increasing transforms of risk measures in Section \ref{sec:prelim}, including MPH rating criteria, with Choquet rating criteria as a subclass, and SEU rating criteria.  
    To distinguish unconditional from state-contingent rating criteria, we introduce law invariance [LI]\footnote{Law invariance is a common property in risk measures, as in e.g., \cite{K01} and \cite{S13}, and it is recently generalized to partial law invariance by \cite{SVW24}.} and its scenario-based counterpart $S$-law invariance [$S$-LI], where  $S$ represents a collection of scenarios, as in \cite{GKWW25}. 
Propositions \ref{prop:LI_SR} and \ref{prop_SEU_LI_SLI} provide the corresponding [LI] and [$S$-LI] representations for Choquet and SEU rating criteria, respectively.

We consider three inter-connected desirable notions of consistency for rating criteria. We call the corresponding properties \emph{risk consistency properties}, formally introduced in Section \ref{sec_rc}. 
The first is based on the procedure of asset pooling, which is an equally fundamental  securitization structure alongside tranching \citep{DeMarzo05, Diamond23}, and thus deserves great attention.
This notion has two closely related versions.
At its most basic level, pooling monotonicity, abbreviated as [PM], requires that a homogeneous asset pool should become safer as the number of assets increases. Here, a homogeneous pool is modeled by a sequence of conditionally independent and identically distributed (iid) asset losses, commonly used in portfolio modeling. For defaultable securities, particularly structured credit products, pooling is often combined with tranching. This leads to  pooling effect consistency, abbreviated as [PE]: holding the tranching scheme fixed, the most senior tranche should be safer as the number of assets in the pool increases. The property [PE] is introduced and briefly discussed in \citet[Appendix D]{GKWW25}. In parallel, we define the scenario-based version of [PM] and [PE], abbreviated as [$S$-PM] and [$S$-PE]. 

The second notion is based on risk aversion, a prominent concept in decision analysis \citep{RS70}, which provides a robust way to formulate desirable ordering among random losses. To reflect this concept, we consider rating criteria satisfying risk aversion consistency, abbreviated as [RA], meaning that the rating outcomes respect the unanimous preferences of risk-averse expected utility investors.  In parallel, we define the scenario-based version of [RA], abbreviated as [$S$-RA].

The third notion directly incorporates diversification. 
In the risk measure literature, the    mathematical notions that best describe the benefit of diversification are convexity and  quasi-convexity, as discussed extensively by \cite{FS02} and \cite{CMMM11}. 
Quasi-convexity,
  abbreviated  as [QC],
is more suitable than convexity for discrete valued mappings. For rating criteria, [QC] means that 
combining a defaulted bond with a better-rated one 
cannot make its rating worse than the original bond.

The seven risk consistency properties [PM], [$S$-PM], [PE], [$S$-PE], [RA], [$S$-RA] and [QC] share some common ideas, as asset
pooling, risk aversion and diversification are closely connected notions, but none of them is equivalent to any other.  
We find two chains of implications in Theorem \ref{thm:summarize} (with a weak notion of continuity):
under law invariance, we have 
[QC] $\Rightarrow$ [RA] $\Rightarrow$ [PM] and [PE]; 
under $S$-law invariance,  we have   [QC] $\Rightarrow$ [$S$-RA] $\Rightarrow$ [$S$-PM] and [$S$-PE]. 
Thus, in both settings, [QC] is the strongest notion among the three consistency properties.

We connect many properties of
rating criteria to those of the representing risk measures in Section \ref{sec:translation}. We show in Theorem \ref{thm_uniqueness} that, the representing risk measure is unique for MPH rating criteria and unique up to scaling for SEU rating criteria. 
Building on this result, Theorem \ref{th:convert} shows that, for both the MPH and SEU classes, [LI], [$S$-LI], [RA], [$S$-RA], [PM], [$S$-PM], and [QC] can all be freely translated from risk measures to rating criteria and back.
For SEU rating criteria on $A=[0,1]$, the two-way translation also holds for [PE] and [$S$-PE].

Section \ref{sec:choquet_chara} characterizes Choquet rating criteria under law invariance and $S$-law invariance,  by first addressing Choquet risk measures, and then extending the results to Choquet rating criteria. 
Under law invariance, Choquet risk measures are commonly known as distortion risk measures  or dual utilities, widely used in economics, insurance and finance (e.g., \citealp{Y87}, \citealp{WYP97} and \citealp{A02}). We show in Theorem \ref{th:PE-LI} that a distortion risk measure satisfies [PM] or [PE] if and only if its distortion function is concave on $[0,1)$. 
Consequently, for a Choquet rating criterion  that satisfies [LI] and an assumption of semi-continuity, 
[QC], [RA], [PM] and [PE]  are equivalent (Proposition \ref{prop:LI-rating}). 
Under $S$-law invariance, Choquet risk measures are called  $S$-distortion risk measures, studied by \cite{WZ21}.
We show  in Theorem \ref{th:SR-Choquet} that, for an $S$-distortion risk measure,
[QC] is equivalent to a componentwise concave and submodular $S$-distortion function. Moreover, 
under an additional continuity assumption, these conditions are equivalent to [$S$-RA], [$S$-PM] and [$S$-PE]. This gives rise to corresponding characterization results on the Choquet rating criteria (Proposition \ref{prop:SR-rating}). From this result we  can construct many families of rating criteria satisfying the risk consistency properties.

Section \ref{sec:SEU_chara} develops full  characterization results for SEU rating criteria.    
Theorem \ref{thm:char-SEU} shows that [QC] is characterized by convex statewise loss functions, whereas [$S$-RA] is characterized by convex scenario-specific loss functions. 
For [$S$-PM] and [$S$-PE], the characterization takes a related form, but allows an additional nonnegative term concentrated at the boundary $X=1$. 
Taking $S=\Omega$ yields the corresponding characterizations for [RA], [PM], and [PE]. 
The representation results corresponding to [QC], [PM], [$S$-PM], [PE], and [$S$-PE] are not confined to [LI] or [$S$-LI]. Theorem \ref{thm_SR_SEU} further shows that, when the representing risk measure satisfies [LC] and [LI] (respectively, [$S$-LI]), [QC], [RA] ([$S$-RA]), [PM] ([$S$-PM]), and [PE] ([$S$-PE]) are equivalent and are characterized in a unified form. 
Establishing these results requires technically sophisticated arguments, because convexity of extended-real-valued statewise loss functions needs to be recovered from the coarse information contained in categorical ratings. 
These results provide a state-dependent expected-utility foundation for risk-consistent rating criteria that is structurally distinct from the MPH framework.

In Section \ref{sec:case}, we provide two case studies on rating 
collateralized loan obligations and investment funds. 
The CLO case shows that risk-consistent scenario-based criteria, such as Average ES and Average MAXVAR, offer useful improvements over the PD and EL criteria currently employed by major credit rating agencies. The fund case shows how an SEU rating criterion can reveal state-contingent performance differences that are not captured by a law-invariant Morningstar-type measure.
Section \ref{sec:conc} concludes the paper. 

All proofs are presented in  Appendix \ref{app:A}, with a technical lemma to support the proof of Theorem \ref{th:SR-Choquet} included in Appendix \ref{app:L5}. Appendices \ref{appx_axioms}--\ref{appx:necessaryPE} provide supplementary theoretical results, while Appendices \ref{appx:CATbond} and \ref{appx:mapping} contain an additional case study and numerical details.

\section{Risk measures and rating criteria}\label{sec:prelim}

\subsection{Risk measures}\label{sec:rm}

Let $(\Omega,\mathcal F,\p)$ be a fixed atomless probability space, and $\L^\infty$ be the set of all bounded random variables on $(\Omega,\mathcal F,\p)$. Let $\overline{\R}=\mathbb{R}\cup\{\infty\}$. 
\begin{definition}
    \label{def:rm}
   
A \emph{risk measure} $\rho$ is a mapping
from a convex set $\L_\rho\subseteq \L^\infty$ to $\overline{\R}$,  satisfying \emph{monotonicity}: $\rho(X)\leq \rho(Y)$ for $X,Y\in \L_\rho$ with $X\leq Y$.
\end{definition}
Throughout, terms like ``increasing'' and ``decreasing'' are in the non-strict sense.
For an interval $A \subseteq \R$, let $\L_A$ denote the subset of $\mathcal{L}^\infty$ consisting of random variables taking values in $A$.

We consider a one-period model for financial products subject to rating, with
total nominal amount normalized to $1$ and random loss $X\in\L_A$. Two types of
products are considered. For defaultable
securities, losses are bounded between zero and the notional amount, so 
$X\in\L_{[0,1]}$; for investment funds, gains are allowed but the maximal loss is bounded by the initial investment, so $X\in\mathcal L_{(-\infty,1]}$. Therefore, throughout, $A$ is either $[0,1]$ or $(-\infty,1]$,
and all risk measures have a domain that contains $\L_A$. 
We next introduce the classes of risk measures suitable for these two settings.

\subsubsection{Monetary and positive homogeneous risk measures}
\label{sec:21-def-MPH} 

First, we recall some basic classes of  risk measures, using the terminology from \cite{FS16}.  A mapping $\rho:\mathcal{L}^\infty\to \mathbb{R}$ is 
\begin{enumerate}[(a)]
    \item a \emph{monetary risk measure} if it satisfies (i) {monotonicity} and (ii) \emph{translation invariance}: $\rho(X+a)=\rho(X)+a$ for $a\in \mathbb{R}$ and $X\in\mathcal{L}^\infty$; 
   
    \item a \emph{positive homogeneous risk measure} if it satisfies (i) and (iii) \emph{positive homogeneity}: $\rho(\lambda X)=\lambda \rho(X)$ for $\lambda > 0$ and $X\in\mathcal{L}^\infty$; 
    \item a \emph{convex risk measure} if it satisfies (i), (ii), and (iv) \emph{convexity}: $\rho(\lambda X+(1-\lambda) Y)\leq \lambda \rho(X)+(1-\lambda) \rho(Y)$ for $X,Y\in \mathcal{L}^\infty$ and $\lambda \in [0,1]$;
\item a \emph{coherent risk measure} if it satisfies (i), (ii), (iii), and (iv).
\end{enumerate}

\begin{definition}
     A \emph{ monetary and positive homogeneous (MPH) risk measure} is a monetary risk measure that satisfies positive homogeneity.
\end{definition}
MPH risk measures provide a canonical benchmark for
credit risk measurement. Familiar examples include the expectation, VaR, and ES.
Our MPH risk measures are defined on $\L^\infty$. We will also consider their restrictions to
$\L_A$.

Next, we formulate an important subclass of MPH risk measures,   the Choquet risk measures.

\begin{definition} 
\label{def:CRM}
A \emph{capacity} is an increasing function $\nu:\mathcal F\to \R$ satisfying $\nu (\varnothing)=0$ and $\nu(\Omega)=1$.
The \emph{Choquet integral} of $X\in \L^\infty$ with respect to a capacity $\nu$ is defined as
 $$
		  \int X \d \nu=  \int_0^\infty \nu (X>x)\d x+
        \int_{-\infty}^0\left( \nu (X>x)-1\right) \d x .
 $$
A \emph{Choquet risk measure} is a  mapping $\rho:\L^\infty\to \R$ given by,     for some  capacity $\nu$,
\begin{equation}\label{eq:scpresentation}
		\rho(X)= \int X \d \nu,~~~~X\in \L^\infty.
	\end{equation} 
 \end{definition}
 
Choquet integrals have been prominent tools in decision theory since the seminal work of \cite{S89}, and a subclass of Choquet risk measures, called  distortion risk measures or dual utility functionals, is widely used in economics, mathematical finance, and insurance; see e.g., \cite{Y87}, \cite{WYP97}, \cite{MFE15} and \cite{FS16}.
We use the term ``{Choquet risk measures}" following \cite{EMWW21}. 

Choquet risk measures satisfy many convenient properties. 
Following from the main result of \cite{S86}, a mapping $\rho:\L^\infty\to \R$ is a Choquet risk measure if it satisfies (i), (v) \emph{normalization}: $\rho(1)=1$, and (vi)  \emph{comonotonic additivity}:  $\rho(X+Y)=\rho(X)+\rho(Y)$ whenever $X$ and $Y$ are comonotonic.\footnote{Two random variables $X$ and $Y$ are \emph{comonotonic} if they can be written as increasing functions of some random variable $Z$.} 
Among these properties, comonotonic additivity is the essential property that underpins Choquet integrals.

A  Choquet risk measure $\rho$ in \eqref{eq:scpresentation} 
always satisfies monotonicity, translation invariance, and positive homogeneity; thus, it is an MPH risk measure. Further, it
 is coherent (equivalently, convex) if and only if 
$\nu$ is \emph{submodular}, meaning
$\nu(A\cap B) + \nu(A\cup B) \le \nu(A)+\nu(B)$ for all $A,B\in \mathcal F$;
see Theorem 4.94 of \cite{FS16}.

\subsubsection{State-dependent expected utility risk measures}
\label{sec:22-def-SEU}

Another broad class relevant for rating purposes consists of the utility-based risk measures, which play an important role in investment fund performance evaluation. A notable example is the rating criterion of \cite{M23}, which is one of the world's largest fund-rating firms. In our loss convention, its rating (i.e., the fund performance ranking) is  generated by
\begin{equation}\label{eq_morningstar}
   \rho(L)=\mathbb{E}[(1-L)^{-2}-1],\quad L \in \L_{(-\infty,1]},  
\end{equation}
where $0^{-2}=\infty$.
Morningstar's criterion has a direct expected-utility interpretation. Let $W=1-L$ denote terminal wealth relative to the initial investment. 
Minimizing \eqref{eq_morningstar} is equivalent to maximizing expected utility
$\mathbb E[u(W)]$ with
$u(W)=1-W^{-2}$. Moreover, \eqref{eq_morningstar} may take the value $\infty$ when $L=1$, which explains why $\infty$ is included in the range of risk measures in Definition \ref{def:rm}.

Motivated by this example, we consider the following state-dependent generalization of expected utility, following \cite{WZ99}.

\begin{definition}
    A \emph{state-dependent expected utility (SEU) risk measure} is a mapping
    $\rho:\L_A\to\overline{\R}$ given by
    \begin{equation}\label{eq_def_seu_wakker}
        \rho(X)=\int_{\Omega}v_\omega(X(\omega))\d\mu,
        \qquad X\in\L_A,
    \end{equation}
    for some finite positive measure $\mu\ll\p$ and increasing functions
    $v_\omega:A\to\overline{\R}$, $\omega\in\Omega$, such that
    $\omega\mapsto v_\omega(X(\omega))$ is $\mathcal F$-measurable for all
    $X\in\L_A$, $\rho(0)=0$, and $0<\rho(x)<\infty$ for all $x\in(0,1)$.
\end{definition}

Note that $\rho(x)<\infty$ for $x<1$
requires that $v_\omega$ is finite on $A\setminus\{1\}$ $\mu$-almost surely. 
If $v_\omega$ is the same across all $\omega$, then $\rho$ is indeed an expected-utility functional.
In the following, we call $v:A\to\overline{\mathbb R}$ a \emph{loss function}, if $v$ is increasing on $A$, finite on $A\setminus\{1\}$, and $v(0)=0$.

For subsequent analysis, it is convenient to express every SEU risk measure under the measure $\mathbb P$. The following lemma provides such a representation.
\begin{lemma}\label{lemma:normal_SEU}
A mapping
$\rho:\mathcal L_A\to\overline{\mathbb R}$ is an SEU risk measure if and only if $\rho$ admits the form
\begin{equation}
    \rho(X)=\int v_\omega(X(\omega))\p(\d \omega)=  \mathbb{E}[v_{\cdot}(X)],\label{eq_SEU_expectation}
\end{equation}
for some class $v_{\cdot}=(v_\omega)_{\omega\in \Omega}$ of loss functions and $0<\rho(x)<\infty$ for $x\in (0,1)$, where
$v_{\cdot}(X):\omega \mapsto v_\omega(X(\omega))$.
\end{lemma}

Although \eqref{eq:scpresentation} and \eqref{eq_SEU_expectation} are both integral representations, they capture nonlinearity through fundamentally different yet complementary channels: a nonadditive capacity in the Choquet case and state-dependent loss functions in the SEU case. This structural distinction is reflected more broadly in the following result, which shows that the intersection of MPH and SEU risk measures consists precisely of expectation functionals.

\begin{proposition}\label{prop_MPH_SEU_intersection}
    A mapping $\rho:\L_A\to \mathbb{R}$ is both an MPH and SEU risk measure 
    if and only if
        \begin{equation}
            \rho(X)=\int_\Omega X\dd \mu,\quad X\in\mathcal L_A,\label{eq_Expectation}
        \end{equation}
		for some probability measure $\mu \ll \p$ on $(\Omega,\mathcal{F})$.
	\end{proposition}
As a consequence, the  results on MPH risk measures and SEU risk measures require separate studies.

\subsection{Rating criteria}\label{sec:prelim_rating}
Categorical ratings are assigned to financial products based on rating criteria. We suppose that there are  $n \geq 2$ rating categories. 
Here, $1$ represents the best rating (e.g., Aaa in Moody's and AAA in S\&P), and $n$ represents the worst rating.  Write $[n]=\{1,\dots,n\}$   for $n\in \N$. A \emph{rating criterion} $\I$ is a mapping from $\L_{A}$ to $[n]$, satisfying \emph{monotonicity}:
$$\mbox{For } L_1,L_2\in \L_A,~ L_1\le L_2 ~\Longrightarrow ~\I(L_1)\le \I(L_2).$$
A rating criterion $\I $ is \emph{represented} by a risk measure $\rho:\L_A\to \overline \R$ if 
\begin{equation}\label{eq:p2r}
	\I(X) = f(\rho(X)), ~~~\mbox{for all  $ X \in \L_A$,}
\end{equation}
for some  increasing  function $f:\overline \R\to [n]$ 
possessing an effective cut point. Here, $c\in\R$ is called an \emph{effective cut
point} if $f(c-)<f(c+)$ and $\rho(X_-) < c < \rho(X_+)<\infty$ for some $X_-,X_+ \in \L_A$.
The technical assumption that $f$ has an effective cut point is automatically satisfied if $\I$ assigns at least $5$ different rating categories, which is standard in financial applications. \footnote{Indeed, suppose that $X_1,\ldots,X_5\in \L_A$ receive successively worse ratings. Since $\I=f\circ\rho$ and $f$ is increasing,
$\rho(X_1)<\rho(X_2)<\rho(X_3)<\rho(X_4)<\rho(X_5)\leq \infty$, and hence $\rho(X_4)<\infty$. Moreover, since $f(\rho(X_2))<f(\rho(X_3))$, there exists a jump point $c\in[\rho(X_2),\rho(X_3)]$ of $f$. Therefore, $\rho(X_1)<c<\rho(X_4)<\infty$, so $c$ is an effective cut point.}

Two most commonly used criteria for rating defaultable securities are the PD criterion and the EL criterion, both of which are   represented by some risk measures. Specifically, the PD criterion is represented by $\rho(X)= \p(X>0)$ for $ X \in \L_{[0,1]}$, while the EL criterion is represented by $\rho(X)= \E[X]$ for $ X \in \L_{[0,1]}$.

The main classes of rating criteria studied in this paper are defined through their representing risk measures as follows.

\begin{definition} A rating criterion $\mathcal I:\mathcal L_A\to[n]$ is called 
\begin{enumerate}[(i)] 
\item an \emph{MPH (resp. Choquet) rating criterion} if it is represented by the restriction to $\mathcal L_A$ of an MPH (resp. Choquet) risk measure on $\mathcal L^\infty$ via \eqref{eq:p2r}. 
\item an \emph{SEU rating criterion} if it is represented by an SEU risk measure on $\mathcal L_A$ via \eqref{eq:p2r}. 

\end{enumerate} 
\end{definition}

As introduced in Section \ref{sec:rm}, MPH risk measures constitute a broad and central class in the study of risk measures, leading to the MPH rating criteria. Moreover, one important subclass, Choquet rating criteria, also has a clear economic foundation. Based on the MM Theorem  of \cite{MM58}, \cite{GKWW25} proposed an axiom of self-consistency for rating criteria. This axiom prevents issuers from exploiting rating criteria to increase the total deal value merely through restructuring. They showed that  Choquet rating criteria are the only class that satisfies self-consistency. 
By contrast, the SEU rating criterion has a different decision-theoretic foundation. In particular, it is motivated by the sure-thing principle \citep{WZ99}: when two funds have identical losses outside an event $B \in \mathcal F$, their relative assessment should depend only on their losses on $B$, and should not be affected by replacing the common losses on $B^c$.

A detailed discussion of the self-consistency axiom and the sure-thing principle is provided in Appendix \ref{appx_axioms}.

\subsection{Law invariance and scenario-based law invariance}
 
To connect with statistical inference and probabilistic modeling, we need to specify how risk and rating assessments are governed by probabilities that can be inferred from data and models. 
For this reason, risk measures and
 rating criteria  in practice  often satisfy some notion of law invariance. 
We will propose two notions of law invariance, a classic one in the  literature of  risk measures (e.g., \citealp{K01}) and a relatively new one based on multiple probability measures.

The first property of law invariance requires that the rating of a financial product is determined by the distribution of the potential loss.
\begin{enumerate}
	\item[{[LI]}] \emph{Law invariance}:
	$\I(L_1)=\I(L_2)$ for all $L_1,L_2\in \L_{A}$ satisfying $ L_1\laweq L_2$,
	where $\laweq$ stands for equality in distribution.
\end{enumerate}

If we replace the rating criterion $\I$ above  by a risk measure $\rho$ defined on the corresponding domain, then the same property is defined for $\rho$. We will tacitly define all properties for risk measures in this way; that is, whenever a property is defined for rating criteria, it is also defined for risk measures.

The next property further makes the probabilistic assessment under different scenarios. 
In industry practice, the loss distribution is often evaluated under different economic scenarios; see e.g., the method of \cite{moodys2023cso}. 
For a thorough and effective evaluation of risk, it is crucial to consider both the economic scenarios and the loss distribution conditionally on the realization of  economic scenarios. 
In the following discussion, we consider a fixed collection of scenarios $S =(S_1,\dots ,S_s) \in \mathcal F^s $, where $s$ is a positive integer. Here $S_1,\dots,S_s$ are disjoint events with non-zero probability, representing economic scenarios of relevance to the rating agency,
 and satisfying $\bigcup_{j=1}^s S_j=\Omega$. For instance, $S$ may represent economic regimes,  different levels of strength of the current economic growth, or indices of  financial stress.  

Denote by $\mathbb{P}_{j}=\mathbb{P}(\cdot|S_{j})$ and   $\mathbb{P}^{S}=(\mathbb{P}_1, \dots, \mathbb{P}_s)$.  
As the probability $\p$ is atomless and each $S_j$ has a positive probability, $\p_j$ is also atomless for $j\in [s]$.  
We say
that two losses  $L_1,L_2\in \L_A$ are equivalent under $S$, denoted by  $ L_1\seq L_2$, if $L_1$ and $L_2$ are identically distributed 
under $\p_j$ for each $j\in [s]$.

As a generalization of law invariance, the following property means that
equivalent losses under $S$ are assessed as equally risky.

\begin{enumerate}[{[$S$-PE]}]
\item[{[$S$-LI]}] \emph{$S$-law invariance}: 
	$\I(L_1)=\I(L_2)$ for all $L_1,L_2\in \L_A$ satisfying $ L_1\seq L_2$.
\end{enumerate}

Classic decision models,
such as  the maxmin expected utility model of \cite{GS89} and many other ambiguity preferences \citep{KMM05, MMR06}, 
satisfy a property similar to [$S$-LI] with endogenous probabilities.\footnote{More precisely, they satisfy: there exists a set $\mathcal Q$ of probability measures such that  
$\I(L_1)=\I(L_2)$ for all $L_1,L_2$ identically distributed under each $Q\in \mathcal Q$.} The formulation of [$S$-LI] with exogenous probabilities $(\mathbb{P}_1, \dots, \mathbb{P}_s)$ was introduced in \cite{WZ21}, where the probability measures do not need to be mutually singular.

Property [LI] is stronger than [$S$-LI], because $ L_1\laweq L_2$ 
is a weaker requirement than $ L_1\seq L_2$ (see also Theorem \ref{thm:summarize}).
\citet[Theorem 3.4]{WZ21} characterized Choquet risk measures  satisfying [$S$-LI].  
\citet[Theorem  7 and Corollary 1]{GKWW25} obtained  representations of Choquet rating criteria satisfying [LI] or [$S$-LI].
We summarize these results   below. 
In what follows, $\mathbf 0$ and $\mathbf 1$ are the vector of zeros and the vector of ones, respectively, in $\R^s$.

\begin{proposition}\label{prop:LI_SR}
	For a Choquet risk measure $\rho$,
	\begin{enumerate}[(i)]
    \item $\rho$ satisfies [LI] if and only if  	\begin{equation}
			\rho (X)=  \int X \d (h\circ \p),~~~X\in \L^\infty\label{eq_rho_LI}
		\end{equation}
		for some increasing function $h:[0,1]\to \R$ with $h(0)=0$ and $h(1)=1$.
			 \item $\rho$ satisfies [$S$-LI] if and only if
		\begin{equation}
			\rho (X)= \int X \d (g\circ \p^S), ~~~X\in \L^\infty,\label{eq_rho_SR}
		\end{equation}
		for some componentwise increasing function $g:[0,1]^s\to\R$  with $g(\mathbf 0)=0$ and $g(\mathbf 1)=1$. 
	\end{enumerate} 
	For a Choquet rating criterion $\I$,
	\begin{enumerate}[(i)] 
		\item[(iii)]  $\I$ satisfies [LI] if and only if it is represented by some $\rho$ in \eqref{eq_rho_LI}. 
		\item[(iv)] $\I$ satisfies [$S$-LI] if and only if it is represented by some $\rho$ in \eqref{eq_rho_SR}.
	\end{enumerate}
\end{proposition}

The risk measure  $\rho$  in \eqref{eq_rho_LI}
is called a \emph{distortion risk measure} with \emph{distortion function} $h$,
and  $\rho$ in 
 \eqref{eq_rho_SR} 
will be called an \emph{$S$-distortion risk measure} with \emph{$S$-distortion function} $g$.

We remark that $h$ in \eqref{eq_rho_LI} is uniquely determined by $\rho$, through 
$h(x)= \rho(\id_{A_x})$ for $x\in [0,1]$, 
where the event $A_x$ satisfies $\p(A_x)=x$. 
Similarly, 
$g$ in \eqref{eq_rho_SR} is also uniquely determined by $\rho$,  through 
$g(\mathbf x)= \rho(\id_{A_{\mathbf x}})$ for $\mathbf x\in [0,1]^s$, 
where the event $A_{\mathbf x}$ satisfies $\p^S(A_{\mathbf x})=\mathbf x$ (such $A_{\mathbf x}$ exists because $\p_1,\dots,\p_s$ are mutually singular). 
Nevertheless, the same $\rho$ may be represented by both \eqref{eq_rho_LI} 
and \eqref{eq_rho_SR}, as [LI] implies [$S$-LI]. In that case, it holds that
$$
g(x_1,\dots,x_s)= h\left(\sum_{j=1}^s \p(S_j) x_j\right),~~~~(x_1,\dots,x_s)\in [0,1]^s.
$$

Next, we characterize SEU risk measures and rating criteria under [LI] and [$S$-LI]. 
    
	\begin{proposition}\label{prop_SEU_LI_SLI}
		For an SEU risk measure $\rho$ on $\mathcal L_A$,
		\begin{enumerate}[(i)]
			\item $\rho$ satisfies [LI] if and only if $\rho$ admits the form
			\begin{equation}\label{eq_SEU_LI}
			    \rho(X)=\mathbb E[v(X)],\quad X\in\mathcal L_A,
			\end{equation}
			for some loss function $v:A\to\overline{\R}$.
			
			\item $\rho$ satisfies [$S$-LI] if and only if $\rho$ admits  the form
			\begin{equation}\label{eq_SEU_SLI}
			    \rho(X)=\sum_{j=1}^s \mathbb E[v_j(X)\id_{ S_j}],\quad X\in\mathcal L_A,
			\end{equation}
			for some  loss functions  $v_j:A\to\overline{\R}$  for $j\in [s]$.
		\end{enumerate}
	For an SEU rating criterion $\I$,
	\begin{enumerate}[(i)] 
		\item[(iii)]  $\I$ satisfies [LI] if and only if it is represented by some $\rho$ in \eqref{eq_SEU_LI}. 
		\item[(iv)] $\I$ satisfies [$S$-LI] if and only if it is represented by some $\rho$ in \eqref{eq_SEU_SLI}.
	\end{enumerate}
        
	\end{proposition}

\section{Risk consistency properties} \label{sec_rc}

In this section, we formulate several risk consistency properties through the concepts of risk aversion, the asset pooling effect, and diversification.    
The relations among these properties will be made clear in Section \ref{sec:relations}.

\subsection{Pooling monotonicity and pooling effect consistency}
Asset pooling and tranching are two prevailing contract structures of securitizations in the financial market \citep{DeMarzo05, Diamond23}. While \cite{brennan2009tranching} and \cite{GKWW25} studied the effect of tranching, asset pooling represents another significant issue that should be considered in the design of rating criteria.

We consider conditionally iid sequences of losses from the underlying assets. Here, we say that a sequence $(L_\ell)_{\ell\in\N}$  of random variables is  {conditionally iid}, if there exists a  random vector $Z$,  such that $(L_\ell)_{\ell\in\N}$ is conditionally iid given $Z$. 
By the de Finetti Theorem,
$(L_\ell)_{\ell\in\N}$ is conditionally iid if and only if it is  exchangeable. 
 
As standard in asset management, the portfolio losses are standardized with the nominal value, denoted by $L^{(\ell)}={(L_1 + \cdots + L_\ell)}/\ell$ for a positive integer $\ell$, where $(L_\ell)_{\ell\in\N}$ is a sequence in $\mathcal L_A$. 

\begin{enumerate}[{[PM]}]
\item[{[PM]}] \emph{Pooling monotonicity}: For any sequence  $(L_\ell)_{\ell\in\N}$   in $\L_A$ that is conditionally iid, if $\ell\le \ell'$, then
$$\I \left(L^{(\ell)}\right)\geq \I\left(L^{(\ell' )}\right).$$
\end{enumerate}

The property [PM] for rating criteria is quite intuitive economically. It means that a homogeneous asset pool should become safer as the number of assets increases. In practice, portfolio models in credit risk may have more sophisticated dependence structures than conditional iid; see e.g., \cite{FM03}. Therefore, [PM] should be interpreted as a minimal natural requirement: the rating criterion $\I$ should at least be consistent with this simple portfolio model.  

To incorporate tranching, as is common in structured finance products, the following property [PE] concerns the impact of pooling on the loss of a senior tranche with some attachment point $K \in [0,1)$, denoted by $(L^{(\ell)}-K)_+$. Here, $(L_\ell)_{\ell\in\N}$ is a sequence in $\L_{[0,1]}$, and the tranche loss is standardized by the tranche notional $1-K$. 
\begin{enumerate}[{[PE]}]
\item[{[PE]}] \emph{Pooling effect consistency}: 
For any sequence $(L_\ell)_{\ell\in\N}$ in $\L_{[0,1]}$ that is conditionally iid, if $\ell\le \ell'$, and $K\in [0,1)$, then
\begin{equation} \label{eq:pe-ineq}
\I \left(\frac{(L^{(\ell)}-K)_+}{1-K}\right)\geq \I\left(\frac{(L^{(\ell' )}-K)_+}{1-K}\right).
\end{equation}

\end{enumerate}

Three points are worth discussing. First, for an MPH risk measure, positive homogeneity and translation invariance imply that \eqref{eq:pe-ineq} for sequences in $\L_{[0,1]}$ also holds for sequences in $\L^\infty$.
Second, [PE] is introduced and briefly discussed in \citet[Appendix D]{GKWW25} with a focus on the PD and EL criteria, where it is shown that the EL criterion satisfies [PE] but the PD criterion does not.
Third, clearly, [PE] implies [PM] when $A=[0,1]$. Later, Theorem \ref{th:PE-LI} shows that, on this domain, [PM] and [PE] are in fact equivalent for distortion risk measures. Outside the domain $A=[0,1]$, however, [PE] is not formulated, whereas [PM] remains meaningful.

Next, we turn to the context of multiple scenarios. 
Recall that  $S=(S_1,\dots,S_s)\in \mathcal F^s $ is a collection of  {economic scenarios}   of   relevance. 
We consider pooling effect in the context of scenario-based rating criteria, which is  formulated for sequences that are conditionally iid under each $\p_j$, $j\in [s]$. 
For instance, if $(L_\ell)_{\ell\in\N}$ is conditionally iid on a variable $Z$
and $S_j=\{Z\in A_j\}$ for  each $j\in [s]$, where $(A_j)_{j\in [s]}$ is a partition of $\R$, then  $(L_\ell)_{\ell\in\N}$ is conditionally iid  under each $\p_j$.

\begin{enumerate}[{[$S$-PM]}]
\item[{[$S$-PM]}] \emph{$S$-pooling monotonicity}: For any  sequence  $(L_\ell)_{\ell\in\N}$   in $\L_{A}$ that is conditionally iid under each $\p_j$, $j\in [s]$, if $\ell\le \ell'$ then
$$\I \left(L^{(\ell)}\right)\geq \I\left(L^{(\ell' )}\right).$$
\end{enumerate}

In contrast to [$S$-PM], [PM] is formulated to hold for all conditionally iid sequences, which are not necessarily conditionally iid under each $\p_j$, $j\in [s]$.  
The property [$S$-PE] is defined analogously. 
 The credit risk model used by \cite{moodys2023cso} assumes conditionally iid on Gaussian risk factors 
under each scenario, and thus [$S$-PE] implies the rating gets better for a larger pool for this model.

 \subsection{Risk aversion consistency}

A common way of formulating desirable ordering among random losses is through notions of risk aversion. For this, we first introduce convex order and increasing convex order.

\begin{definition}
For two random variables $X,Y\in \L^\infty$, we say that $X$ is smaller than $Y$ in \emph{(increasing) convex order}, denoted by $X\leq_{\rm{cx}} Y$ ($X\leq_{\rm{icx}} Y$), if $\mathbb{E}[\phi(X)]\leq \mathbb{E}[\phi(Y)]$ holds for any (increasing) convex function $\phi$. 
\end{definition}

\cite{RS70} used convex order   to define the concept of risk aversion for decision makers; if $X\le _{\rm cx} Y$, then $Y$ carries more risk than $X$.  
Increasing convex order
combines both preferences for certainty and preferences for more sure gain.   In light of this, we introduce a property [RA]. 

\begin{enumerate}[{[RA]}]
	\item[{[RA]}] \emph{Risk aversion consistency}: $\I(L_1)\leq \I(L_2)$ for $L_1, L_2\in \L_A$ with $L_1\leq_{\rm{icx}}L_2$.
\end{enumerate}

For rating criteria, property [RA] means that if a loss is agreed to be riskier than another loss
by all risk-averse expected utility investors, then the riskier loss cannot be rated better than the other one. 
Property [RA] is characterized for various decision models and risk measures; see \cite{SZ08} and \cite{BM06}.

To compare [RA] with [PM] and [PE],
we can show that [RA] implies [PM], and when $A=[0,1]$, [RA] implies [PE]. These are consequences of the following lemma.
This implication is formally summarized in Theorem \ref{thm:summarize} in Section \ref{sec:relations}.

\begin{lemma}
\label{lem:RC-pool}
For any conditionally iid sequence
$L_1,L_2,\dots\in\L_A$ and $\ell\leq \ell'$,
\begin{equation*}
L^{(\ell)}
\geq_{\rm icx}
L^{(\ell')}.
\end{equation*}
Moreover, when $A=[0,1]$, for all $K\in[0,1)$,
\begin{equation*}
\frac{(L^{(\ell)}-K)_+}{1-K}
\geq_{\rm icx}
\frac{(L^{(\ell')}-K)_+}{1-K}.
\end{equation*}
\end{lemma}

Similarly to [RA], we now introduce a scenario-based version of increasing convex order.

\begin{definition}
Given a collection of scenarios $S=(S_1,\dots,S_s)$, we say that $X$ is smaller than $Y$ in the \emph{$S$-increasing convex order}, denoted by $X\leq_{S\opName{-icx}} Y$, if $\mathbb{E}^{\p_j} [\phi(X)]\leq \mathbb{E}^{\p_j} [\phi(Y)]$ holds for any increasing convex function $\phi$ and $j\in [s]$.
\end{definition}

The $S$-increasing convex order can be formulated for a  general set  $\mathcal Q$ of probability measures:
We write 
 $X\leq_{\mathcal Q\opName{-icx}} Y$, if $\mathbb{E}^{Q} [\phi(X)]\leq \mathbb{E}^{Q} [\phi(Y)]$ holds for any increasing convex function $\phi$ and $Q\in \mathcal Q$. 
 The $S$-increasing convex order corresponds to the case $\mathcal Q=\{\p_j:j\in [s]\}$. 
The relation $X\leq_{\mathcal Q\opName{-icx}} Y$ 
has 
a clear economic interpretation. Every (expected utility) decision maker with an  increasing concave utility function $u$ and a subjective probability $Q $   in the convex hull of $\mathcal Q$ prefers loss $X$ over loss $Y$. This is because 
$$
X\leq_{\mathcal Q\opName{-icx}} Y \iff \E^Q[u(w-X)]\ge \E^Q[u(w-Y)], ~\forall u, Q \mbox{ above}.
$$
Therefore, $X\leq_{\mathcal Q\opName{-icx}} Y$ is a natural generalization of the usual stochastic dominance, as a robust notion of preferences under multiple subjective probabilities. 
A similar notion to $X\leq_{\mathcal Q\opName{-icx}} Y$ was introduced and  studied by \cite{DR10} in  the context of stochastic optimization.  

Although the stochastic order relation can be formulated for a general set $\mathcal Q$, we focus on the case of $\mathcal Q=\{\p_j:j\in [s]\}$ because of its relevance in credit models and technical tractability. 
The set $\{\p_j:j\in [s]\}$ 
contains mutually singular probability measures, allowing for representation results that are not available in the general case, as discussed by \cite{WZ21}.

The natural analogue of [RA] in our setting is the consistency with respect to $S$-increasing convex order, which we define below.
\begin{enumerate}[{[$S$-RA]}]
\item[{[$S$-RA]}] \emph{$S$-risk aversion consistency}: 
$\I(L_1)\leq \I(L_2)$ for $L_1, L_2\in \L_A$ with $L_1\leq_{S\opName{-icx}}L_2$.
\end{enumerate}

The notion of $S$-risk aversion consistency reflects the principle that if a loss is considered less risky (in the sense of $\le_{\rm icx}$) in all economic scenarios, then it should be seen as less risky overall and given a better rating.  
In different forms, $S$-risk aversion consistency 
is recently considered by \cite{zang2024random} in the context of  random distortion functions, and by  \cite{BMWW24}  for modeling preferences in decision theory. 

As far as we know, there are no existing  characterization results of  [$S$-RA] for risk measures, and we will offer one for Choquet risk measures in Section \ref{sec:charac-2}.

\subsection{Convexity and quasi-convexity}
\label{sec:convex}

In the  literature of risk measures, 
convexity  (defined in Section \ref{sec:21-def-MPH}) is
  a popular property reflecting consistency with respect to risk pooling and diversification,   characterized by \cite{FR02} and \cite{FS02}. 
For a rating criterion $\I$, which is categorical, convexity is not suitable. For instance, $x\mapsto\I(x)$ on $[0,1]$ cannot be convex if $\I$ takes at least $3$ values on $[0,1]$. 
For such a functional, 
quasi-convexity is a natural property reflecting the benefit of  diversification,  as discussed by 
\cite{CMMM11}, formally defined below. 
\begin{enumerate}[{[QC]}]
\item[{[QC]}] \emph{Quasi-convexity}: $\I(\lambda X+(1-\lambda) Y)\leq \max\{ \I(X), \I(Y)\}$ for $X,Y\in \L_A$ and $\lambda \in [0,1]$.
\end{enumerate}

Quasi-convexity has an intuitive meaning that combining a less risky asset A and a more risky asset B
does not get more risky than B. 
This makes the property appealing for the measurement of market risk, credit risk, and diversification. For instance, \cite{HLW25} used quasi-convexity as a main axiom for    quantifying diversification. 
For rating criteria, this property means that combining a good bond and a bad bond would not get a  rating worse than the bad bond.

For  Choquet risk measures, coherence, convexity, and quasi-convexity are equivalent; see \citet[Proposition 2.1]{CMMM11}. 
For distortion risk measures with a distortion function $h$, these properties are further equivalent to [RA] and concavity of $h$;  see \citet[Theorem 3]{WWW20} and Lemma \ref{lem:RC-trivial} below. 
Therefore, we do not need to distinguish [QC] and [RA] for distortion risk measures, but later we will see that [QC] and [$S$-RA] generally are not equivalent for $S$-distortion risk measures (Table \ref{tab:jump-exs} below). 

Quasi-convexity is typically studied together with some notion of continuity or semi-continuity. 
We also introduce two notions of semi-continuity that are useful for our paper.
\begin{enumerate}
	\item[{[LS] }] \emph{Lower semi-continuity under almost sure convergence}: $\liminf_{\ell\rightarrow \infty}\I(X_{\ell})\geq \I(X)$ for any 
    $(X_\ell)_{\ell\in\N} $ with 
    $X_{\ell} \rightarrow X$ almost surely.
    \item[{[LS*]}] \emph{Lower semi-continuity  with respect to $\L^\infty$}: $\liminf_{\ell\rightarrow \infty}\I(X_{\ell})\geq \I(X)$ for any 
    $(X_\ell)_{\ell\in\N} $ with 
    $X_{\ell} \rightarrow X$ in $\L^\infty$.
\end{enumerate}

Both [LS] and [LS*] are widely studied in the literature. 
Note that [LS] is stronger than [LS*] because $\L^\infty$ convergence is stronger than almost sure convergence. 
Property 
[LS*] is satisfied by all monetary risk measures, which are indeed continuous in $\L^\infty$. 
Property [LS] for risk measures  is  introduced by \cite{WZ21b} under the name of prudence,\footnote{  \cite{WZ21b} used point-wise convergence instead of  almost sure convergence, but these two continuity properties are equivalent for law-invariant risk measures. For some subtle differences, see \cite{AL24}.} which reflects a statistical consideration of robustness, and it is further studied by \cite{AL24} in detail. In the context of rating practices, it reflects the idea that if
the loss $L$ is statistically modeled using an approximation that converges to $L$, then the rating of the approximation should not be better than the rating of $L$ when the approximation error tends to zero.

\subsection{Relations among the risk consistency properties}
\label{sec:relations}

The properties [QC], [PM], [PE],  [RA], 
[$S$-PM], [$S$-PE] and [$S$-RA] will all be referred to as \emph{risk consistency properties}.  These properties 
reflect ideas that are interconnected. 
We next study relations among them. 

We first present a few examples to show that these conceptually similar notions are not equivalent, even within the class of   $S$-distortion risk measures with $s=2$ scenarios. 
    In what follows,  the essential supremum and infimum of a random variable $X$ under a probability measure $\mathbb Q$ are defined by, respectively,  
\begin{align*} 
\esssup^{\mathbb Q}(X) &=\inf\{x\in \mathbb{R} : \mathbb Q(X>x)=0\};\\ 
\essinf^{\mathbb Q}(X)&=\sup\{x\in \mathbb{R}: \mathbb Q(X< x)=0\}.\end{align*}

\begin{example}
\label{ex:jumps}
    
 Four different $S$-distortion risk measures with $s=2$, each with a discontinuous $S$-distortion function $g$, are presented in Table \ref{tab:jump-exs}. The assertions therein are straightforward to verify. We can see that they may satisfy some of [QC], [$S$-RA], [$S$-PM] and [$S$-PE]. 
\end{example}

\begin{table}[t]

\centering
\small
\caption{Examples of risk measures satisfying  risk consistency properties\label{tab:jump-exs}}
\resizebox{\textwidth}{!}{%
\begin{tabular}{@{}c@{\quad}|@{\quad}c@{\quad}|@{\quad}c@{\quad}|@{\quad}c@{\quad}|@{\quad}c@{\quad}|@{\quad}c@{\quad}|@{\quad}c@{}}
\hline 
$\rho$ & $g(x_1,x_2)$ &  continuity of $g$ & [QC] & [$S$-RA] & [$S$-PM]  & [$S$-PE]           \\ \hline 
 $ \esssup^{\mathbb{P}_1} \vee \esssup^{\mathbb{P}_2} $ & $\id_{\{x_1\vee x_2>0\}}$ & lower & $\checkmark$ & $\checkmark$ & $\checkmark$ &  $\checkmark$ 
 \\
        $\esssup^{\mathbb{P}_1} \wedge\esssup^{\mathbb{P}_2} $ & $\id_{\{x_1\wedge x_2>0\}}$ & lower & $\times$ & $\checkmark$ & $\checkmark$ &  $\checkmark$ 
        \\
        $\essinf^{\mathbb{P}_1} \vee\essinf^{\mathbb{P}_2} $ & $\id_{\{x_1\vee x_2=1\}}$ & upper & $\times$ & $\times$ & $\checkmark$ &  $\checkmark$ 
        \\
        $\essinf^{\mathbb{P}_1} \wedge \essinf^{\mathbb{P}_2} $ & $\id_{\{x_1\wedge x_2=1\}}$ & upper & $\times$ & $\times$ & $\checkmark$ &  $\checkmark$ 
 
 \\ \hline
\end{tabular}
}

\end{table}

In view of Example \ref{ex:jumps}, the risk consistency properties are not equivalent, but they have some 
connections. 
We summarize the implications among these properties in the next theorem, which apply to both MPH risk measures
$\rho:\mathcal L^\infty\to\mathbb R$ and SEU risk measures
$\rho:\mathcal L_A\to\overline{\mathbb R}$,
as well as to rating criteria
$\I:\mathcal L_A\to[n]$. 
Recall that monotonicity of both $\rho$ and $\I$ is assumed throughout, and it is used in the proof of parts (iv) and (v) of Theorem \ref{thm:summarize}.

\begin{theorem}
\label{thm:summarize}
Fix a collection of scenarios $S$.  
\begin{enumerate}[(i)]
\item The following implications hold:  $$
\mbox{[LI]}\implies   \mbox{[$S$-LI]};~~
\mbox{[RA]}\implies   \mbox{[$S$-RA]};~~
\mbox{[PE]}\implies   \mbox{[$S$-PE]};~~
\mbox{[PM]}\implies   \mbox{[$S$-PM]}.
 $$

\item {[RA]} implies [LI] and [PM]. When $A=[0,1]$, [RA] also implies [PE].
\item  {[$S$-RA]} implies [$S$-LI] and [$S$-PM]. When $A=[0,1]$, [$S$-RA] also implies [$S$-PE].

\item {[QC]}, [LS*] and [LI] together imply [RA].
\item {[QC]}, [LS*] and [$S$-LI] together imply [$S$-RA].
\item {[QC]}  and [LI] together imply [PM].
\item {[QC]}  and [$S$-LI] together imply [$S$-PM].
\end{enumerate}

\end{theorem}

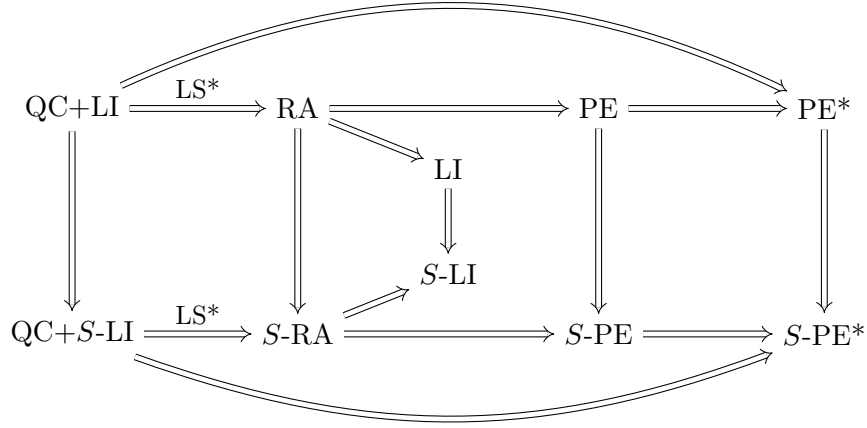
\begin{figure}[t]
\centering
\caption{Summary of the relations among the properties in Theorem \ref{thm:summarize}; note that [LS*] is automatic for monetary risk measures, and the arrows involving [PE] and [$S$-PE] are understood only on $\L_{[0,1]}$.}
\label{fig:summary}  
\begin{tikzpicture}[
    node distance = 0.8cm and 1cm,
    node/.style={inner sep=6pt, align=center}, 
    lrarrow/.style={<->, double, double distance=2pt, >=Implies},
    larrow/.style={<-, double, double distance=1.8pt, >=Implies},
    rarrow/.style={->, double, double distance=1.8pt, >=Implies}
]
   
    \node (QC1) at (-1,1.5) {QC+LI};
     \node (QC2) at (-1,-1.5) {QC+$S$-LI};
    \node (RC) at (2,1.5) {RA};
    \node (SRC) at (2,-1.5) {$S$-RA};
    \node (LI) at (4,0.7) {LI};  
    \node (SLI) at (4,-0.7) {$S$-LI}; 
    \node (PE) at (6, 1.5) {PE};
    \node (SPE) at (6, -1.5) {$S$-PE};
    \node (PEs) at (9, 1.5) {PM};
    \node (SPEs) at (9, -1.5) {$S$-PM};

    \draw[rarrow] (QC1) -- (RC) node[midway,sloped, above] {\small   LS*};
    \draw[rarrow] (QC2) -- (SRC) node[midway,sloped, above] {\small  LS*};    \draw[rarrow] (QC1) -- (QC2);
    \draw[rarrow] (RC) -- (SRC);
    \draw[rarrow] (RC) -- (LI);
    \draw[rarrow] (LI) -- (SLI);
    \draw[rarrow] (SRC) -- (SLI);
    \draw[rarrow] (RC) -- (PE);
    \draw[rarrow] (PE) -- (SPE);
    \draw[rarrow] (SRC) -- (SPE);
    \draw[rarrow] (PE) -- (PEs);
    \draw[rarrow] (SPE) -- (SPEs);
    \draw[rarrow] (PEs) -- (SPEs);
    \draw[rarrow, bend left=25] (QC1) to   (PEs) ;
    \draw[rarrow, bend right=20] (QC2) to     (SPEs) ;
\end{tikzpicture}

\end{figure}

We summarize in Figure \ref{fig:summary} the relations obtained in Theorem \ref{thm:summarize}.
Since $S$-distortion risk measures automatically satisfy [LS*] and [$S$-LI], they satisfy the chain of implications [QC]$\Rightarrow$[$S$-RA]$\Rightarrow$\allowbreak[$S$-PE]$\Rightarrow$[$S$-PM] when $A=[0,1]$, and [QC]$\Rightarrow$[$S$-RA]$\Rightarrow$[$S$-PM] when $A=(-\infty,1]$. 
Similarly, for distortion risk measures, we have the chain of implications [QC]$\Rightarrow$[RA]$\Rightarrow$[PE] $\Rightarrow$[PM], where [PE] appears only when $A=[0,1]$.
For these risk measures, [QC] is the strongest notion among risk consistency properties; this was also reflected in Table \ref{tab:jump-exs}.

\section{Connecting properties of risk measures and rating criteria}\label{sec:translation}

This section shows how we can freely translate most properties in Section \ref{sec_rc} 
 from a risk measure $\rho$  
 to the rating criteria $\I$ it represents, and back. 
 The forward direction is simple, as $\I$ is an increasing function of $\rho$, but the backward direction is more complicated.
 In particular, for a given $\I$, we first need to understand whether its representing risk measure is unique. The following theorem shows that this is the case within the MPH class, while within the SEU class the representing risk measure is unique up to a positive scalar.

\begin{theorem} \label{thm_uniqueness}
Suppose that a rating criterion $\I$ is represented by two risk measures $\rho_1$ and $\rho_2$.
\begin{enumerate}[(i)]
    \item If $\rho_1$ and $\rho_2$ are both MPH risk measures, then $\rho_1 = \rho_2$.
    \item  If $\rho_1$ and $\rho_2$ are both SEU risk measures, then $\rho_1=c\rho_2$ for some $c>0$.
\end{enumerate}
	
\end{theorem}

Part (i) of Theorem \ref{thm_uniqueness} implies that there is a bijection between MPH risk measures  and MPH rating criteria. 
This may appear counterintuitive at first, because
rating criteria are discrete transforms of risk measures, with a loss of information. To see an intuition behind  Theorem \ref{thm_uniqueness}, suppose that two  MPH risk measures   disagree on the ranking of two random losses $X$ and $Y$. We can find losses $X'$ and $Y'$ 
 using translation invariance and positive homogeneity of MPH risk measures such that the corresponding MPH rating criteria  disagree on the ranking of $X'$ and $Y'$, and thus they cannot be the same rating criterion.  Moreover, this result justifies that for a Choquet rating criterion $\I$ satisfying [LI] (resp.~[$S$-LI]), its  distortion function   in \eqref{eq_rho_LI} (resp.~$S$-distortion function in \eqref{eq_rho_SR}) is unique.
 
Part (ii) of Theorem \ref{thm_uniqueness} shows that an SEU rating criterion determines its representing SEU risk measure up to a positive scalar. This constant scalar will be fixed by imposing  {the normalization $\rho(1/2)=1/2$ (recall that $\rho(1)$ may be infinite as in the Morningstar example)}.  
Throughout the rest of the paper, $\rho_{\I}$ denotes the representing risk measure of $\I$, uniquely determined in the MPH case and normalized as above in the SEU case. The corresponding function $f$ in \eqref{eq:p2r} is denoted by $f_{\I}$.

\begin{remark}
Theorem \ref{thm_uniqueness} is a within-class uniqueness result. Cross-class uniqueness need not hold: the same rating criterion may admit both MPH and SEU representations with different representing risk measures; see the example in Appendix \ref{app:counterexampleForUniqueness}.
Also shown by this example, although the two classes of risk measures only have an expected value 
as their common element (Proposition \ref{prop_MPH_SEU_intersection}), the common rating criteria of the two classes are more than those generated by the expected value.  
\end{remark}

With the help of Theorem \ref{thm_uniqueness}, now we can discuss how the properties in Section \ref{sec_rc} can be translated between rating criteria and the corresponding risk measures.

\begin{theorem}\label{th:convert}
Let $\I$ be either an MPH or an SEU rating criterion, with
corresponding representing risk measure $\rho_{\I}$. The following statements hold.

\begin{enumerate}[(i)]
    \item For either class, $\I$  satisfies property $P$ if and only if $\rho_\I$ satisfies the same property,

    where $P$ is any one of [LI], [$S$-LI], [RA], [$S$-RA], [PM],
    [$S$-PM], and [QC].

    \item If $\I$ is an MPH rating criterion, then $\I$ satisfies property $P$ if and only if $f_\I$ is left-continuous on $A$ and $\rho_\I$ satisfies the same property on $\mathcal L_A$,

    where $P$ is either [LS] or [LS*].

        \item If $\I$ is an SEU rating criterion and $A=[0,1]$, then $\I$  satisfies property $P$ if and only if $\rho_\I$ satisfies the same property,

    where $P$ is either [PE] or [$S$-PE].
\end{enumerate}
\end{theorem}

The ``only if" parts of Theorem \ref{th:convert} are particularly interesting, as they derive properties of the finer mapping $\rho_\I$ from the coarser mapping $\I$.
The MPH assertions are proved directly. For SEU rating criteria, 
the [LI] and [$S$-LI] translation results follow from
Proposition~\ref{prop_SEU_LI_SLI}, whereas other assertions are proved jointly with the
characterization results in Theorem~\ref{thm:char-SEU}.

For MPH rating criteria in Theorem \ref{th:convert}, we excluded the two properties 
[PE] and [$S$-PE]. This is because the linear transformation method used in the proof arguments is not compatible with the truncated random variable $(L^{(\ell)}-K)_+/(1-K)$ unless $K=0$. 
For the Choquet subclass considered below, Theorem \ref{th:PE-LI} shows that [PE] and [PM] are equivalent.  A corresponding result for [$S$-PE] is given in Theorem \ref{th:SR-Choquet}, which requires a continuity condition.  

\begin{remark} Not all properties of $\rho_\I$ can be translated into those of $\I$ (and vice versa). For instance,   if $\rho_\I$ is convex, then $\I$ is quasi-convex (by Theorem \ref{th:convert}, as convexity is stronger than quasi-convexity), but not convex, because $x \mapsto \I(x)$ has jumps on $(0,1)$, as in rating practice.
\end{remark}

\section{Characterization results for Choquet rating criteria}
\label{sec:choquet_chara}

\subsection{Characterization results under law invariance}

\label{sec:chara-LI}

We now fully characterize the risk consistency properties for Choquet rating criteria under law invariance. 
We will first present the corresponding characterization
for distortion risk measures, 
and then use Theorem \ref{th:convert} to get the corresponding results for Choquet rating criteria. 
The equivalence between risk aversion consistency and quasi-convexity of distortion risk measures
is well-known in the literature (e.g., \citealp{DVGKTV06} and \citealp{WWW20}), and we summarize it below. 

\begin{lemma}
\label{lem:RC-trivial}
 For a distortion risk measure $\rho$ with distortion function $h$, the following are equivalent:
\begin{enumerate}[(i)]
\item   $h$ is concave;
\item   $ \rho(X)\leq \rho(Y)$ for  all  $X,Y\in \mathcal L_{[0,1]}$ with $X\leq_{\rm{icx}}Y$ (i.e., [RA]);
\item   $ \rho(X)\leq \rho(Y)$ for  all  $X,Y\in \mathcal L_{[0,1]}$ with $X\leq_{\rm{cx}}Y$;
\item $\rho$ is convex;
\item $\rho$ is quasi-convex (i.e., [QC]).
\end{enumerate}  
\end{lemma}

We next characterize pooling effect consistency for distortion risk measures.
\begin{theorem}
\label{th:PE-LI}
For a distortion risk measure $\rho$ with distortion function $h$, the following are equivalent:
\begin{enumerate}[(i)]
\item   $h$ is concave on $[0,1)$;  

\item $\rho$ satisfies [PE]; 
\item  $\rho$ satisfies [PM];  
\item   $\rho(L_1) \ge \rho(L^{(\ell)}) $   for all conditionally iid sequence $L_1,L_2,\dots\in \L_{[0,1]}$ and  all $\ell\geq 1$.  
\end{enumerate}  
\end{theorem}

Different from Lemma \ref{lem:RC-trivial}, the first condition in Theorem \ref{th:PE-LI} states that $h$ is concave on $[0,1)$, but it is not necessarily concave on $[0,1]$. 
This subtle difference is illustrated by the following example. 
\begin{example}
\label{ex:essinf}
An example of $h$ that is concave on $[0,1)$ but not on $[0,1]$ is $h(x)=\id_{\{x=1\}}$ for $x\in [0,1]$.  
The corresponding distortion risk measure $\rho$ is given by
$\rho(X)=\essinf(X)$,  the essential infimum of $X$ under $\p$. 
Since $L^{(\ell)}$ and $L^{(\ell')}$ for any $\ell,  \ell' \in \N $  have the same essential infimum (see Appendix \ref{appx:proof_sec5} for a proof),
items (ii)--(iv) in Theorem \ref{th:PE-LI}  hold. However, any item  in Lemma \ref{lem:RC-trivial} does not hold as $h$ is not concave on $[0,1]$.
\end{example}

The following corollary is immediate from Lemma \ref{lem:RC-trivial} and Theorem \ref{th:PE-LI}.

\begin{corollary}
\label{coro:equiv}
For a distortion risk measure  satisfying [LS], the four properties
[PM], [PE], [RA] and [QC] are equivalent. 
\end{corollary}

Corollary \ref{coro:equiv} suggests that, for distortion risk measures satisfying [LS],  all risk consistency properties under law invariance are indeed equivalent. 
This is not necessarily true for other classes of risk measures.\footnote{ \cite{MW20} characterized  monetary risk measures that satisfy [RA], and they do not necessarily satisfy [QC]. An example of [RA]$\not\Rightarrow$[QC] is given by the minimum $\min\{\rho_1,\rho_2\}$ of two distortion risk measures $\rho_1,\rho_2$ that are convex (hence $\rho_1,\rho_2$ satisfy both [RA] and [QC]). The  minima of normalized convex risk measures are characterized by \cite{CCMTW22} as  star-shaped risk measures, and they are not necessarily quasi-convex.}

By combining Theorem \ref{th:convert}, Lemma \ref{lem:RC-trivial}, and Corollary \ref{coro:equiv}, we obtain a characterization of risk consistency properties for law-invariant Choquet rating criteria satisfying [LS].

\begin{proposition}\label{prop:LI-rating}
	For a Choquet rating criterion $\I$ that satisfies [LS] and [LI], the following are equivalent:
	\begin{enumerate}[(i)]
	\item $\I$ satisfies [QC]; 
	\item $\I$ satisfies [RA]; 
    \item $I$ satisfies [PM];
	\item   $\I$ satisfies  [PE]; 
\item $\rho_\I$ has a concave distortion function. 
\end{enumerate}	
\end{proposition}

Proposition \ref{prop:LI-rating}
has a useful interpretation. 
Assuming [LS] and [LI],
if the Choquet rating criterion agrees with a risk ranking as soon as
all risk-averse  expected utility investors agree with it, 
then this rating criterion must assign a better rating to the senior tranche backed by a larger pool than to the senior tranche backed by a smaller pool, assuming the same quality of underlying assets and attachment point. The reverse direction is also true.  

The standing assumption that $f_{\I}$ in \eqref{eq:p2r} 
has an effective cut point
is needed for the implication (iii)$\Rightarrow$(iv) in Proposition \ref{prop:LI-rating}.  
Property [LS] is also needed for this implication, otherwise  Example \ref{ex:essinf} yields a counter-example.

\subsection{Characterization results under multiple scenarios} \label{sec:charac-2}

Now we proceed to consider Choquet rating criteria and Choquet risk measures under $S$-law invariance, and analyze conditions on the $S$-distortion function for them to satisfy the risk consistency properties. 
We follow a route that is similar to Section \ref{sec:chara-LI}, by first addressing Choquet risk measures, and then discussing the corresponding rating criteria.

\subsubsection{Scenario-based  Choquet risk measures}
Recall from Proposition \ref{prop:LI_SR} that  an $S$-law-invariant Choquet risk measure $\rho$ 
has the form
\begin{equation}
    \label{eq:S-based}
    \rho(X)= \int X\dd (g \circ \mathbb{P}^{S}),~~X\in \L^\infty,
\end{equation}
for an $S$-distortion function $g$ in \eqref{eq_rho_SR}.  
A function $g:\mathbb{R}^s\rightarrow\mathbb{R}$ is called \emph{submodular} if 
\begin{equation}
    g(\mathbf x\wedge \mathbf y) +g(\mathbf x\vee \mathbf y) \le g(\mathbf x ) + g(\mathbf y) \mbox{~~~for all $\mathbf x,\mathbf y\in \R^s$},\label{eq:submodular1}
\end{equation}
where $\vee$ is the componentwise maximum and $\wedge $ is the componentwise minimum.

For a characterization of [$S$-RA] and  [$S$-PE],  it turns out that $g$ should satisfy a suitable continuity condition.
The following notion of continuity is specifically designed for risk measures, as a rating criterion $\I$ always has jumps. 
 
\begin{enumerate}
	\item[{[LC]}] \emph{Lebesgue continuity}: $\lim_{\ell\rightarrow \infty}\rho(X_{\ell})=\rho(X)$ whenever 
    $(X_\ell)_{\ell\in\N} $ is a bounded sequence in $\L^{\infty}$ that converges almost surely to $X$.
\end{enumerate}
Property [LC] is called the Lebesgue property in \cite{FS16}, and is satisfied by many commonly used risk measures such as the expectation. One can verify that [LC] is stronger than [LS*], but it is not comparable with [LS]. For instance, the expectation satisfies [LC] but fails to satisfy [LS] on $L^\infty$, whereas the essential supremum satisfies [LS] but fails [LC].
For $S$-distortion risk measures, [LC] is equivalent to the continuity of $S$-distortion function, as shown in the next result.  
\begin{lemma}
    \label{lemma:LC}
An $S$-distortion risk measure with $S$-distortion function $g$ satisfies [LC] if and only if $g$ is continuous.
\end{lemma} 

With the help of the above lemma and equivalent conditions for submodular functions (Lemma \ref{lemma:submodular} in Appendix \ref{app:L5}), we obtain in the next result a characterization of $S$-distortion risk measures that satisfy one of the risk consistency properties in Section \ref{sec_rc}.
 
\begin{theorem}
    \label{th:SR-Choquet}
    For an $S$-distortion risk measure  $\rho$ with $S$-distortion function $g$, 
 the following conditions (i)--(iii) are equivalent.
    \begin{enumerate}[(i)]
        \item $\rho$ is coherent;
        \item $\rho$ is quasi-convex (i.e. [QC]);
        
        \item  $g$  is componentwise concave and submodular on $[0,1]^s$. 
    \end{enumerate}
          Conditions (i)--(iii) imply
    \begin{enumerate}[(i)]
                
        \item[(iv)] $\rho$ satisfies [$S$-RA];
        \item[(v)] $\rho$ satisfies [$S$-PE];
        \item[(vi)] $\rho$ satisfies [$S$-PM]. 
    \end{enumerate}
If [LC] holds for $\rho$, then all of (i)--(vi) are equivalent. 
\end{theorem}

 For an $S$-distortion risk measure $\rho$, 
property [LC] (equivalent to the continuity of $g$) is not dispensable for the equivalence of (i)--(vi) in Theorem \ref{th:SR-Choquet}. 
In the proof,  both upper semi-continuity and lower semi-continuity of $g$ are used. 
Without the continuity assumption,  we have the implications [QC]$\Rightarrow$[$S$-RA]$\Rightarrow$[$S$-PE], but the converse directions do not hold, as demonstrated by the examples in Table \ref{tab:jump-exs}.
We do not know whether [$S$-PE]$\Leftrightarrow$[$S$-PM] for $S$-distortion risk measures.

In case $s=1$, Theorem \ref{th:SR-Choquet} 
gives the classic result reported in  Lemma \ref{lem:RC-trivial} that for 
a distortion risk measure, risk aversion consistency is equivalent to coherence. 
In this case, continuity of the distortion function is not needed for the equivalence result, but the situation becomes subtle and complicated as soon as $s\ge 2$, as discussed above.

Building on the results obtained so far, in particular Theorems \ref{th:PE-LI}--\ref{th:SR-Choquet}, we identify an intrinsic parallel between the properties of distortion risk measures and those of $S$-distortion risk measures,  summarized in Figure \ref{fig:summary measure}.

\begin{figure}[t]
\centering
\caption{Relations between risk consistency  properties for distortion and $S$-distortion risk measures, where  
$\mathcal D_{\rm cv}$ is the set of concave distortion functions,
and $\mathcal D_{\rm ccs}$ is the set of componentwise concave and submodular $S$-distortion functions}
\label{fig:summary measure}  
\begin{tikzpicture}[
    node distance = 0.8cm and 1cm,
    node/.style={inner sep=6pt, align=center}, 
    bidir/.style={<->, double, double distance=2pt, >=Implies},
    arrow/.style={->, double, double distance=2pt, >=Implies}
]

\begin{scope}[local bounding box=left]
    \node[node] (H) {$h\in \mathcal D_{\rm cv}$};
    \node[node, right=of H] (QC) {~QC};
    \node[node, right=of QC] (RC) {RA};
    
    \node[node, below=of QC, yshift=-0.35cm] (PES) {PM};
    \node[node, below=of RC, yshift=-0.4cm] (PE) {PE};
    
    \draw[bidir] (H) -- (QC);
    \draw[bidir] (QC) -- (RC);
    \draw[bidir] (PES) -- (PE);
    \draw[arrow] (PES) -- node[midway, left] {LS} (QC);
    \draw[arrow] (RC.south) -- (PE.north); 
\end{scope}

\begin{scope}[local bounding box=right, xshift=7cm] 
    \node[node] (H2) {$g\in \mathcal D_{\rm ccs}$};
    \node[node, right=of H2] (QC2) {~QC~};
    \node[node, right=of QC2] (SRC) {$S$-RA};
    
    \node[node, below=of QC2, yshift=-0.35cm] (SPES) {$S$-PM};
    \node[node, below=of SRC, yshift=-0.4cm] (SPE) {$S$-PE};
    
    \draw[bidir] (H2) -- (QC2);
    \draw[arrow] (QC2) -- (SRC);
    \draw[arrow] (SPE) -- (SPES);
    \draw[arrow] (SPES) -- node[midway, left] {LC} (QC2);
    \draw[arrow] (SRC.south) -- (SPE.north); 
\end{scope}

\node[below=5mm of left.south, anchor=north] {(a) Distortion risk measures};
\node[below=5mm of right.south, anchor=north] {(b) $S$-distortion risk measures};
\end{tikzpicture}
\end{figure}
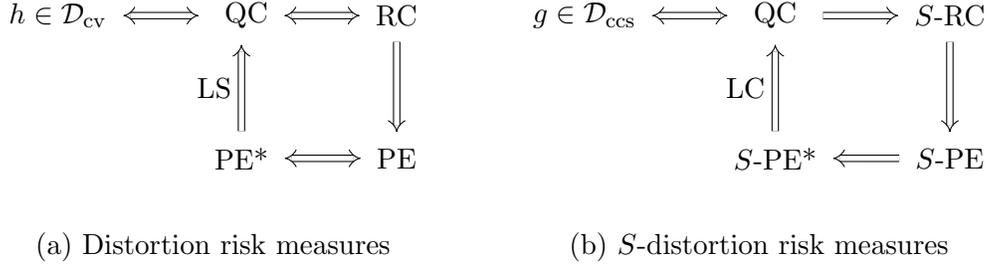 

\subsubsection{Scenario-based Choquet rating criteria}
By using 
Theorem \ref{th:convert}, Theorem \ref{th:SR-Choquet}, and Lemma \ref{lemma:LC}, we obtain the following characterization of 
Choquet rating criteria that satisfy risk consistency properties, which  
can be seen as a parallel result to Proposition~\ref{prop:LI-rating}.

\begin{proposition}\label{prop:SR-rating}
    For a Choquet rating criterion $\I$ satisfying   [$S$-LI] with $\rho_\I$ satisfying [LC], the following are equivalent:
    \begin{enumerate}[(i)]
        \item $\I$ satisfies [QC]; 
        \item $\I$ satisfies [$S$-RA];
        \item $I$ satisfies [$S$-PM];
        \item $\I$ satisfies [$S$-PE];
        \item $\rho_\I$ has a componentwise concave and submodular $S$-distortion function.
    \end{enumerate}
\end{proposition}

Proposition \ref{prop:SR-rating} allows us to construct many Choquet rating criteria $\I$ that satisfy [QC], [$S$-RA], and [$S$-PE], through choosing the $S$-distortion function $g$. 
The simplest choice is to use the 
additively separable function $g$ given by  
\begin{align}
\label{eq:separable}
    g(x_1,\dots,x_s) = 
    \sum_{j=1}^s a_j h_j(x_j),~~~(x_1,\dots,x_s)\in [0,1]^s, 
\end{align}
 where  for each $j\in [s]$, $a_j\ge 0$ and
$h_j:[0,1]\to \R$ is a distortion function, and $\sum_{j=1}^s a_j  =1$. 
The corresponding $S$-distortion risk measure $\rho_\I$ is, when restricted to $\L_{[0,1]}$,
\begin{equation}\label{eq:decom3}
\rho_\I(X)= \sum_{j=1}^s a_j \int_0^1 h_j(\p_j(X>x))\dd x,~~~~~X\in \L_{[0,1]}. 
\end{equation}

\begin{proposition} \label{prop_RCconcave}
If  
$h_j$  in \eqref{eq:decom3} is concave for each $j\in [s]$, then the Choquet rating criterion $\I$ satisfies [QC], [$S$-RA], [$S$-PM], and [$S$-PE].    
\end{proposition}
The proposition directly follows from Proposition \ref{prop:SR-rating}. 

Unlike the characterization of [PE] for distortion risk measures in Theorem \ref{th:PE-LI}, 
a full characterization of [$S$-PE]  without assuming a continuous $S$-distortion function seems to be out of reach with the current technique. 
Instead, we provide one necessary condition for [$S$-PE] in Appendix \ref{appx:necessaryPE}.

Finally, some examples of scenario-based risk measures are collected in Table \ref{tab:examplesSR}, and they represent different rating criteria. Some of these examples appear in \cite{GKWW25}, and our focus here is on their risk consistency properties. 
In Table \ref{tab:examplesSR},  three useful examples of Choquet rating criteria---Average EL, Average ES, and Average MAXVAR\footnote{The name MAXVAR comes from \cite{cherny2009new}, which is the special case of our Average MAXVAR when $s=1$.}--- satisfy [QC], [$S$-RA], and [$S$-PE] by Proposition \ref{prop_RCconcave}. They belong to the class \eqref{eq:decom3} with  additively separable $S$-distortion functions.
Two other examples of Choquet risk measures---Average VaR and Max VaR---do not satisfy [$S$-PE] by Proposition \ref{prop:SPE2} and, consequently, do not satisfy [$S$-RA]. 
For a comparison in the next section, we also include the risk measure $\rho_\I$ that generates the scenario-based   PD criterion, called the Average PD, which is outside the Choquet class. Among the six examples, each of Average MAXVAR, Average ES, Average VaR and Max VaR forms a one-parameter family.

\begin{table}[t]

\centering
\small
\caption{Examples of scenario-based risk measures, where $\gamma \in (0,1]$, $p\in [0,1)$, $a_1,\dots,a_s$ are nonnegative weights that sum to $1$,   and $\mathrm{VaR}_q(X|S_j)$ denotes the 
conditional $q$-quantile function of $X$ given $S_j$ for $j\in[s]$ and $q\in [0,1)$\label{tab:examplesSR}}
\resizebox{\textwidth}{!}{%
\begin{tabular}{@{}c@{\quad}|@{\quad}c@{\quad}|@{\quad}c@{\quad}|@{\quad}c@{\quad}|@{\quad}c@{}}
\hline 
 Name & $\rho_\I(X)$ & $ g(x_1,\dots,x_s)$   & [$S$-RA] & [$S$-PE] \\
    \hline
    Average EL & $\sum_{j=1}^s a_j\E[X|S_j]$ & $\sum_{j=1}^s a_j x_j$ & $\checkmark$ & $\checkmark$ \\
    
    Average ES &  $\sum_{j=1}^s \frac{a_j}{1-p}\int_p^1 \mathrm{VaR}_q(X|S_j)\dd q$ & $\sum_{j=1}^s \frac{a_j}{1-p}  (x_j\wedge (1-p))$ & $\checkmark$  & $\checkmark$\\
    Average MAXVAR & $\sum_{j=1}^s a_j\int_0^1 (\p(X>x|S_j))^\gamma \dd x$ & $ \sum_{j=1}^s a_j  x_j^\gamma$ & $\checkmark$ & $\checkmark$ \\
    
    \hline
    Average VaR & $\sum_{j=1}^s a_j \mathrm{VaR}_p(X|S_j)$ & $\sum_{j=1}^s a_j  \id_{\{x_j>1-p\}}$ & $\times$ & $\times$ \\
    Max  VaR & $\bigvee_{j=1}^s \VaR_p(X|S_j)$ & $\bigvee_{j=1}^s  \id_{\{x_j>1-p\}}$ & $\times$ & $\times$ \\
    Average PD & $\sum_{j=1}^s a_j \p(X>0|S_j)$ &  not Choquet    & $\times$ & $\times$
 
 \\ \hline
\end{tabular}
}

\end{table}

\section{Characterization results for SEU rating criteria}
\label{sec:SEU_chara}

We now fully characterize the risk consistency properties for SEU rating criteria. 
The SEU translation results for risk-consistency properties stated in Theorem~\ref{th:convert} are closely intertwined with the most advanced characterization results below, and are therefore established jointly in the proof of the following theorem.

\begin{theorem}\label{thm:char-SEU}
Let $\I$ be an SEU rating criterion with SEU risk measure $\rho_{\I}$.  
\begin{enumerate}[(i)] 
    \item $\I$ satisfies [QC] if and only if $\rho_{\I}$ admits the form \eqref{eq_SEU_expectation}, where each $v_\omega$ is a convex loss function.
    
    \item $\I$ satisfies [$S$-RA] if and only if $\rho_{\I}$ admits the form \eqref{eq_SEU_SLI}, where each $v_j$ is a convex loss function.
    
    \item $\I$ satisfies [$S$-PM] if and only if, for some convex loss functions $v_j$, $j\in [s]$, and a nonnegative random variable $Z$,
    \begin{equation}
        \rho_{\I}(X)
        =
        \sum_{j=1}^s \mathbb E[v_j(X)\id_{S_j}]
        +
        \mathbb{E}[Z\id_{\{X=1\}}],
        \quad X\in\mathcal L_A.
        \label{eq_rho_SPE}
    \end{equation}
    
    \item When $A=[0,1]$, $\I$ satisfies [$S$-PE] if and only if any equivalent condition in (iii) holds.
\end{enumerate}
\end{theorem}

Under expected utility theory,  it is known that risk aversion holds if and only if the utility function is concave \citep{RS70}, and therefore the convexity of the loss functions in Theorem~\ref{thm:char-SEU} 
is anticipated under [LI]. However, 
the risk measures $\rho_\I$ in parts (i), (iii) and (iv) of Theorem~\ref{thm:char-SEU} do not satisfy [LI] or [$S$-LI], and a more delicate analysis is needed. For instance, $\rho_\I(X)$ in  
\eqref{eq_rho_SPE}
depends on the joint distribution of $(X,Z)$ instead of the distributions of $X$ under $\p^S$ alone.

By taking $S=\Omega$, we immediately obtain the following implications from Theorem~\ref{thm:char-SEU}. 
\begin{enumerate}[(a)]
   \item $\I$ satisfies [RA] if and only if $\rho_{\I}$ admits the form \eqref{eq_SEU_LI} for some convex loss function $v$.
    \item $\I$ satisfies [PM] if and only if for     some  convex loss function $v$ and nonnegative random variable $Z$, 
    \begin{equation}
        \rho_{\I}(X)=\mathbb E[v(X)+Z\id_{\{X=1\}}],\quad X\in\mathcal L_A.\label{eq_rho_PE}
    \end{equation}

    \item When $A=[0,1]$, $\I$ satisfies [PE] if and only if any condition in (b) holds.

\end{enumerate}

The risk measures in (i)--(iii) of Theorem~\ref{thm:char-SEU} are similar, but still slightly different. 
Under a continuity assumption and law invariance, they can be unified.

\begin{theorem}\label{thm_SR_SEU}
For an SEU risk measure $\rho$ satisfying [LC] and [$S$-LI], the following statements are equivalent.
\begin{enumerate}[(i)]
\item  $\rho$ satisfies [QC];
\item  $\rho$ satisfies [$S$-RA];
\item $\rho$ satisfies [$S$-PM];
\item $\rho$ satisfies [$S$-PE] (in the case $A=[0,1]$);
\item $\rho$ admits the form \eqref{eq_SEU_SLI} where each $v_j$ is a finite convex loss function.
\end{enumerate}
\end{theorem}

Again, by taking $S=\Omega$,
Theorem~\ref{thm_SR_SEU} implies the equivalence between the following properties under [LC] and [LI]:
[QC], [RA], [PE] (when $A=[0,1]$), and [PM].

Combining Theorems~\ref{th:convert}, \ref{thm:char-SEU} and \ref{thm_SR_SEU}, we obtain a unifying result for SEU rating criteria: For an SEU rating criterion $\I$ whose representing SEU risk measure $\rho_{\I}$ satisfies [LC] and [$S$-LI], the following properties are equivalent: [QC]; [$S$-RA]; [$S$-PE] when $A=[0,1]$; and [$S$-PM]. Moreover, these properties are equivalent to $\rho_{\I}$ admitting the representation \eqref{eq_SEU_SLI}, where each $v_j$ is a finite convex loss function.

To summarize the findings in this section, convex loss functions are crucial for the SEU rating criteria to exhibit any of the risk consistency properties. 
For practical modeling purposes, a convex loss function or a class of convex loss functions should be chosen. This is consistent with the choice of Morningstar in the industry.

\section{Case studies}
\label{sec:case}

This section illustrates two implications of our framework using standard financial products. The collateralized loan obligation (CLO) case demonstrates the pooling effect under different credit rating criteria, while the investment fund case compares scenario-based rankings with law-invariant rankings. An additional catastrophe bond case is provided in Appendix \ref{appx:CATbond}.

\subsection{Collateralized loan obligations}\label{sec:CLO}

We give a numerical illustration using CLOs, which accounted for an \$82 billion asset class within the broader \$388 billion US asset backed securities in 2024.\footnote{The numbers are from SIFMA report of US Asset Backed Securities Statistics. Available at \url{https://www.sifma.org/resources/research/statistics/us-asset-backed-securities-statistics/} (accessed September 2026).} 
Assume for a CLO product, the underlying pool consists of identical corporate loans with five-year maturities and Ba3 ratings using Moody's rating symbols.\footnote{Moody's ratings are ordered as Aaa, Aa, A, Baa, Ba, B, Caa, Ca, and C, with categories from Aa to Caa further split into three notches (1, 2, 3, where 1 is best). We map these ratings to integers from 1 (Aaa) to 21 (C). Ba3 corresponds to 13.} 
Assume that $Z>0$ is a random variable representing the common risk factor. Conditional on $Z=z$,  the random loss $L_\ell$ for each loan follows a Beta distribution $\mathrm{Beta}(z,1)$.
This model admits a convenient representation $L_\ell=U_\ell^{1/Z}$ where $U_\ell$ is uniformly distributed on $[0,1]$ and independent of $Z$. Note that a larger value of $z$ makes $L_\ell$ strictly riskier.

To accommodate different scenarios, we assume that the loan's credit quality would be improved with 50\% probability under an optimistic scenario, and be worse with 50\% probability under a pessimistic scenario. Correspondingly, we assume  that $Z$ follows a uniform distribution on $[0.007, 0.009]$ and on $[ 0.1, 0.15]$, respectively, in the two scenarios. Since $\E[L_\ell|Z=z] = z/(1+z)$, the ranges of the expected loss for $L_\ell$ are $[0.0070
,0.0089]$ in the optimistic scenario and $[0.091,0.130]$ otherwise. These ranges align with the values associated with Moody's Baa2 and B2 rating categories, respectively, according to Moody's idealized default and loss rate table.\footnote{Moody's ``Rating Symbols and Definitions'' document (\url{https://ratings.moodys.com/rmc-documents/53954}, accessed September 2026) includes a link on page 41 to download the table of Moody's idealized default and loss rates.}

To construct comparable rating thresholds for all rating categories under the newly proposed criteria, we start from the thresholds given by Moody's idealized default and loss rate table for the Average EL criterion.
For other criteria, we calibrate the thresholds so that, under each criterion, the senior tranche with $K=0.1$ receives a Ba2 rating when there is only one Ba3-rated asset in the pool. 
To implement this, we calculate the risk measure values under each criterion in this single asset case, and then scale the EL thresholds by the ratio of the given criterion's value to the EL value. The resulting calibrated thresholds for all criteria are reported in Table \ref{tbl:mapping_clo} in Appendix \ref{appx:mapping}.

Figure \ref{fig_measures2} illustrates the $S$-pooling effect of senior tranches with $K=0.1$.\footnote{\cite{GN23} reported that the equity tranche comprises, on average, 10.3\% of the pool principal for their CLO sample spanning January to August 2020.} Panel (a) displays the rating values of Average EL, Average ES ($p = 0.9$), and Average MAXVAR ($\gamma = 0.3$) introduced in Table \ref{tab:examplesSR}. 
As anticipated by Propositions \ref{prop:SR-rating}--\ref{prop_RCconcave}, panel (a) shows an improving trend as the number of loans increases. Furthermore, the ratings based on Average EL tend to lag behind those based on Average ES and Average MAXVAR. Intuitively, EL reflects only the mean loss and therefore captures diversification benefits less strongly than tail-sensitive measures.
Panel (b) presents the ratings of Average PD, Average VaR ($p = 0.8$) and Max VaR ($p = 0.8$) defined in Table \ref{tab:examplesSR}. We observe that the ratings of these criteria do not improve monotonically  with the number of assets, indicating that a larger number of conditionally iid assets may worsen the assigned rating, arguably an undesirable situation.
Thus, Average ES and Average MAXVAR may serve as more suitable criteria for capturing the preferences of risk-averse investors among all candidates.

\begin{figure*}[t]
    \centering
    \caption{$S$-pooling effect on CLO tranche ratings as loans are added sequentially}
    \label{fig_measures2}

    \begin{subfigure}[b]{0.45\textwidth}
        \centering
        \caption[]{}
        \includegraphics[width=\textwidth]{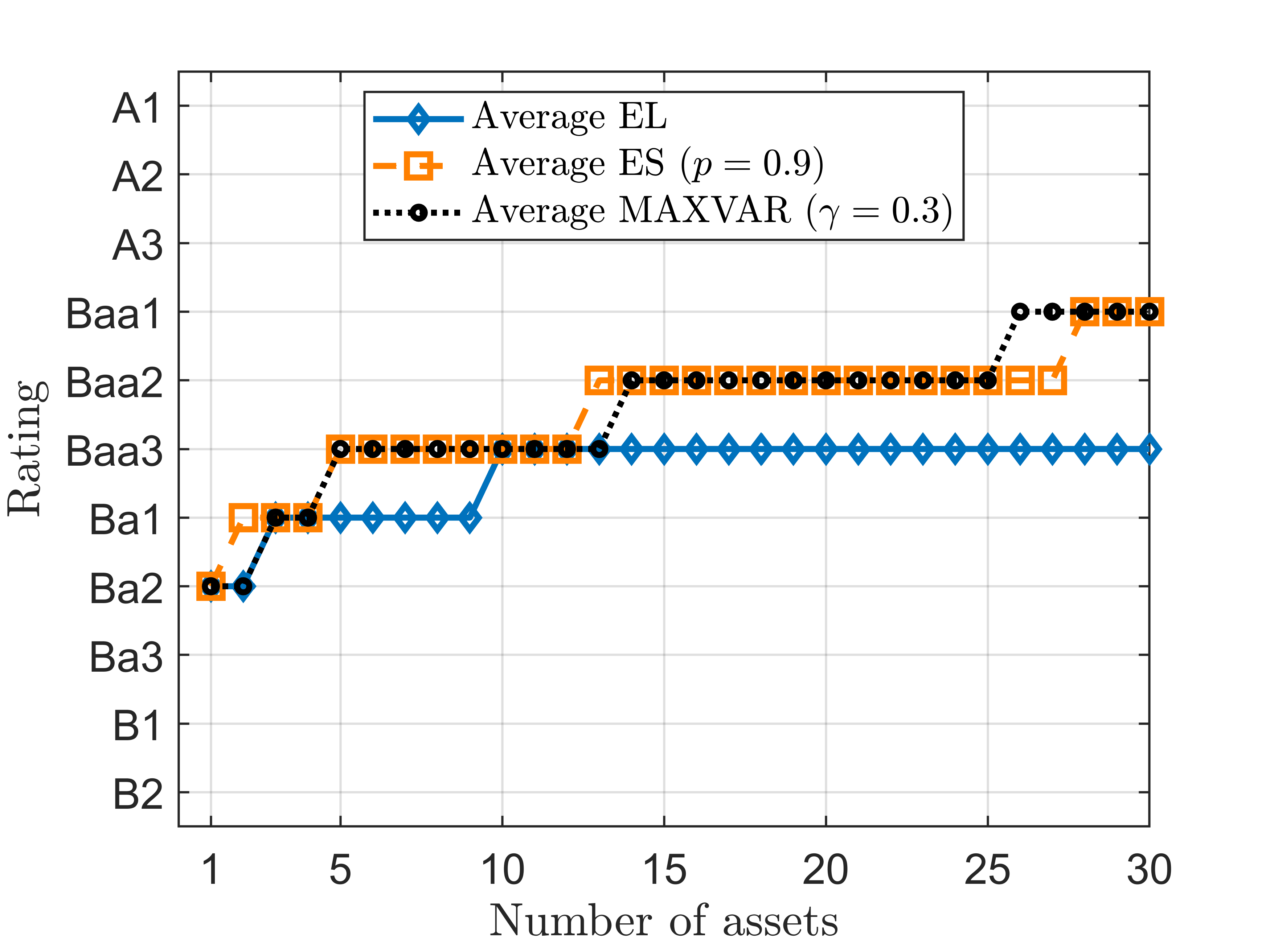}
        \label{fig:mean and std of net34}
    \end{subfigure}
    \hfill
    \begin{subfigure}[b]{0.45\textwidth}
        \centering
        \caption[]{}
        \includegraphics[width=\textwidth]{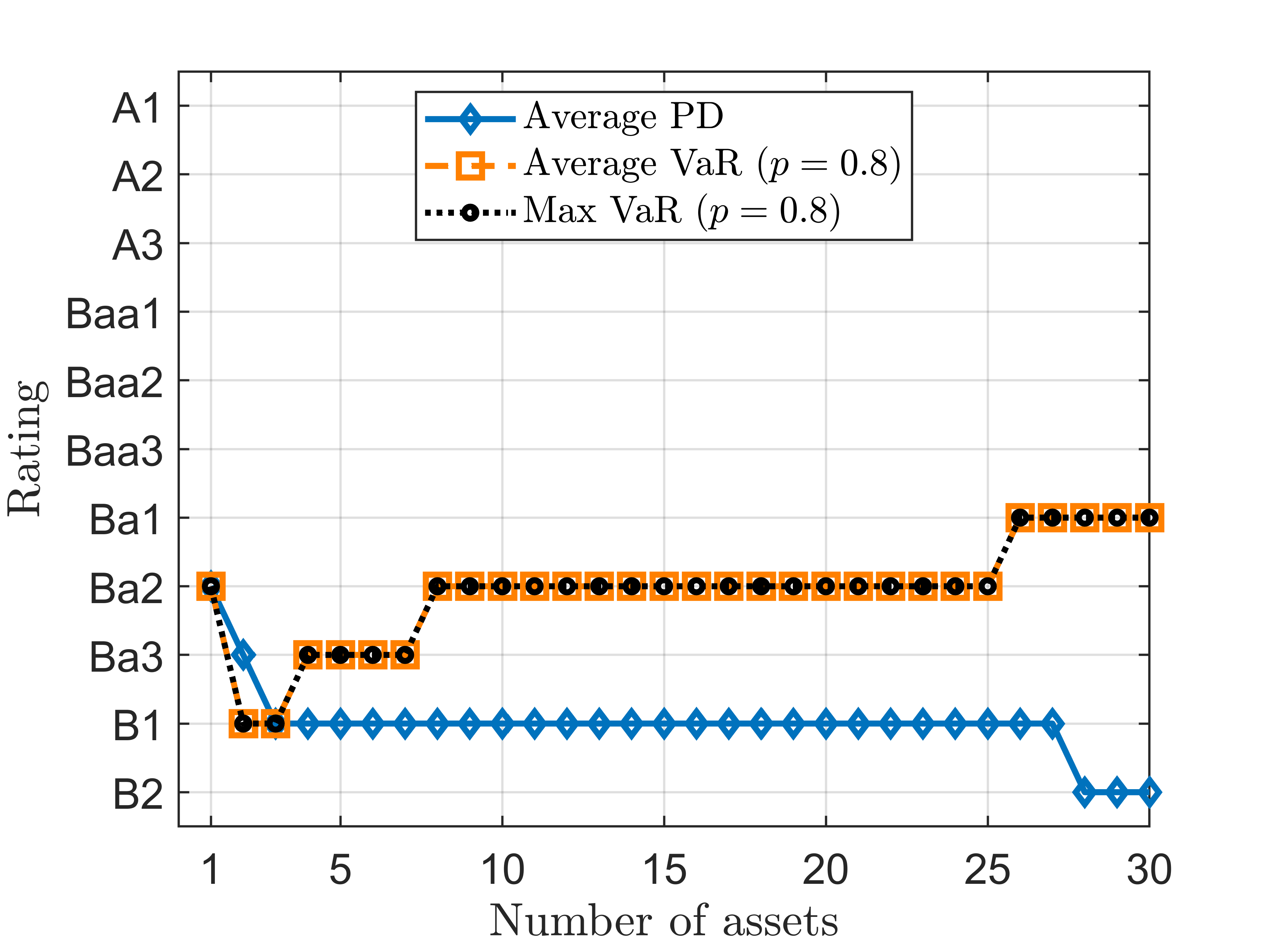}
        \label{fig:mean and std of net44}
    \end{subfigure}
\end{figure*}

\subsection{Investment funds}
We compare Morningstar's law-invariant rating measure with a scenario-based SEU measure. 
Let $L$ denote the monthly negative excess return of a fund (i.e., the negative geometric return after risk-free adjustment). As previously introduced, \cite{M23} ranks funds by \eqref{eq_morningstar}.
In contrast, we propose a corresponding SEU risk measure
 $$\sum_{j=1}^{3}q_j\mathbb E[(1-L)^{-2}-1\mid S_j],$$
which allows different weights across three economic states, where $q_j\geq0$, $\sum_{j=1}^{3}q_j=1$. Note that this is a special case of \eqref{eq_SEU_SLI} by taking $v_j(\ell)=((1-\ell)^{-2}-1)q_j/\p(S_j)$ with $v_j(1)=\infty$. 

The three economic states are constructed from the terciles of CFNAI-MA3, the three-month moving
average of the Chicago Fed National Activity Index.\footnote{The Chicago Fed reports that the CFNAI is constructed from 85 individual indicators and can be decomposed into four broad categories: production and income; sales, orders, and inventories; employment, unemployment, and hours; and personal consumption and housing. A value of zero indicates that the national economy is expanding at its historical trend rate; negative values indicate below-average growth, while positive values indicate above-average growth.} Therefore they have approximately equal empirical probabilities in the sample. However, the SEU weights $q_j$
 are not restricted to equal these probabilities.
 Motivated by the idea that losses matter more in weaker economic states, we assign higher weights to distressed states. Specifically, the weights for the weak, normal, and strong states are set to $(0.5,0.3,0.2)$.

We apply the comparison to four widely traded U.S. index ETFs: SPY, QQQ, DIA, and MDY, representing different segments of the U.S. equity market. SPY tracks the S\&P 500 and serves as a broad large-cap benchmark. QQQ tracks the Nasdaq-100 and represents technology-oriented exposure. DIA tracks the Dow Jones Industrial Average and represents blue-chip U.S. equities. MDY tracks the S\&P MidCap 400 and represents mid-cap U.S. equities. 
The monthly return data are retrieved from CRSP Survivor-Bias-Free US Mutual Fund Database.

Our evaluation period runs from January 2016 to December 2025. At each evaluation date, we use the preceding three years of monthly observations and define the economic states within the corresponding rolling window. Figure \ref{fig:SEU} reports the monthly rankings of the four ETFs computed from rolling three-year windows, where rank 1 denotes the most preferred ETF. The two measures agree on the main empirical message. QQQ is ranked first most frequently under both methods, 
while MDY is ranked last in most months under both methods. This suggests that the two criteria deliver similar broad conclusions about the relative performance of the ETFs.

Nevertheless, Figure \ref{fig:SEU} also shows important differences between the two measures. The divergence is most pronounced from late 2022 through early 2025, when QQQ's Morningstar ranking falls sharply, sometimes to the bottom. This deterioration is associated with the large drawdowns of growth and technology stocks in 2022, which generated substantial losses for QQQ and remained in the three-year rolling evaluation windows.
By contrast, QQQ's SEU ranking is more stable. A plausible explanation is that the SEU criterion places less weight on losses occurring in strong economic states. Thus, the SEU evaluation suggests that QQQ may be more favorable than its law-invariant Morningstar ranking indicates.

\begin{figure}[t]
    \centering
    \caption{Rolling three-year rankings of U.S. equity ETFs under Morningstar and SEU Measures}
    \includegraphics[width=0.9\textwidth]{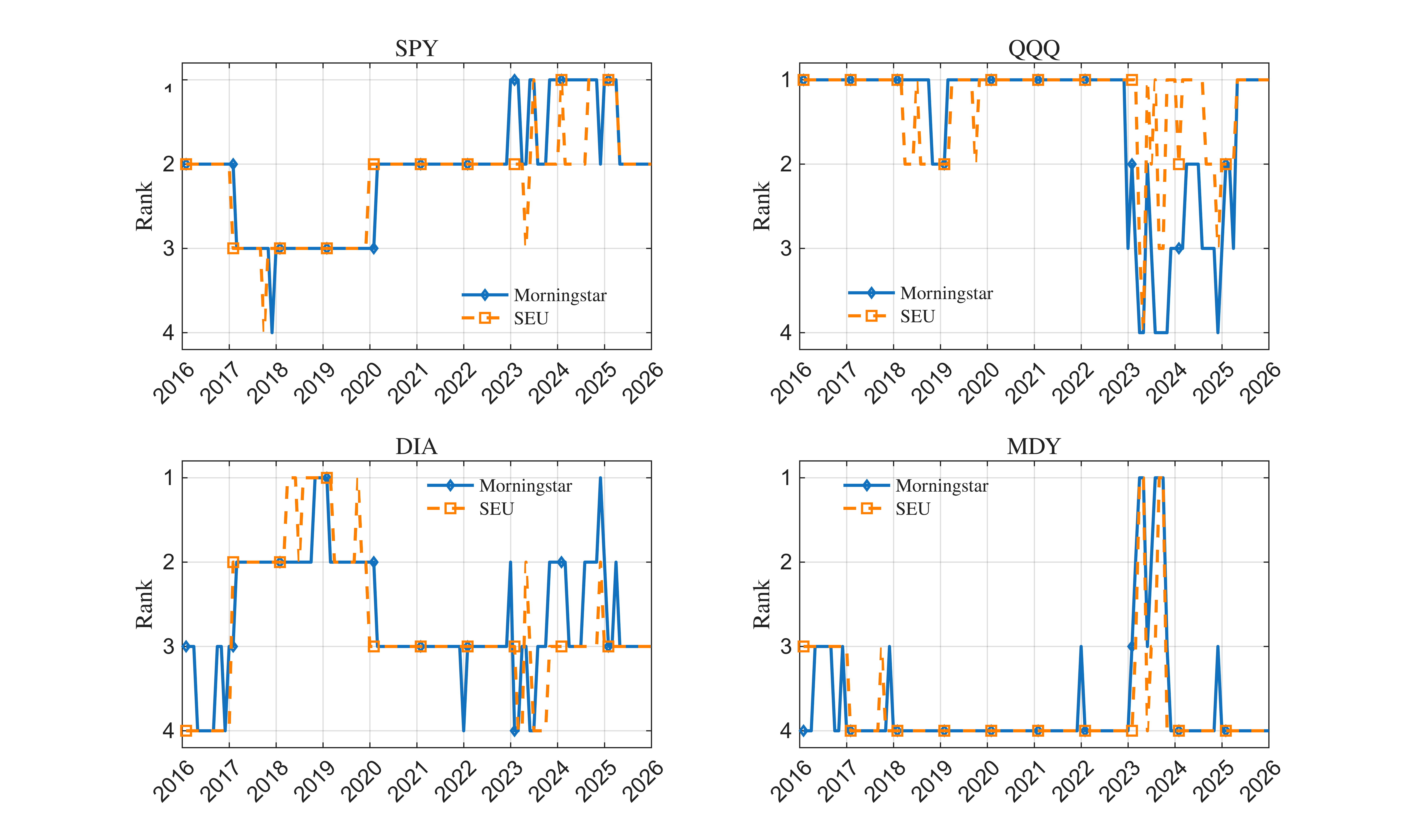}
    \label{fig:SEU}
\end{figure}

\section{Conclusion}
\label{sec:conc}
In this paper, we offer a first comprehensive study of risk 
consistency properties for risk
measures and rating criteria. 
        These risk 
consistency properties are formulated under both law invariance and scenario-based law invariance, and their connections are established. 
Our main results show how these properties translate between MPH and SEU rating criteria and their representing risk measures, and characterize Choquet and SEU rating criteria satisfying the risk consistency properties.

Our results provide large classes of risk measures and rating criteria with complementary foundations: 
\begin{enumerate}[(a)]
    
   \item self-consistency for Choquet (special case of MPH)  criteria from a market perspective, 
    \item SEU foundations for fund ratings from a decision-theoretic perspective, 
    \item risk consistency from a risk-management perspective, 
    \item scenario-based law invariance from a modeling and statistical perspective. 
\end{enumerate}
In credit rating, Average ES and Average MAXVAR satisfy (a), (c), and (d), and hence offer risk-sensitive alternatives to Moody's scenario-based EL criterion (Average EL in Table \ref{tab:examplesSR}). In fund rating, the SEU criteria of the form \eqref{eq_SEU_SLI} with convex loss functions satisfy (b)--(d), and can reveal state-contingent performance differences that may be missed by a law-invariant Morningstar-type measure.
Therefore, the improvements and generalizations developed in this paper may be useful when expectation-based evaluation does not align with the risk-management objectives of investors who use categorical ratings or risk measures.

\section*{Acknowledgments}
The research was supported by the National Natural Science Foundation of China
(Grant Nos.~12471445, 12071016) and the Natural Sciences and Engineering Research
Council of Canada (RGPIN-2024-03728, CRC-2022-00141).

\newpage
\appendix
\renewcommand{\theequation}{\Alph{section}.\arabic{equation}}
\renewcommand{\theHequation}{appendix.\Alph{section}.\arabic{equation}}
\providecommand{\theHtheorem}{}
\renewcommand{\theHtheorem}{appendix.\Alph{section}.\arabic{theorem}}
\providecommand{\theHlemma}{}
\renewcommand{\theHlemma}{appendix.\Alph{section}.\arabic{lemma}}
\providecommand{\theHproposition}{}
\renewcommand{\theHproposition}{appendix.\Alph{section}.\arabic{proposition}}
\providecommand{\theHdefinition}{}
\renewcommand{\theHdefinition}{appendix.\Alph{section}.\arabic{definition}}
\providecommand{\theHcorollary}{}
\renewcommand{\theHcorollary}{appendix.\Alph{section}.\arabic{corollary}}
\providecommand{\theHexample}{}
\renewcommand{\theHexample}{appendix.\Alph{section}.\arabic{example}}
\providecommand{\theHremark}{}
\renewcommand{\theHremark}{appendix.\Alph{section}.\arabic{remark}}
\providecommand{\theHfigure}{}
\renewcommand{\theHfigure}{appendix.\Alph{section}.\arabic{figure}}
\providecommand{\theHtable}{}
\renewcommand{\theHtable}{appendix.\Alph{section}.\arabic{table}}

\renewcommand{\thetheorem}{\Alph{section}.\arabic{theorem}}
\renewcommand{\theproposition}{\Alph{section}.\arabic{proposition}}
\renewcommand{\thelemma}{\Alph{section}.\arabic{lemma}}
\renewcommand{\thedefinition}{\Alph{section}.\arabic{definition}}
\renewcommand{\thecorollary}{\Alph{section}.\arabic{corollary}}
\renewcommand{\theexample}{\Alph{section}.\arabic{example}}
\renewcommand{\theremark}{\Alph{section}.\arabic{remark}}
\renewcommand{\thefigure}{\Alph{section}.\arabic{figure}}
\renewcommand{\thetable}{\Alph{section}.\arabic{table}}
\setcounter{equation}{0}
\setcounter{theorem}{0}
\setcounter{lemma}{0}
\setcounter{proposition}{0}
\setcounter{definition}{0}
\setcounter{corollary}{0}
\setcounter{example}{0}
\setcounter{remark}{0}
\setcounter{table}{0}
\setcounter{figure}{0}

\begin{center}
    \Large Supplementary Material: Appendices
\end{center}
The supplementary material includes proofs of all results, supplementary theoretical results, and additional numerical results.

\section{Proofs}\label{app:A}

\subsection{Proofs for Section \ref{sec:prelim}}

\begin{proof}[Proof of Lemma \ref{lemma:normal_SEU}]
        The ``if'' direction is immediate. For the ``only if'' direction, we have
        \begin{equation}
			\rho(X)=\int_\Omega [v_\omega(X(\omega))-v_{\omega}(0)]\dd \mu=
            \int_\Omega [v_\omega(X(\omega))-v_{\omega}(0)]\frac{\dd \mu}{\dd \mathbb{P}}(\omega)\dd \mathbb{P},\quad X\in\mathcal L_A.
	    \end{equation}
        Define $\tilde v_\omega(x)=[v_\omega(x)-v_{\omega}(0)]\dd \mu/\dd \mathbb{P}(\omega)$ for $\omega\in \Omega$ and $x \in A$. We use the convention $\infty\cdot 0=0$. 
        Then we have $\tilde v_{\omega}(x)\equiv0$ or $\tilde v_\omega$ is increasing with $\tilde v_{\omega}(0)=0$ and is finite on $A\setminus\{1\}$.
        Moreover,
        $\rho(X)=\int_\Omega \tilde v_\omega(X(\omega))\dd \mathbb{P}$ for $X\in \L_A$, which gives representation \eqref{eq_SEU_expectation}.

\end{proof}

\begin{proof}[Proof of Proposition \ref{prop_MPH_SEU_intersection}]
The ``if'' direction is immediate. Conversely, suppose that $\rho$ is both an SEU and MPH risk measure. By Lemma \ref{lemma:normal_SEU}, $\rho$ admits the form
$\rho(X)=\int_{\Omega} v_{\omega}(X(\omega))\d \mathbb{P}$, $X\in\mathcal L_A$, 
where each $v_\omega$ is increasing and satisfies $v_\omega(0)=0$. Since $\rho$ is an MPH measure, we have $\rho(0)=0$ and $\rho(1)=1$, and therefore $v_\omega(1)<\infty$ for $\mathbb P$-a.s. $\omega$.

We claim that
\begin{equation}\label{eq:SEU_MPH_claim}
\int_B v_\omega(x)\,\dd\mathbb P = \int_B x v_\omega(1)\,\dd\mathbb P,\quad \forall B\in\mathcal F, x\in A.
\end{equation}
Fix a set $B\in\mathcal F$. For $x\in [0,1]$, positive homogeneity gives
$\rho(x\id_B)=x\rho(\id_B)$, thus \eqref{eq:SEU_MPH_claim} holds. For $x\in [-1,0]\cap A$, 
we have 
\begin{align*}
    \int_B v_\omega(x)\,\dd\mathbb P=&\rho(x\id_B)=\rho(x+(-x)\id_{B^c})=x+\rho((-x)\id_{B^c})\\
    =&x\rho(1)+(-x)\rho(\id_{B^c})=\int_B x v_\omega(1)\,\dd\mathbb P.
\end{align*}
By positive homogeneity, the identity \eqref{eq:SEU_MPH_claim}
extends to all $x\in A$.

It follows that $v_\omega(x)=xv_\omega(1)$, $\mathbb{P}$-a.s. for all $x\in A$.
Consequently, 
$$N=\{\omega\in \Omega: v_\omega(x)\neq xv_\omega(1) \mathrm{~for~some~} x\in \mathbb Q\cap A\}=\bigcup_{x\in \mathbb Q\cap A}\{\omega\in \Omega: v_\omega(x)\neq xv_\omega(1)\}$$
is a $\mathbb{P}$-null set, where $\mathbb Q$ denotes the set of rational numbers.
Since each $v_\omega$ is increasing and $\mathbb{Q}$ is dense in $A$, for all $\omega\notin N$, we have $v_\omega(x)=xv_\omega(1)$, $x\in A$.
Define
$$\mu(B) := \int_B v_\omega(1)\,\dd\mathbb P,\quad B\in\mathcal F.$$
Then $\mu\ll\mathbb P$ and it is a probability measure due to
$\mu(\Omega)=\rho(1)=1$. Finally,
$$\rho(X) = \int_\Omega v_\omega(X(\omega))\dd\p=\int_\Omega X(\omega)v_\omega(1)\dd\p=\int_\Omega X\,\dd\mu,\quad X\in\mathcal L_A.$$
This completes the proof.

\end{proof}

\begin{proof}[Proof of Proposition \ref{prop_SEU_LI_SLI}]
        We prove the results with [$S$-LI], i.e. parts (ii) and (iv). The proof for [LI] is the same with $S=\Omega$.
        
		(ii) The ``if'' direction is immediate. Conversely, suppose that $\rho$ satisfies [$S$-LI]. By Lemma \ref{lemma:normal_SEU}, $\rho$ admits the representation
		$\rho(X)=\int_\Omega v_\omega(X(\omega))\dd\mathbb{P}$, $X\in\mathcal L_A$, where $v_{\omega}(0)=0$ for $\omega\in \Omega$.
		In the following, we show that for every $j\in [s]$ and $x\in A$, 
		\begin{equation}
		    v_\omega(x)=v_j(x),\ \mathbb P_j\text{-a.s., where } v_j(x)=\int_{S_j} v_\omega(x)\dd\mathbb P_j \in \overline{\R}.\label{eq_vj_const}
		\end{equation}
		Suppose otherwise, that \eqref{eq_vj_const} fails for some $j\in [s]$ and $x\in A$. If $v_j(x)<\infty$, let $r=v_j(x)$.
        Since $\int_{S_j} (v_\omega(x)-v_j(x))\dd\mathbb P_j=0$, the sets
		\begin{equation}
        B_1=\{\omega\in S_j: v_\omega(x)<r\} \mathrm{~~and~~} B_2=\{\omega\in S_j: v_\omega(x)>r\}\label{eq_B_C}
        \end{equation}
		both have positive $\mathbb P_j$-measure and $\mathbb P$-measure.
        If $v_j(x)=\infty$ (necessarily $x=1$), since \eqref{eq_vj_const} fails, there exists $r\in\mathbb R$ such that the set $B_1$ in \eqref{eq_B_C}
        has positive $\mathbb P_j$-measure. Since $v_j(x)=\infty$, the set $B_2$ in \eqref{eq_B_C} also has positive $\mathbb P_j$-measure.
        By atomlessness of $\mathbb P$, we can choose $D_1\subseteq B_1$ and $D_2 \subseteq B_2$ such that $\mathbb{P}(D_1)=\mathbb{P}(D_2)>0$. Then
		\begin{equation}
        \rho(x\id_{D_1})=\int_{D_1}v_{\omega}(x) \dd \mathbb{P}< r\mathbb{P}(D_1) = r\mathbb{P}(D_2)< \int_{D_2}v_{\omega}(x) \dd \mathbb{P}=\rho(x\id_{D_2}),\label{eq_B<C}
        \end{equation}
		which contradicts [$S$-LI]. Thus, \eqref{eq_vj_const} holds.
		Since $v_j$ and $v_{\omega}$ are increasing, with the same argument as the proof in Proposition \ref{prop_MPH_SEU_intersection}, there exists a null set $N_j\subseteq S_j$ such that 
		$v_\omega(x)=v_j(x)$, $x\in A$, for all $\omega\in S_j\setminus N_j$. Therefore, 
        \begin{equation*}
            \rho(X)=\mathbb{E}[v_\omega(X(\omega))]=
        \sum_{j=1}^s \mathbb E[v_j(X)\id_{ S_j}]
        \end{equation*}
		This completes the proof.

        (iv) The ``if'' direction is immediate. Conversely, suppose that $\I$ satisfies [$S$-LI]. Suppose that $\I$ is generated by an SEU measure $\rho$ with a threshold $c$ such that $\rho>c$ and $\rho<c$ are rated differently. We can take $Y_1,Y_2\in \L_A$ such that 
        $$\rho(Y_1)<c<\rho(Y_2)<\infty.$$ 
        It suffices to show that \eqref{eq_vj_const} holds. Suppose otherwise. As in the proof of part (ii), there exists $r\in\mathbb R$ such that the sets $B_1,B_2$ in \eqref{eq_B_C}
		both have positive $\mathbb P$-measure.
        By atomlessness of $\mathbb P$, we can choose $D_1\subseteq B_1$ and $D_2 \subseteq B_2$ such that $\mathbb{P}(D_1)=\mathbb{P}(D_2)>0$ with this common value chosen arbitrarily small that
        $$\rho(Y_1)+\int_{D_1}v_{\omega}(x) \dd \mathbb{P}-\int_{D_1\cup D_2}v_{\omega}(Y_1(\omega))\dd \mathbb{P}
        <c<\rho(Y_2)+\int_{D_2}v_{\omega}(x) \dd \mathbb{P}-\int_{D_1\cup D_2}v_{\omega}(Y_2(\omega))\dd \mathbb{P}.$$
        Let $D=D_1\cup D_2$. If $\int_{D_2}v_{\omega}(x) \dd \mathbb{P}=\infty$, then 
        $$\rho(x\id_{D_1}+Y_1\id_{D^c})<c<\rho(x\id_{D_2}+Y_1\id_{D^c}),$$
        which contradicts [$S$-LI] of $\I$.
        If $\int_{D_2}v_{\omega}(x) \dd \mathbb{P}<\infty$, we have
        $$\int_{D^{c}}v_{\omega}(Y_1(\omega))\dd \mathbb{P}
        <c-\int_{D_2}v_{\omega}(x) \dd \mathbb{P}<c-\int_{D_1}v_{\omega}(x) \dd \mathbb{P}
        <\int_{D^{c}}v_{\omega}(Y_2(\omega))\dd \mathbb{P},$$
        where the second inequality is due to \eqref{eq_B<C}. By Lyapunov's convexity theorem and the atomlessness of $\mathbb P$, there exists $C_1\subseteq D^{c}$ and $C_2=D^{c}\setminus C_1$, such that 
        $$c-\int_{D_2}v_{\omega}(x) \dd \mathbb{P}<
        \int_{C_1}v_{\omega}(Y_1(\omega))\dd \mathbb{P}+
        \int_{C_2}v_{\omega}(Y_2(\omega))\dd \mathbb{P}
        <c-\int_{D_1}v_{\omega}(x) \dd \mathbb{P}.$$
        Consequently,
        $$\rho(x\id_{D_1}+Y_1\id_{C_1}+Y_2\id_{C_2})<c        <\rho(x\id_{D_2}+Y_1\id_{C_1}+Y_2\id_{C_2}),$$
        which contradicts [$S$-LI] of $\I$. Hence, \eqref{eq_vj_const} holds and $\rho$ satisfies [$S$-LI].

\end{proof}

\subsection{Proofs for Section \ref{sec_rc}}

\begin{proof}[Proof of Lemma \ref{lem:RC-pool}]
The relation $L^{(\ell')} \le_{\rm{icx}} L^{(\ell)}$ can be seen from Example 3.A.29 of \cite{SS07}. 
	Note that the function $x \mapsto (x - K)_+/(1-K)$ for $K \in [0, 1)$ is increasing convex,
    and the compositions of increasing convex functions are increasing convex.
The desired relation follows.

\end{proof}

\begin{proof}[Proof of Theorem \ref{thm:summarize}]
We use the following domain convention throughout the proof. Except for the arguments involving [PE] and [$S$-PE], all arguments apply, without change, to risk measures
$\rho:\mathcal L^\infty\to\overline{\mathbb R}$,
$\rho:\mathcal L^\infty\to\mathbb R$, and
$\rho:\mathcal L_A\to\overline{\mathbb R}$, as well as to rating criteria
$\mathcal I:\mathcal L_A\to[n]$, with each property formulated on the corresponding domain. For extended-real-valued risk measures, all inequalities, maxima, and lower limits are understood in $\overline{\mathbb R}$. Since $A$ is an interval, the convex combinations and arithmetic averages used below remain in $\mathcal L_A$. The arguments involving [PE] and [$S$-PE] are understood on the defaultable-security domain $A=[0,1]$.

\begin{enumerate}[(i)]
\item   {[LI]} $\Rightarrow$ [$S$-LI]: It follows from the fact that $ L_1\laweq L_2$ is a weaker requirement than $ L_1\seq L_2$. 

[RA] $\Rightarrow$ [$S$-RA]: It suffices to show that $X\leq_{S\opName{-icx}}Y$ implies $X\leq_{\rm{icx}}Y$. If $X\leq_{S\opName{-icx}}Y$, i.e., $\mathbb{E}[\phi(X)|S_j] = \E^{\p_j}[\phi(X)] \leq  \E^{\p_j}[\phi(Y)] = \mathbb{E}[\phi(Y)|S_j]$ holds for any increasing convex function $\phi$ and $j\in [s]$, then $$\mathbb{E}[\phi(X)]=\sum_{j=1} ^s 
\mathbb{E}[\phi(X)\id_{S_j}]\leq\sum_{j=1}^s \mathbb{E}[\phi(Y)  \id_{S_j}]= \mathbb{E}[\phi(Y)]$$ holds for any increasing convex function $\phi$. This yields $X\leq_{\rm{icx}}Y$. Therefore, [RA] implies [$S$-RA].

[PE] $\Rightarrow$ [$S$-PE]: It suffices to show that conditionally iid under each scenario implies conditionally iid. Take any sequence $(L_\ell)_{\ell\in\N}$   in $\L_{[0,1]}$ that is conditionally iid on a random vector $Z_j$ under each $\p_j$, $j\in [s]$, and let $Y$ be a random variable that equals $j\in [s]$ under $S_j$. 
It then follows that $(L_\ell)_{\ell\in\N}$   
is conditionally iid on $(Z_Y,Y)$. Therefore, [PE] implies [$S$-PE].

[PM] $\Rightarrow$ [$S$-PM]: This follows from the same argument as [PE] $\Rightarrow$ [$S$-PE] that conditionally iid under each scenario implies conditionally iid.

With $S=\Omega$, the above implications become equivalences,   following directly from definitions.

\item 
Part (ii) is a special case of (iii) when $S=\Omega$.
\item  {[$S$-RA]} $\Rightarrow$ [$S$-LI]: For any $L_1,L_2\in \mathcal{L}_{A}$, $L_1\overset{S}{\sim} L_2$ implies both $L_1\leq_{S\opName{-icx}} L_2$ and $L_2\leq_{S\opName{-icx}} L_1$, since $\mathbb{E}^{\p_j}[\phi(L_1)]= \mathbb{E}^{\p_j}[\phi(L_2)]$ holds for any increasing convex function $\phi$ and $j\in[s]$. Therefore, [$S$-RA] implies  that $L_1$ and $L_2$ are assigned the same value, and thus [$S$-LI] holds.

[$S$-RA] $\Rightarrow$ [$S$-PM]:
For any sequence $(L_\ell)_{\ell\in\mathbb N}$ that is conditionally iid
under each $\mathbb P_j$, Lemma \ref{lem:RC-pool} applied under each $\mathbb P_j$ gives
$L^{(\ell')} \leq_{S\text{-icx}} L^{(\ell)}$, for all $\ell\leq\ell'.$
Hence [$S$-RA] implies [$S$-PM].

[$S$-RA] $\Rightarrow$ [$S$-PE]: For any  sequence  $(L_\ell)_{\ell\in\N}$   in $\L_{[0,1]}$ that is conditionally iid under each $\p_j$, $j\in [s]$, if $\ell\le \ell'$, and $K\in [0,1)$, by Lemma \ref{lem:RC-pool} applied to each $\p_j$, $j\in[s]$, we have $ (L^{(\ell')}-K)_+ \le_{S\opName{-icx}} (L^{(\ell)}-K)_+$. Therefore, [$S$-RA] implies [$S$-PE]. 

\item We first discuss the argument for $\rho:\L^\infty\to \overline\R$. Take $X,Y$ with $X\le_{\rm cx} Y$. Let $X_1,X_2,\dots$ be distributed as $Y$ such that 
$X^{(n)}:=\sum_{i=1}^n X_i/n \to X$ in $L^\infty$ as $n\to\infty$; for the existence of such a sequence, see e.g.,~\citet[Theorem 3.5 and Remark 3.2]{MW15}.
By [QC], 
$\rho(X^{(n)}) \le \max\{\rho(X_1),\dots,\rho(X_n)\} = \rho(Y)$.
By [LS*], 
$$\rho(X) \le \liminf_{n\to\infty} \rho(X^{(n)}) \le \rho(Y).$$
Therefore, $\rho$ is increasing in the convex order. This then implies [RA] because $\rho$ is monotone. Applying the preceding arguments for $\rho$ to $\mathcal I$, we similarly conclude that $\mathcal I$ satisfies [RA].

\item This follows from the same argument in (iv) by constructing the sequence $X_1,X_2,\dots$ with the corresponding distributions of $X,Y$ under each $\p_j$, $j\in [s]$. 

\item Take a  conditionally iid sequence  $(L_\ell)_{\ell\in\N}$   in $\L_{A}$ and fix $\ell\geq 2$.
For $k\in [\ell]$, let $L^{(\ell)}_{-k} $ denote the average of $L_j$ for $j\in [\ell]\setminus\{k\}$. Clearly, 
$L^{(\ell)}$
is the average of 
$
L^{(\ell)}_{-1},\dots,
L^{(\ell)}_{-\ell}.
$
Therefore, [QC] and [LI] imply 
$$
\rho(L^{(\ell)}) \le \max\{\rho(L^{(\ell)}_{-1}),\dots,
\rho(L^{(\ell)}_{-\ell}) \} 
= \rho(L^{(\ell)}_{-\ell}) = \rho(L^{(\ell-1)}).
$$
Iterating this inequality yields [PM].
\item This follows from the same argument in (vi) by taking  a conditionally iid sequence  under each $\p_j$, $j\in [s]$.
\end{enumerate}

\end{proof}

\subsection{Proofs for Section \ref{sec:translation}}

\begin{proof}[Proof of Theorem \ref{thm_uniqueness}]

\begin{enumerate}[(i)]
\item By the condition, there exist increasing functions $f_1:A\rightarrow [n]$ and $f_2:A \rightarrow [n]$ such that
\begin{equation}
    \label{eq:unique-pf}    \I(X)=f_1(\rho_1(X))=f_2(\rho_2(X)),~~~\mbox{ for all } X\in \L_A.
\end{equation} 
Note that an MPH risk measure $\rho$ satisfies $\rho(x)=x$ for $x\in \mathbb{R}$.
Therefore, we have
$$f_1(x)=f_1(\rho_1(x))=\I(x)=f_2(\rho_2(x))=f_2(x),~~~\mbox{ for all } x\in A.$$
Thus, the functions $f_1$ and $f_2$ coincide; we denote this common function by $f$. Since $f$ is not constant on the interior of $A$, there exists some $c$ such that $f(x) < f(y)$ for any $x < c < y$.

Suppose, to the contrary, that  $\rho_1 \ne \rho_2$. Then there exists  $X\in \L_A$ such that $\rho_1(X)\neq \rho_2(X)$. We can assume $\rho_1(X)< \rho_2(X)$ without loss of generality. Choose $c^{-},c^{+}\in A$ such that $c^-<c<c^+$.
Note that any MPH risk measure $\rho$ satisfies $\rho(\lambda X+a)=\lambda \rho(X)+a$ for $\lambda>0$, $a\in \R$, $X\in \L_A$ with $\lambda X+a\in A$.
If $\rho_1(X)<c<\rho_2(X)$, then it is a contradiction to \eqref{eq:unique-pf}. 
If $c\leq \rho_1(X)<\rho_2(X)$, we can find $\lambda\in (0,1)$ such that $\lambda \rho_1(X)+(1-\lambda)c^{-}<c<\lambda \rho_2(X)+(1-\lambda)c^{-}$. Then $\rho_1(\lambda X+ (1-\lambda)c^{-})<c<\rho_2(\lambda X+ (1-\lambda)c^{-})$ yields a contradiction to \eqref{eq:unique-pf}.
If $\rho_1(X)<\rho_2(X)\leq c$, we can find $\lambda\in (0,1)$ such that $\lambda \rho_1(X)+(1-\lambda)c^{+}<c<\lambda \rho_2(X)+(1-\lambda)c^{+}$. Then $\rho_1(\lambda X+ (1-\lambda)c^{+})<c<\rho_2(\lambda X+ (1-\lambda)c^{+})$ yields a contradiction to \eqref{eq:unique-pf}. 
Therefore, $\rho_1=\rho_2$ holds true on $\mathcal L_A$, and hence also on $\mathcal L^\infty$.

\item First we give a technical lemma.
	\begin{lemma}\label{lem:two-quadrant-convexity}
		Let $M\subset\mathbb R^2$ be a convex set and let $c_1,c_2\in\mathbb R$. Suppose that
		$$M\subset \bigl((-\infty,c_1]\times(-\infty,c_2]\bigr) \cup \bigl([c_1,\infty)\times[c_2,\infty)\bigr).$$
		If there exist $a,b\in M$ such that
		$a\in (-\infty,c_1)\times(-\infty,c_2)$ and $b\in (c_1,\infty)\times(c_2,\infty)$,
		then $M$ is contained in an affine line passing through $(c_1,c_2)$.
	\end{lemma}
	\begin{proof}
		Let $\ell$ be the affine line passing through $a$ and $b$, and let $\ell_0$ be the line segment joining $a$ and $b$. First, we show that $(c_1,c_2)\in \ell_0$.
		Since $M$ is convex, $\ell_0\subset M$.
		Suppose that $(c_1,c_2)\notin \ell_0$. Then, along the segment from $a$ to $b$, one coordinate must cross its threshold before the other. Hence there exists a point $d\in \ell_0$ belonging to either
		$(-\infty,c_1)\times(c_2,\infty)$ or $(c_1,\infty)\times(-\infty,c_2)$,
		contradicting the assumption on $M$. Therefore,	$(c_1,c_2)\in \ell_0.$
		
		Next, we claim that $M\subset\ell$. If there exists $z\in M\setminus\ell$, by the convexity of $M$, the convex hull $\operatorname{conv}\{a,b,z\}$ is a nondegenerate triangle contained in $M$. Thus, near $a$, we can choose 
		$$a'\in \operatorname{conv}\{a,b,z\}\cap \Big((-\infty,c_1)\times(-\infty,c_2)\Big)$$ 
		such that $a'\notin \ell$. However, with the same argument above, $(c_1,c_2)$ is also in the segment from $a'$ to $b$, which contradicts. Hence $M\subset\ell$.

\end{proof}

   Then we give the proof of part (ii) of Theorem \ref{thm_uniqueness}. 
   
By the definition of rating criterion, there exist thresholds $c_1$ for $\rho_1$ and $c_2$ for $\rho_2$ such that the two functionals distinguish the same pair of adjacent ratings. Hence
\begin{equation}
    (\rho_1(X)-c_1)(\rho_2(X)-c_2)\ge0,\quad X\in\mathcal L_A.\label{eq_threshold_prop}
\end{equation}
By Lemma \ref{lemma:normal_SEU}, for $i=1,2$,
\begin{equation}
    \rho_i(X)=\int_\Omega v_{i,\omega}(X(\omega))\dd\mathbb P,
    \quad X\in\mathcal L_A,
\end{equation}
where $v_{i,\omega}(0)=0$ for $\omega\in\Omega$.

We first note that $\rho_1$ and $\rho_2$ have the same finiteness domain. Suppose, for instance, that $\rho_1(X)<\infty$ and $\rho_2(X)=\infty$ for some $X\in\mathcal L_A$. Choose $Y_1\in\mathcal L_A$ such that $\rho_1(Y_1)<c_1$ and $\rho_2(Y_1)<c_2$. Since $\mathbb P$ is atomless and
$$
    \int_\Omega v_{1,\omega}(X(\omega))\dd\mathbb P<\infty,
    \quad
    \int_\Omega v_{2,\omega}(X(\omega))\dd\mathbb P=\infty,
$$
we may choose a measurable set $D\subseteq\Omega$ with $\p(D)>0$ sufficiently small such that
$$
    \rho_1(Y_1\id_{D^c}+X\id_D)<c_1
    \quad\text{and}\quad
    \rho_2(Y_1\id_{D^c}+X\id_D)=\infty>c_2 .
$$
This contradicts that the two representations generate the same rating criterion. The reverse case is identical. Hence $\rho_1(X)<\infty$ if and only if $\rho_2(X)<\infty$.

Next, define
$$
    M=\{(\rho_1(X),\rho_2(X)):X\in\mathcal L_A,\ \rho_1(X),\rho_2(X)<\infty\}.
$$
We claim that $M$ is convex. Take $X,Y\in\mathcal L_A$ such that $(\rho_1(X),\rho_2(X)),(\rho_1(Y),\rho_2(Y))\in M$, and let $\lambda\in(0,1)$. For $i=1,2$, define $g_i(\omega)=v_{i,\omega}(X(\omega))-v_{i,\omega}(Y(\omega))$, $\omega\in\Omega$, and consider the vector measure
$$
    m(B)=\left(\int_B g_1\dd\mathbb P,\int_B g_2\dd\mathbb P\right),
    \quad B\in\mathcal F.
$$
Since $\mathbb P$ is atomless, $m$ is atomless. By Lyapunov's convexity theorem, the range $\{m(B):B\in\mathcal F\}$ is convex. Hence it contains the segment from $m(\varnothing)=(0,0)$ to $m(\Omega)$. Therefore, there exists $B\in\mathcal F$ such that
$$
    \left(\int_B g_1\dd\mathbb P,\int_B g_2\dd\mathbb P\right)
    =
    \left(\lambda\int_\Omega g_1\dd\mathbb P,\lambda\int_\Omega g_2\dd\mathbb P\right).
$$
Define $Z=X\id_B+Y\id_{B^c}$. Then, for $i=1,2$,
$$
    \rho_i(Z)-\rho_i(Y)
    =
    \int_B g_i\dd\mathbb P
    =
    \lambda\int_\Omega g_i\dd\mathbb P
    =
    \lambda(\rho_i(X)-\rho_i(Y)).
$$
Thus $(\rho_1(Z),\rho_2(Z))=\lambda(\rho_1(X),\rho_2(X))+(1-\lambda)(\rho_1(Y),\rho_2(Y))$, and hence $M$ is convex.

By the threshold property \eqref{eq_threshold_prop},
$$
    M\subseteq\bigl((-\infty,c_1]\times(-\infty,c_2]\bigr)
    \cup
    \bigl([c_1,\infty)\times[c_2,\infty)\bigr).
$$
Moreover, $M$ intersects both $(-\infty,c_1)\times(-\infty,c_2)$ and $(c_1,\infty)\times(c_2,\infty)$.
Applying Lemma \ref{lem:two-quadrant-convexity}, we conclude that $M$ is contained in an affine line passing through $(c_1,c_2)$. Since $(0,0)\in M$, this affine line passes through the origin. Hence there exists $c\in\mathbb R$ such that
$$
    \rho_1(X)=c\rho_2(X)
$$
whenever $\rho_1(X),\rho_2(X)<\infty$. By monotonicity and nondegeneracy, $c>0$. Finally, for $X$ outside the common finiteness domain, we have $\rho_1(X)=\rho_2(X)=\infty$, so the same equality holds in $\overline{\mathbb R}$. Therefore $\rho_1=c\rho_2$ on $\mathcal L_A$.

\end{enumerate}

\end{proof}

\begin{proof}[Proof of Theorem \ref{th:convert}]
We prove here the assertions for MPH rating criteria. 
The assertions for SEU rating criteria, namely the SEU case of part~(i) and part~(iii), 
are established in the proof of Theorem~\ref{thm:char-SEU}.

Let $\I$ be an MPH rating criterion.

\begin{enumerate}[(i)]
\item If $\rho_\I$ satisfies one of the above properties, then $\I$ satisfies the same property since $\I=f_\I\circ\rho_\I$ and $f_\I$ is increasing. 
Conversely, we will show that if $\I$ satisfies any one of the above properties, then so does $\rho_\I$. Suppose otherwise; that is, one of the following conditions holds:
\begin{enumerate}[(a)]
    \item $\I$ satisfies [LI] and $\rho_{\I}(X)< \rho_{\I}(Y)$ for some $X,Y\in \L_A$ with $X\overset{\d}{=}Y$;
    \item $\I$ satisfies [$S$-LI] and $\rho_{\I}(X)< \rho_{\I}(Y)$ for some $X,Y\in \L_A$ with $X\overset{S}{\sim}Y$;
    \item $\I$ satisfies [RA] and $\rho_{\I}(X)< \rho_{\I}(Y)$ for some $X,Y\in \L_A$ with $Y\le_{\opName{icx}}X$;
    \item $\I$ satisfies [$S$-RA] and $\rho_{\I}(X)< \rho_{\I}(Y)$ for some $X,Y\in \L_A$ with $Y\le_{S\opName{-icx}}X$;
    \item $\I$ satisfies [PM] and $\rho_{\I}(X)< \rho_{\I}(Y)$, where $X=L^{(\ell)}$ and $Y=L^{(\ell')}$ with $\ell\le \ell'$ and some conditionally iid sequence $L_1,L_2,\dots\in \L_A$;
    \item $\I$ satisfies [$S$-PM] and $\rho_{\I}(X)< \rho_{\I}(Y)$, where $X=L^{(\ell)}$ and $Y=L^{(\ell')}$ with $\ell\le \ell'$ and some sequence $(L_\ell)_{\ell\in\N}$ in $\L_A$ that is conditionally iid under each $\p_j$, $j\in [s]$;
    \item $\I$ satisfies [QC] and $\rho_{\I}(Z)\le \rho_{\I}(X)< \rho_{\I}(Y)$, where $Y=\lambda X+(1-\lambda) Z$ for some $X,Z\in \L_A$ and $\lambda \in [0,1]$.
\end{enumerate}

Take a discontinuity point $c$ of $f_{\I}$ in $A$ such that any $X_{-}$ with $\rho_{\I}(X_{-})<c$ and any $X_+$ with $\rho_{\I}(X_+)>c$ are rated differently by $\I$. Take $c^-,c^+\in A$ such that $c^{-}<c<c^{+}$. Choose $d$ as follows:
$$
d=
\begin{cases}
c^-, & c\leq \rho_\I(X)<\rho_\I(Y),\\
c, & \rho_\I(X)<c<\rho_\I(Y),\\
c^+, & \rho_\I(X)<\rho_\I(Y)\leq c.
\end{cases}
$$
Then we can always choose $\theta\in(0,1)$ such that
$$
    \theta\rho_\I(X)+(1-\theta)\rho_\I(d)
    <c<
    \theta\rho_\I(Y)+(1-\theta)\rho_\I(d).
$$
Set $\widetilde X=\theta X+(1-\theta)d$ and $\widetilde Y=\theta Y+(1-\theta)d$. Since $\rho_\I$ is an MPH risk measure,
$$
    \rho_\I(\widetilde X)<c<\rho_\I(\widetilde Y).
$$
Hence $\widetilde X$ is rated better than $\widetilde Y$, which yields a contradiction if any of the conditions (a)--(d) holds, since the corresponding relation between $X$ and $Y$ is preserved under the common convex combination with the constant $d$.

If condition (e) or (f) holds, then the sequence $\theta L_1+(1-\theta)d,\theta L_2+(1-\theta)d,\dots$ is still conditionally iid, respectively conditionally iid under each $\p_j$, and its pooled losses are $\widetilde X$ and $\widetilde Y$. Applying [PM] or [$S$-PM] gives a contradiction. If condition (g) holds, set $\widetilde Z=\theta Z+(1-\theta)d$. Then $\rho_\I(\widetilde Z)\leq \rho_\I(\widetilde X)<\rho_\I(\widetilde Y)$ and $\widetilde Y=\lambda\widetilde X+(1-\lambda)\widetilde Z$, contradicting [QC]. Hence $\rho_\I$ satisfies the same property.

\item If $f_\I$ is left-continuous and $\rho_\I$ satisfies property [LS] (resp.~[LS*]), then $\I$ satisfies the same property by definition \eqref{eq:p2r}. Now suppose that $\I$ satisfies property [LS] (resp.~[LS*]).

Taking random variables $X_{\ell}=\lambda_{\ell}$ and $X=\lambda$ with constants $\lambda,\lambda_1,\lambda_2,\ldots\in A$ and $\lambda_{\ell}\uparrow \lambda$, we have $X,X_1,X_2,\ldots\in \L_A$ and $X_{\ell}\uparrow X$. According to [LS] (resp.~[LS*]),
$$
\liminf_{\ell\rightarrow\infty}f_{\I}(\lambda_\ell)=\liminf_{\ell\rightarrow\infty}\I(X_{\ell}) \geq \I(X)=f_{\I}(\lambda),
$$
which yields that $\lim_{\ell\rightarrow\infty}f_{\I}(\lambda_\ell)=f_{\I}(\lambda)$, i.e., $f_{\I}$ is left-continuous.
Then we show that $\rho_\I$ satisfies the same property. Suppose otherwise; that is, $\liminf_{\ell\rightarrow \infty}\rho_\I(X_{\ell})< \rho_\I(X)$ holds for some $(X_\ell)_{\ell\in\N}$ with $X_{\ell}\rightarrow X$ almost surely (resp.~in $\mathcal L^\infty$). 
Consider the same discontinuity point $c$ of $f_\I$ as in part (i). Following the same argument presented in (i), there exist $d\in\{c^-,c,c^+\}$ and $\theta\in(0,1)$ such that
$$
    \theta \liminf_{\ell\rightarrow \infty}\rho_\I(X_{\ell})+(1-\theta)\rho_\I(d)
    <c<
    \theta\rho_\I(X)+(1-\theta)\rho_\I(d).
$$
Set $\widetilde X_\ell=\theta X_\ell+(1-\theta)d$ and $\widetilde X=\theta X+(1-\theta)d$. Then $\widetilde X_\ell\rightarrow\widetilde X$ almost surely (resp.~in $\mathcal L^\infty$). By MPH,
$$
    \liminf_{\ell\rightarrow\infty}\rho_\I(\widetilde X_\ell)<c<\rho_\I(\widetilde X).
$$
Consequently, there exist infinitely many $\ell$ such that $\widetilde X_\ell$ is rated better than $\widetilde X$, which yields a contradiction by applying [LS] (resp.~[LS*]) to $\widetilde X$ and the sequence $\widetilde X_1,\widetilde X_2,\dots$ in $\L_A$. So $\rho_\I$ satisfies the same property.
\end{enumerate}
\end{proof}

\subsection{Proofs for Section \ref{sec:choquet_chara}}\label{appx:proof_sec5}

\begin{proof}[Proof of Lemma \ref{lem:RC-trivial}]
 (i)$\Rightarrow$(ii) follows from   
 Theorem 5.2.1 of \cite{DVGKTV06}.
  (ii)$\Rightarrow$(iii) follows
  because  the convex order implies the increasing convex order.  
  The equivalences (i)$\Leftrightarrow$(iii)$\Leftrightarrow$(iv)$\Leftrightarrow$(v) follow from Theorem 3 of \cite{WWW20EC}.

\end{proof}

\begin{proof}[Proof of Theorem \ref{th:PE-LI}]
(i)$\Rightarrow$(ii): 
Note that we can write $h=\lambda \tilde h+(1-\lambda)\delta$ for some $\lambda \in [0,1]$, 
where $\tilde h$ is concave on $[0,1]$ and $\delta$ is given by $\delta(t)=\id_{\{t=1\}}.$
Let $\tilde \rho$ be the distortion risk measure with distortion function $\tilde h$
and $ \rho_\delta$ be the distortion risk measure with distortion function $\delta$.
Note that $\rho_\delta$ is the essential infimum. 
By Lemmas \ref{lem:RC-pool}  and \ref{lem:RC-trivial}, 
we get,  for $\ell\le\ell'$, 
$$\tilde \rho  \left( \frac{(L^{(\ell)}-K)_+}{1-K} \right)\ge\tilde  \rho \left( \frac{(L^{(\ell' )}-K)_+}{1-K} \right).$$  
Moreover, since $ L^{(\ell)} $
and $L^{(\ell')}$ have the same essential infimum, we have 
$$ \rho_\delta \left( \frac{(L^{(\ell)}-K)_+}{1-K} \right)=    \rho_\delta \left( \frac{(L^{(\ell' )}-K)_+}{1-K} \right).$$  
Therefore,  using $\rho=\lambda \tilde \rho + (1-\lambda)\rho_\delta$, we obtain that (ii) holds.

    (ii)$\Rightarrow$(iii) follows trivially by setting $K=0$. 
    
        (iii)$\Rightarrow$(iv) follows trivially by definition. 

    (iv)$\Rightarrow$(i): 
    Denote by $C_h\subseteq (0,1)$ the set of points of continuity of $h$ in $(0,1)$.
    Take an arbitrary $\lambda \in (0,1)$ and let $x,y\in   C_h$ with $x< y$. 
Let $Y$ be a random variable such that $\p(Y=1)=x$, $\p(Y=\lambda)=y-x$, and $\p(Y=0)=1-y$.
Define a sequence $L_1,L_2,\dots \in \mathcal L_{[0,1]}$ 
that is  conditionally iid on $Y$ such that
$L_1$ is distributed as Bernoulli($Y$) given  $Y$.

Clearly, $\rho(L_1) = h(\p(L_1=1)) = h(x+\lambda y -\lambda x)$.
Moreover, $L^{(\ell)}$ converges almost surely to $Y$ as $\ell\to\infty$.
By the dominated convergence theorem,  as $\ell\to\infty$,
\begin{align*} 
\rho(L^{(\ell)}) 
&=\int_0^1 h (\p( L^{(\ell)}>z)) \d z
\\&= 
\int_0^{\lambda} h (\p( L^{(\ell)}>z)) \d z
+
\int_{\lambda}^1 h (\p( L^{(\ell)}>z)) \d z
\\&\to  \lambda  h(y) +(1-\lambda)  h(x).
\end{align*}
Hence, from $\rho(L_1) \ge \rho(L^{(\ell)})$,  we get 
\begin{align}
    \label{eq:hconti}
    h( \lambda y+(1-\lambda) x) \ge  \lambda h(y)  + (1-\lambda) h(x) \mbox{~~~for all  $x,y\in C_h$, $x<y$ and $\lambda \in(0,1)$.}
\end{align} 
Next, consider $x,y\in(0,1)$ at which $h$ is possibly not continuous. 
Since $C_h$ is dense in $[0,1]$, we can find  a sequence $\epsilon_k\downarrow 0$ such that $x+\epsilon_k$
and $y+\epsilon_k$ are in $C_h$ for each $k\in \N$.
For any $\lambda \in (0,1)$
let $\lambda_k\in(0,1) $ be such that 
$\lambda x+ (1-\lambda)y =
\lambda_k x+ (1-\lambda_k)y+\epsilon_k
 $, which exists for $k$ large enough, and moreover $\lambda_k\to \lambda$ as $k\to\infty$.
 By \eqref{eq:hconti}, we get 
\begin{align*}
h( \lambda x+(1-\lambda) y)  & =  h( \lambda_k x+(1-\lambda_k) y+\epsilon_k) 
\\ &\ge \lambda_k h(x+\epsilon_k)  + (1-\lambda_k) h(y+\epsilon_k)\\ & \ge  \lambda_k h(x)  + (1-\lambda_k) h(y)
\to   \lambda  h(x)  + (1-\lambda  ) h(y)  \mbox{~~~as $k\to\infty$}.
\end{align*}
Hence, \eqref{eq:hconti} also holds for $x,y$ outside $C_h$, proving concavity of $h$ on $(0,1).$ Since $h$ is nonnegative, we know that $h$ is also concave on $[0,1)$.

\end{proof}

\begin{proof}[Proof of Example \ref{ex:essinf}]
Denote by $m = \essinf (L_1)$. In fact, we have $\essinf(L^{(\ell)}) = m$ for every $\ell\in\mathbb{N}$. A brief proof is as follows. Since $L^{(\ell)} \geq m$ a.s., it suffices to show that $\p(L^{(\ell)} < m+\varepsilon)>0$ for any $\varepsilon>0$. Note that $(L_\ell)_{\ell\in\N}$ is conditionally iid on a variable $Z$. By definition of $m$,
$$\mathbb{P}(L_1 < m+\varepsilon) =\int\mathbb{P}(L_1 < m+\varepsilon|Z=z)\dd F_Z(z)> 0,\ \forall \varepsilon>0,$$ where $F_Z$ is the distribution function of $Z$. Thus there exists a set of positive probability on which $\mathbb{P}(L_1 < m+\varepsilon|Z) > 0$. Based on the conditional i.i.d. structure, we have 
$$\mathbb{P}(L^{(\ell)} < m+\varepsilon) \geq \mathbb{P}(L_1,\cdots,L_\ell < m+\varepsilon)=\int\mathbb{P}(L_1 < m+\varepsilon|Z=z)^\ell\dd F_Z(z)> 0,\ \forall \varepsilon>0,$$ which finishes the proof.

\end{proof}

\begin{proof}[Proof of Corollary \ref{coro:equiv}]
Let $h$ be the distortion function of a distortion risk measure $\rho$ satisfying [LS]. By Lemma \ref{lem:RC-trivial} and Theorem \ref{th:PE-LI}, 
it suffices to verify that [LS] implies that $h$ is continuous at $1$. 
Take $\lambda_1,\lambda_2,\dots \in (0,1]$ such that $\lambda_{k}\uparrow 1$, and let $U$ be uniformly distributed on $[0,1]$. 
The sequence $\id_{\{U \le \lambda_k\}}$ converges to   $\id_{\{U \le 1\}}$ almost surely as $k\to\infty$.
Hence, 

    $$\liminf_{k \rightarrow\infty}h (\lambda_{k})=\liminf_{k \rightarrow\infty}\rho(\id_{\{U \le \lambda_k\}})\ge \rho(\id_{\{U \le 1\}})=h(1),$$
	where the inequality is due to [LS]. Since $h$ is increasing, $\limsup_{k \rightarrow\infty}h (\lambda_{k})\leq h(1)$. Hence $\lim_{k \rightarrow\infty}h (\lambda_{k})= h(1)$, so $h$ is continuous at $1$. Thus, all four properties are equivalent to concavity of $h$ on $[0,1]$.

\end{proof}

\begin{proof}[Proof of Lemma \ref{lemma:LC}]
Let $\rho$ be the $S$-distortion risk measure with $S$-distortion function $g$.
First, suppose that $\rho$ satisfies [LC]. Let $U$ be a random variable uniformly distributed on $[0,1]$ under each $\p_j$, $j\in [s]$; such $U$ exists because $\p_1,\dots,\p_s$ are mutually singular.
For $\mathbf{x}=(x_1,\dots,x_s)\in [0,1]^s$, denote by $A_{\mathbf x}:=\{U\leq \sum_{j=1}^{s}x_{j}\id_{S_j}\}$ which satisfies $\p^S(A_{\mathbf x})=\mathbf x$. 
This yields $g(\mathbf x)= \rho(\id_{A_{\mathbf x}})$. 
For  a sequence $(\mathbf{x}_{\ell})_{\ell\in \N}$ that converges to $\mathbf{x}$, we have 
$\id_{A_{\mathbf {x}_{\ell}}}\rightarrow\id_{A_{\mathbf {x}}}$ almost surely. By [LC], 
\begin{equation*}
    \lim_{\ell\rightarrow \infty}g(\mathbf{x}_{\ell})=\lim_{\ell\rightarrow \infty}\rho(\id_{A_{\mathbf {x}_{\ell}}})=\rho(\id_{A_{\mathbf {x}}})=g(\mathbf{x}),
\end{equation*}
and thus $g$ is continuous.

Next, suppose that $g$ is continuous. Take a bounded sequence $(X_\ell)_{\ell\in\N}$ such that $|X_\ell|\leq M$ for some $M>0$ and $X_\ell \to X$ almost surely. Then for $j\in [s]$, $X_{\ell}$ converges $\p_j$-almost surely to $X$ as $\ell\to\infty$. This implies 
\begin{equation*}
   \lim_{\ell\to\infty} \p^S(X_{\ell}>z)= \p^S(X>z)
\end{equation*}
for all $z\in [-M,M]$ except on a Lebesgue null set $N\subseteq [-M,M]$.
For  $z \in [-M,M]\setminus N$, the continuity of $g$ implies 
\begin{equation*}
    \lim_{\ell\to\infty}(g\circ \p^S)(X_{\ell}>z)=(g\circ \p^S)(X>z).
\end{equation*}
By the dominated convergence theorem,
\begin{equation*}
\lim_{\ell\to\infty}\rho(X_{\ell})=\lim_{\ell\to\infty}\int_{-M}^{M}(g\circ \p^S)(X_{\ell}>z)\dd z-M=\int_{-M}^{M}(g\circ \p^S)(X>z)\dd z-M=\rho(X).
\end{equation*}
Thus $\rho$ satisfies [LC].

\end{proof}

\begin{proof}[Proof of Theorem \ref{th:SR-Choquet}]
    (i)$\Leftrightarrow$(ii) is well known, and explained in Section \ref{sec:convex}.
 
    (iii)$\Leftrightarrow$(i): By Lemma \ref{lemma:submodular}, the condition in (iii) on $g$ is equivalent to     \eqref{eq:submodular2}. 
    Moreover, we can check that 
    \eqref{eq:submodular2} is equivalent to  
    \begin{equation*}
        g \circ \mathbb{P}^{S}(A\cap B)+g \circ \mathbb{P}^{S}(A\cup B)\leq g \circ \mathbb{P}^{S}(A)+g \circ \mathbb{P}^{S}(B) ~~~\mbox{for all $A,B\in \mathcal{F}$,}
    \end{equation*}
     because $\mathbb{P}^S(A\cap B)+ \mathbb{P}^S(A\cup B) = \mathbb{P}^S(A)+\mathbb{P}^S(B)$, and the above four vectors of probabilities can attain all values of $\mathbf x_1,\dots,\mathbf x_4$ in \eqref{eq:submodular2}.
     Therefore,     \eqref{eq:submodular2} is equivalent to   
 the submodularity of $g \circ \mathbb{P}^{S}$, which is 
  further  equivalent to the coherence of $\rho$, as mentioned in Section \ref{sec:21-def-MPH}.

    Noting that [LS*] holds automatically for Choquet risk measures,  (ii)$\Rightarrow$(iv)$\Rightarrow$(v)$\Rightarrow$(vi) follows from Theorem \ref{thm:summarize} and the definition.

    Below, we assume that $g$ is continuous, and prove the direction    
    (vi)$\Rightarrow$(iii).  Since $g$ is continuous (Lemma \ref{lemma:LC}), it suffices to show that \eqref{eq:submodular4} in Lemma \ref{lemma:submodular} holds. First, fix distinct $i,j\in [s]$. Take $\mathbf{x}=(x_1,\dots,x_s)\in [0,1]^s$ and $\varepsilon,\delta\geq 0$ such that $\mathbf{x}-\varepsilon\mathbf{e}_i,\mathbf{x}+\varepsilon\mathbf{e}_i+\delta\mathbf{e}_j\in [0,1]^s$. 
    Let $Y$ be a random variable with probability mass function $f_k$ under $\p_k$ for $k\in [s]$ given below:
    \begin{equation*} 
        \begin{cases}
            f_i(0) = 1 - x_i-\varepsilon,\ f_i(1/2) = 2\varepsilon,\ f_i(1) = x_i-\varepsilon; \\
            f_j(0) = 1 - x_j-\delta,\ f_j(1/2) = \delta,\ f_j(1) = x_j; \\
            f_k(0) = 1 - x_k,\ f_k(1) = x_k,\ \forall k\neq i,j.
        \end{cases}
    \end{equation*}
    Define a sequence $L_1,L_2,\dots \in \mathcal L_{[0,1]}$ that is conditionally iid on $Y$ under each $\p_k$, $k\in [s]$, such that $L_1$ is uniformly distributed on $\{1/4,3/4\}$ on the event $Y=1/2$ under $\p_i$, and $L_1=Y$ otherwise. 
    The construction of such random variables is possible because the probabilities $\p_k,$ $k\in [s]$ are mutually singular, so we can arbitrarily assign probabilities. 
    
    For the variable $L_1$, its probability mass function $g_i$ under $\p_i$ is given as 
    \begin{equation*}
        g_i(0) = 1 - x_i-\varepsilon,\ g_i(1/4) = \varepsilon,\ g_i(3/4) = \varepsilon, \ g_i(1) = x_i-\varepsilon,
    \end{equation*}
    and it follows that 
    \begin{align*}
        \rho(L_1)&=\int_0^1 g(\mathbb{P}_1(L_1> x),\dots, \mathbb{P}_s(L_1> x))\dd x\\
        &=\frac{1}{4}(g(\mathbf{x}+\varepsilon\mathbf{e}_i+\delta\mathbf{e}_j)+g(\mathbf{x}+\delta\mathbf{e}_j)+g(\mathbf{x})+g(\mathbf{x}-\varepsilon\mathbf{e}_i)).
    \end{align*}
    Moreover, $L^{(\ell)}$ converges almost surely to $Y$ as $\ell\to\infty$, and the sequence $(L^{(\ell)})_{\ell\in \N}$ is bounded. 
    By [LC], $$\lim_{\ell\to\infty} \rho(L^{(\ell)}) =  \rho(Y)= \frac{1}{2}(g(\mathbf{x}+\varepsilon\mathbf{e}_i+\delta\mathbf{e}_j)+g(\mathbf{x}-\varepsilon\mathbf{e}_i)).$$

    Property [$S$-PM] implies
    $\rho(L_1) \geq \rho(L^{(\ell)})$, so \eqref{eq:submodular4} holds for $i\neq j$. 
    Taking $\delta=0$ in \eqref{eq:submodular4},  we obtain that $g$  is componentwise concave. 
    Thus, \eqref{eq:submodular4} also holds for $i=j$, which completes the proof.

\end{proof}

\subsection{Proofs for Section \ref{sec:SEU_chara}}
For proving results in Section \ref{sec:SEU_chara}, we first present two lemmas as follows.

\begin{lemma}
    \label{lem:separation}
    Let $X_1,X_2,Y_1,Y_2$ be integrable random variables. Suppose that $X_1<X_2$ on a set $B\subseteq\Omega$ with $\mathbb P(B)>0$ and
    $\mathbb E[Y_1]<c<\mathbb E[Y_2]<\infty$
    for some $c\in\mathbb R$. Then there exist $D\subseteq B$ with $\mathbb P(D)>0$ and a set $C\subseteq\Omega$ such that, with
    $Y=Y_1\id_C+Y_2\id_{C^c}$,
    \begin{equation}\label{eq:separation}
        \mathbb E[X_1\id_D+Y\id_{D^c}]
        <c<
        \mathbb E[X_2\id_D+Y\id_{D^c}]
        <\infty.
    \end{equation}
\end{lemma}

\begin{proof}[Proof of Lemma \ref{lem:separation}]
    By atomlessness of $\mathbb P$ and the absolute continuity of the integrals of $X_2-Y_1$ and $X_1-Y_2$, we may choose $D\subseteq B$ with $\mathbb P(D)>0$ sufficiently small such that
    $$
        \mathbb E[Y_1]+\int_D(X_2-Y_1)\dd\mathbb P
        <c<
        \mathbb E[Y_2]+\int_D(X_1-Y_2)\dd\mathbb P
        <\infty.
    $$
    Equivalently,
    $$
        \int_{D^c}Y_1\dd\mathbb P
        <c-\int_DX_2\dd\mathbb P
        <c-\int_DX_1\dd\mathbb P
        <\int_{D^c}Y_2\dd\mathbb P.
    $$
    By Lyapunov's convexity theorem and the atomlessness of $\mathbb P$, there exists $C\subseteq D^c$ such that
    $$
        \int_CY_1\dd\mathbb P+\int_{D^c\setminus C}Y_2\dd\mathbb P
        \in
        \left(
        c-\int_DX_2\dd\mathbb P,\,
        c-\int_DX_1\dd\mathbb P
        \right).
    $$
    Let $Y=Y_1\id_C+Y_2\id_{C^c}$. Then \eqref{eq:separation} holds.

\begin{lemma}\label{lemma:SEU_SPE}
Let $\I$ be an SEU rating criterion with SEU measure $\rho_{\I}$. If $\I$ satisfies [$S$-PE] in the case $A=[0,1]$, or [$S$-PM] for general $A$, then \eqref{eq_rho_SPE} holds.
\end{lemma}

\end{proof}

\begin{proof}[Proof of Lemma \ref{lemma:SEU_SPE}]
Note that in the case $A=[0,1]$, [$S$-PE] implies [$S$-PM], so we may assume that $\I$ satisfies [$S$-PM] for general $A$. Write $\rho=\rho_{\I}$. By Lemma \ref{lemma:normal_SEU},
$\rho(X)=\int_{\Omega}v_\omega(X(\omega))\dd\mathbb P$,
$X\in\mathcal L_A,
$
where $v_\omega(0)=0$ for $\omega\in \Omega$. Let $c$ be an effective cut point of the representing function $f$, and take $Y_1,Y_2\in\mathcal L_A$ such that
$$
    \rho(Y_1)<c<\rho(Y_2)<\infty.
$$

Fix $j\in[s]$ and $x<y<1$, $x,y\in A$, and set $m=(x+y)/2$. We first prove
\begin{equation}\label{eq_v_sup_inf}
    \esssup_{\omega\in S_j}\{v_\omega(m)-v_\omega(x)\}
    \le
    \essinf_{\omega\in S_j}\{v_\omega(y)-v_\omega(m)\}.
\end{equation}
Suppose otherwise. Then there exists $r\in\mathbb R$ such that
$$
    B_1=\{\omega\in S_j:v_\omega(m)-v_\omega(x)>r\},
    \quad
    B_2=\{\omega\in S_j:v_\omega(y)-v_\omega(m)<r\}
$$
both have positive probability. By atomlessness of $\mathbb P$, we can choose pairwise disjoint measurable sets $D_1,D_2,D_3,D_4\subset S_j$ with the same probability $p>0$ sufficiently small, satisfying $D_1\subset B_1$, $D_2\subset B_2$ and all integrals below are finite, such that
$$\rho(Y_1)
    -\int_D v_\omega(Y_1(\omega))\,d\mathbb P
    +\int_{D_1\cup D_2}v_\omega(m)\,d\mathbb P
    +\int_{D_3}v_\omega(x)\,d\mathbb P
    +\int_{D_4}v_\omega(y)\,d\mathbb P<c$$
and
$$\rho(Y_2)-\int_D v_\omega(Y_2(\omega))\dd \mathbb P
    +\int_{D_1\cup D_3}v_\omega(x)\dd\mathbb P
    +\int_{D_2\cup D_4}v_\omega(y)\dd\mathbb P>c,$$ 
where $D=D_1\cup D_2\cup D_3\cup D_4$. 
Equivalently,
$$
\begin{aligned}
    \int_{D^c}v_\omega(Y_1(\omega))\,d\mathbb P
    &<
    c-\int_{D_1\cup D_2}v_\omega(m)\,d\mathbb P
      -\int_{D_3}v_\omega(x)\,d\mathbb P
      -\int_{D_4}v_\omega(y)\,d\mathbb P  \\
    &<
    c-\int_{D_1\cup D_3}v_\omega(x)\,d\mathbb P
      -\int_{D_2\cup D_4}v_\omega(y)\,d\mathbb P
    <
    \int_{D^c}v_\omega(Y_2(\omega))\,d\mathbb P .
\end{aligned}
$$
The middle inequality follows from
$$
\int_{D_1\cup D_2}v_\omega(m)\,d\mathbb P
-
\int_{D_1}v_\omega(x)\,d\mathbb P
-\int_{D_2}v_\omega(y)\,d\mathbb P
>r\mathbb P(D_1)-r\mathbb P(D_2)=0.
$$
By Lyapunov's convexity theorem, there exists a set $C\subset D^c$ such that
$\int_C v_\omega(Y_1(\omega))\,d\mathbb P
+\int_{D^c\setminus C}v_\omega(Y_2(\omega))\,d\mathbb P$ belongs to 
$$
\Bigg(
c-\int_{D_1\cup D_2}v_\omega(m)\,d\mathbb P
  -\int_{D_3}v_\omega(x)\,d\mathbb P
  -\int_{D_4}v_\omega(y)\,d\mathbb P,\,
c-\int_{D_1\cup D_3}v_\omega(x)\,d\mathbb P
  -\int_{D_2\cup D_4}v_\omega(y)\,d\mathbb P
\Bigg).
$$

Let $Y=Y_1\id_C+Y_2\id_{D^c\setminus C}$ and define
$$X_1=x\id_{D_1\cup D_3}
    +y\id_{D_2\cup D_4}
    +Y\id_{D^c},\quad
X_2=x\id_{D_2\cup D_3}+y\id_{D_1\cup D_4}+Y\id_{D^c}.
$$
Then $X_1$ and $X_2$ can be embedded as the first two elements of a sequence which is conditionally iid under every scenario. Specifically, on $D^c$, all elements of the sequence are equal to the common background $Y$, and hence are conditionally iid given $Y$. On $D\subseteq S_j$, conditional on $D$, the sequence is constructed to be iid and equally distributed on $\{x,y\}$, with the first two coordinates realized by the four equal-probability sets $D_1,D_2,D_3,D_4$. Moreover,
$$
    \rho(X_1)<c<\rho\left(\frac{X_1+X_2}{2}\right),
$$
which contradicts [$S$-PM] of $\I$. Hence \eqref{eq_v_sup_inf} holds.

Next, we show that for every $x\in A^\circ$, 
\begin{equation}
    \esssup_{\omega\in S_j}v_\omega(x)
    =
    \essinf_{\omega\in S_j}v_\omega(x).\label{eq_v_sup=inf}
\end{equation}
Suppose otherwise, that this fails for some $x>0$, without loss of generality.
Let $n\in\mathbb N$ and set $h=x/n$. Applying \eqref{eq_v_sup_inf} with $x=(k-1)h,\ y=(k+1)h$ gives
$$
\esssup_{\omega\in S_j}
\{v_\omega(kh)-v_\omega((k-1)h)\}
\le
\essinf_{\omega\in S_j}
\{v_\omega((k+1)h)-v_\omega(kh)\}
$$
for $k\in [n]$. Hence, choosing $0<\delta<
\esssup_{\omega\in S_j}v_\omega(x)
-
\essinf_{\omega\in S_j}v_\omega(x)$, we have
\begin{align*}
\delta
&\leq
\sum_{k=1}^n
\esssup_{\omega\in S_j}
\{v_\omega(kh)-v_\omega((k-1)h)\}
-
\sum_{k=1}^n
\essinf_{\omega\in S_j}
\{v_\omega(kh)-v_\omega((k-1)h)\}\\
& \leq
\sum_{k=1}^n
\essinf_{\omega\in S_j}
\{v_\omega((k+1)h)-v_\omega(kh)\}
-
\sum_{k=1}^n
\essinf_{\omega\in S_j}
\{v_\omega(kh)-v_\omega((k-1)h)\}\\
&\leq 
\essinf_{\omega\in S_j}
\{v_\omega((n+1)h)-v_\omega(nh)\}.
\end{align*}
Since $n$ can be chosen arbitrarily large and $v_{\omega}$ is increasing, for any $y\in A^\circ$ with $y>x$, 
$$
    v_\omega(y)-v_\omega(x)
    \geq
    \delta,\quad \mathbb P_j\text{-a.s..}
$$

Now fix $y>x$ and let $h'=(y-x)/m$. Applying \eqref{eq_v_sup_inf} to the grid
$x,x+h',\ldots,x+mh'=y$, we obtain
$$\essinf_{\omega\in S_j}
\{v_\omega(x+kh')-v_\omega(x+(k-1)h')\}
\ge
\essinf_{\omega\in S_j}
\{v_\omega(x+h')-v_\omega(x)\}\geq \delta$$
for $k\in[m]$. Thus,
$$
\essinf_{\omega\in S_j}\{v_\omega(y)-v_\omega(x)\}
\geq m\delta.
$$
Letting $m\to\infty$ gives a contradiction. 
The case $x<0$ follows analogously by applying the positive-case argument to the transformed functions $\tilde{v}_{\omega}(t):=-v_{\omega}(-t)$ for $t>0$. The case $x=0$ is trivial because $v_\omega(0)=0$. Hence \eqref{eq_v_sup=inf} holds.

Finally, define $v_j(x):=\esssup_{\omega\in S_j}v_\omega(x)$ for $j\in [s]$. Then by \eqref{eq_v_sup=inf}, for all $x\in A^\circ$,
\begin{equation}
    v_{\omega}(x)=v_j(x),\quad \mathbb P_j\text{-a.s.}.\label{eq_vw_SPE}
\end{equation}

By the same argument as in Proposition \ref{prop_SEU_LI_SLI} (ii), there exists a null set $N_j\subseteq S_j$ such that $v_\omega(x)=v_j(x)$, $x\in A^\circ$, for all $\omega\in S_j\setminus N_j$. 
By definition, $v_j(x)$ is an increasing function with $v_j(0)=0$. Moreover, \eqref{eq_v_sup_inf} implies 
$$2v_j(m)\leq v_j(x)+v_j(y),\quad x,y \in A^{\circ},$$
so each $v_j$ is convex and hence continuous on $A^\circ$. Note that $A$ can be either $[0,1]$ or $(-\infty, 1]$. If $A$ has a finite lower endpoint, i.e. $A=[0,1]$, the same endpoint argument based on \eqref{eq_v_sup_inf} gives $\lim_{x\downarrow0}v_j(x)=v_j(0)=0$, so the convexity extends to $0$. Thus $v_j$ is increasing, finite, continuous and convex on $A\setminus\{1\}$.

For the finite upper endpoint $1\in A$, define 
$v_j(1):=\lim_{x\uparrow1}v_j(x)\in\overline{\mathbb R}$ for each $j\in [s]$. Then $v_j$ is increasing and convex on $A$.
For $\omega\in S_j$, define $Z(\omega)=v_\omega(1)-v_j(1)$ if $v_j(1)<\infty$, and $Z(\omega)=0$ if $v_j(1)=\infty$. Since $v_\omega$ is increasing, by taking $x\uparrow 1$ along a countable sequence, we have $Z\geq 0$.
Therefore,
\begin{align*}
    \rho(X)
    =&\sum_{j=1}^{s}\Bigg[
        \int_{S_j} v_j(X(\omega)) \id_{\{X(\omega)<1\}}\dd\mathbb P
        + \int_{S_j}
        v_{\omega}(1)
        \id_{\{X(\omega)=1\}}\dd\mathbb P
    \Bigg]\\
    =&\sum_{j=1}^s \mathbb E[v_j(X)\id_{ S_j}]+\mathbb{E}[Z\id_{\{X=1\}}].
\end{align*}
This completes the proof of Lemma \ref{lemma:SEU_SPE}.

\end{proof}

\begin{proof}[Proof of Theorem \ref{thm:char-SEU}]
We prove the characterization results and, at the same time, the corresponding
SEU translation assertions in Theorem~\ref{th:convert}.

In the following, write $\rho=\rho_{\I}$.
\begin{enumerate}[(i)]
    \item 
We prove the results for [QC].
If $\rho$ admits the representation \eqref{eq_SEU_expectation} with $v_\omega$ a convex loss function for each $\omega\in \Omega$, then $\rho$ is convex and hence satisfies [QC]. The property [QC] of $\rho$ immediately implies [QC] of $\I=f\circ\rho$, since $f$ is increasing. Conversely, suppose that $\I$ satisfies [QC].

By Lemma \ref{lemma:normal_SEU}, $\rho$ admits the form $\rho(X)=\int_\Omega v_\omega(X(\omega))\dd\mathbb{P}$, $X\in\mathcal L_A$, where $v_{\omega}(0)=0$ for $\omega\in \Omega$. First, we claim that for any $a,b\in A\setminus\{1\}$, 
    \begin{equation}
        v_\omega\left(\frac{a+b}{2}\right)\leq
        \frac{v_\omega(a)+v_\omega(b)}{2},\quad
        \mathbb{P}\text{-a.s..}\label{eq_convex_QC}
    \end{equation}
    Suppose otherwise, that $B=\{\omega: v_\omega\left((a+b)/{2}\right)>(v_\omega(a)+v_\omega(b))/2\}$ is a set with positive probability. Since $a,b,(a+b)/2<1$, the values involved are finite pointwise. By passing to a positive-measure subset of $B$ if necessary, we may assume that
$v_\omega(a)$, $v_\omega(b)$, and $v_\omega((a+b)/2)$ are bounded on $B$.

By Lemma \ref{lem:separation}, there exists $D\subseteq B$ with $\mathbb{P}(D)>0$ and $Y\in \L_A$ of the form $Y=Y_1\id_{C}+Y_2\id_{C^c}$ for some $C\subseteq \Omega$ such that 
    \begin{equation}
        \int_{D}\frac{v_\omega(a)+v_\omega(b)}{2} \dd \mathbb{P}+\int_{D^c}v_{\omega}(Y(\omega)) \dd \mathbb{P}
    <c<\int_{D}v_\omega\left(\frac{a+b}{2}\right) \dd \mathbb{P}+\int_{D^c}v_{\omega}(Y(\omega)) \dd \mathbb{P}<\infty.\label{eq_separartion_QC1}
    \end{equation}
    By Lyapunov's convexity theorem, we may choose $D_1\subseteq D$ and $D_2=D\setminus D_1$ such that 
    \begin{equation}
        \int_{D_1}v_\omega(a)\dd \mathbb{P}+\int_{D_2}v_\omega(b) \dd \mathbb{P}
    =\int_{D}\frac{v_\omega(a)+v_\omega(b)}{2} \dd \mathbb{P}.\label{eq_separartion_QC2}
    \end{equation}
    
    Define 
    $$X_1=a\id_{D_1}+b\id_{D_2}+Y\id_{D^c}\mathrm{~~and~~} X_2=b\id_{D_1}+a\id_{D_2}+Y\id_{D^c}.$$
    Combining \eqref{eq_separartion_QC1} and \eqref{eq_separartion_QC2},
    we have
    $\rho(X_1)=\rho(X_2)<c<\rho\left((X_1+X_2)/{2}\right)<\infty$,
    which contradicts [QC] of $\mathcal I$. Thus, \eqref{eq_convex_QC} holds for any $a,b\in A\setminus\{1\}$. Consequently, 
    $$N=\left\{\omega\in \Omega: v_\omega\left(\frac{a+b}{2}\right) >
    \frac{v_\omega(a)+v_\omega(b)}{2} \mathrm{~for~some~} a,b\in \mathbb Q\cap A\setminus\{1\}\right\}$$
    is a null set.
    Since each $v_\omega$ is increasing and $\mathbb Q$ is dense in $A\setminus\{1\}$, it follows that $v_\omega$ is convex on $A\setminus\{1\}$ for all $\omega\notin N$. The convexity extends to the endpoint $1$ in the extended sense by monotonicity. Redefining $v_\omega$ on the null set $N$, if necessary, we obtain a representation of the form \eqref{eq_SEU_expectation} in which each $v_\omega$ is a convex loss function. 

    In particular, $\rho$ is convex and hence satisfies [QC].
Thus we have also shown
$$
    \I\text{ satisfies [QC]}
    \quad\Longleftrightarrow\quad
    \rho\text{ satisfies [QC]}.
$$

 \item We give the proof for [$S$-RA]. If $\rho$ admits the representation \eqref{eq_SEU_SLI} with $v_j$ a convex loss function for each $j\in [s]$, then $\rho$ satisfies [$S$-RA] by definition. The property [$S$-RA] of $\rho$ immediately implies [$S$-RA] of $\I=f\circ\rho$, since $f$ is increasing. Conversely, suppose that $\I$ satisfies [$S$-RA]. 
    
    By Theorem \ref{thm:summarize}, $\mathcal I$ satisfies [$S$-LI]. By Proposition \ref{prop_SEU_LI_SLI},
    $\rho(X)=\sum_{j=1}^s \mathbb E[v_j(X)\id_{ S_j}]$, $X\in\mathcal L_A$,
    where each $v_j$ is an increasing
    loss function with $v_j(0)=0$. 
    First, we claim that for any $j\in [s]$ and $a,b\in A\setminus\{1\}$, 
    \begin{equation}
        v_j\left(\frac{a+b}{2}\right)
    \leq
    \frac{v_j(a)+v_j(b)}{2}.\label{eq_vj_convex}
    \end{equation}
    Suppose otherwise, that \eqref{eq_vj_convex} fails for some $j\in [s]$ and $a,b\in A\setminus\{1\}$. Define $v_{\omega}(x)=\sum_{j=1}^{s}v_j(x)\id_{S_j}(\omega)$, $x\in A$, then $\rho$ admits the form \eqref{eq_SEU_expectation}. By Lemma \ref{lem:separation}, there exists $D\subseteq S_j$ with $\mathbb{P}(D)>0$ and $Y\in \L_A$ of the form $Y=Y_1\id_{C}+Y_2\id_{C^c}$ for some $C\subseteq \Omega$ such that 
    \begin{equation}
        \int_{D}\frac{v_j(a)+v_j(b)}{2} \dd \mathbb{P}+\int_{D^c}v_{\omega}(Y(\omega)) \dd \mathbb{P}
    <c<\int_{D}v_j\left(\frac{a+b}{2}\right) \dd \mathbb{P}+\int_{D^c}v_{\omega}(Y(\omega)) \dd \mathbb{P}.\label{eq_separartion_RA}
    \end{equation}
    By atomlessness of $\mathbb{P}$, split $D$ into disjoint sets $D_1,D_2$ with $\mathbb{P}(D_1)=\mathbb{P}(D_2)>0$. 
    Define 
    $$X_1=a\id_{D_1}+b\id_{D_2}+Y\id_{D^c}\mathrm{~~and~~}X_2=\frac{a+b}{2}\id_{D}+Y\id_{D^c}.$$
    It follows that
    $\rho(X_1)<c<\rho(X_2).$
    Moreover, $\mathbb{E}[X_1|\sigma(\id_D,Y\id_{D^c})]=X_2$ on $S_j$ and $X_1=X_2$ on other scenarios implies $X_2\leq_{S\opName{-icx}}X_1$, which
    contradicts [$S$-RA] of $\I$. Hence each $v_j$ is convex on $A\setminus\{1\}$. The convexity extends to the endpoint $1$ in the extended sense by monotonicity. Therefore, each $v_j$ is convex. Consequently, $\rho$  satisfies [$S$-RA]. Hence
$$
    \I\text{ satisfies [$S$-RA]}
    \quad\Longleftrightarrow\quad
    \rho\text{ satisfies [$S$-RA]}
    \quad\Longleftrightarrow\quad
    \rho\text{ admits \eqref{eq_SEU_SLI} with convex }v_j.
$$

    \item
    Now we give the proof for [$S$-PM]. The implications from $\rho$ to $\mathcal I=f\circ\rho$ are immediate, since $f$ is increasing.     
    If $\mathcal I$ satisfies [$S$-PM], then by Lemma \ref{lemma:SEU_SPE}, $\rho$ admits the form \eqref{eq_rho_SPE}. Next, suppose that $\rho$ admits the form \eqref{eq_rho_SPE}. By part (ii), the functional $\sum_{j=1}^s \mathbb E[v_j(X)\id_{ S_j}]$ with convex functions $v_j$ satisfies [$S$-RA]. By Theorem \ref{thm:summarize}, it further satisfies [$S$-PM]. Moreover, since 
    $$\{L^{(\ell')}=1\}\subseteq \{L^{(\ell)}=1\} \mathrm{~for~} \ell\leq \ell',$$
    it is straightforward that $\mathbb{E}[Z\id_{\{X=1\}}]$
    satisfies [$S$-PM] for general $A$. Thus, $\rho$ satisfies [$S$-PM]. We have therefore shown
$$
    \I\text{ satisfies [$S$-PM]}
    \quad\Longleftrightarrow\quad
    \rho\text{ satisfies [$S$-PM]}
    \quad\Longleftrightarrow\quad
    \rho\text{ admits \eqref{eq_rho_SPE}}.
$$

    \item For [$S$-PE] when $A=[0,1]$, the two implications from $\rho$ to $\mathcal I$ and from $\mathcal I$ to the form \eqref{eq_rho_SPE} are the same as in part (iii). If $\rho$ admits the form \eqref{eq_rho_SPE}, the first component $\sum_{j=1}^s \mathbb E[v_j(X)\id_{S_j}]$ satisfies [$S$-RA] by part (ii), hence [$S$-PE] by Theorem \ref{thm:summarize}. For the second component, since
$$
\left\{\frac{(L^{(\ell')}-K)_+}{1-K}=1\right\}
= \{L^{(\ell')}=1\}
\subseteq \{L^{(\ell)}=1\}
= \left\{\frac{(L^{(\ell)}-K)_+}{1-K}=1\right\}
$$
for $K\in[0,1)$ and $\ell\le\ell'$, it is immediate that the functional $X\to \mathbb{E}[Z\id_{\{X=1\}}]$
    satisfies [$S$-PE]. 
Hence $\rho$ satisfies [$S$-PE]. Therefore
$$
    \I\text{ satisfies [$S$-PE]}
    \quad\Longleftrightarrow\quad
    \rho\text{ satisfies [$S$-PE]}
    \quad\Longleftrightarrow\quad
    \rho\text{ admits \eqref{eq_rho_SPE}}.
$$

\end{enumerate}

Finally, taking $S=\Omega$ in parts~(ii)--(iv) yields the corresponding
translation results for [RA], [PM], and, when $A=[0,1]$, [PE].
Together with Proposition~\ref{prop_SEU_LI_SLI}, this establishes all
SEU assertions in Theorem~\ref{th:convert}.

\end{proof}

\begin{proof}[Proof of Theorem \ref{thm_SR_SEU}]
To apply Theorem \ref{thm:char-SEU} to the risk measure $\rho$, let $a=2\rho(1/2)$ and 
$\widetilde{\rho}:=\rho/a$.
Then $\widetilde{\rho}(1/2)=1/2$. Choose an increasing function
$f:\overline{\R}\to[n]$ with an effective cut point and define
$\mathcal I=f\circ\widetilde{\rho}$. 
Thus, $\widetilde{\rho}$ is the
normalized representing SEU risk measure of $\mathcal I$.

Hence Theorem \ref{thm:char-SEU} applies to $\widetilde{\rho}$. Since multiplication by a positive
constant preserves all the properties and representation forms involved in
Theorem \ref{thm:char-SEU}, the corresponding implications also hold for $\rho$. 
The remainder of the proof follows.

Since [LC] implies [LS*], Theorem \ref{thm:summarize} gives (i) $\Rightarrow$ (ii) $\Rightarrow$ (iv) $\Rightarrow$ (iii) when $A=[0,1]$, and (i) $\Rightarrow$ (ii) $\Rightarrow$ (iii) when $A=(-\infty,1]$. 
If (v) holds, i.e. $\rho$ admits the form \eqref{eq_SEU_SLI} where each $v_j$ is a convex loss function, then $\rho$ naturally satisfies [$S$-LI]. Moreover, since each $v_j$ is convex, each mapping $X\to \mathbb{E}[v_j(X)\id_{S_j}]$ is convex, hence $\rho$ is convex. Thus $\rho$ satisfies [QC], proving (v) $\Rightarrow$ (i). 

It remains to prove (iii) $\Rightarrow$ (v). By Theorem \ref{thm:char-SEU}, if (iii) holds, then $\rho$ admits the form \eqref{eq_rho_SPE} where each $v_j$ is a convex loss function, and $Z\geq0$ is measurable. Define 
$v_j(1-)=\lim_{x\uparrow1}v_j(x)$ for each $j\in [s]$. By monotonicity, $v_j(1-)\leq v_j(1)$. 
We show that, under [LC], each $v_j(1-)=v_j(1)$, $v_j(1)$ is finite, and $Z=0$ a.s.. 

Suppose that $v_j(1-)=\infty$ for some $j\in[s]$. Choose a positive sequence $x_n\uparrow1$ such that $v_j(x_n)\geq 2^n$. Since $\mathbb P_j$ is atomless, we can choose pairwise disjoint sets $B_n\subseteq S_j$ such that $\mathbb{P}(B_n)=2^{-n}\mathbb{P}(S_j)$.
Then we have
$$x_n\id_{B_n}\to0 \mathrm{~a.s.~and~} |x_n\id_{B_n}|\leq 1.$$
Hence [LC] gives
$\rho(x_n\id_{B_n})\to\rho(0)=0$.
On the other hand, 
$$
\rho(x_n\id_{B_n})\geq \mathbb{P}(B_n)v_j(x_n)\geq \mathbb{P}(S_j)>0,
$$
a contradiction. Therefore $v_j(1-)<\infty$ for every $j\in[s]$.

Now take any sequence $x_n\uparrow1$. Since each $v_j(x_n)\to v_j(1-)$, we have
$$
\rho(x_n)
=\sum_{j=1}^{s}\mathbb P(S_j)v_j(x_n)
\longrightarrow
\sum_{j=1}^{s}\mathbb P(S_j)v_j(1-)<\infty.
$$
Meanwhile, by [LC], applied to the constant random variables $x_n$, we get
$$
\rho(x_n)\longrightarrow \rho(1)
=\sum_{j=1}^s\mathbb P(S_j)v_j(1)+\mathbb E[Z].
$$
Comparing the two limits yields
$$
\sum_{j=1}^s\mathbb P(S_j)\bigl(v_j(1)-v_j(1-)\bigr)+\mathbb E[Z]=0.
$$
By monotonicity, $v_j(1)-v_j(1-)\geq 0$ for every $j\in[s]$, and
$\mathbb E[Z]\geq 0$. Since $\mathbb P(S_j)>0$, it follows that
$v_j(1)=v_j(1-)$ for every $j\in[s]$ and $\mathbb E[Z]=0$. As
$Z\geq 0$, we further have $Z=0$ a.s. and hence (v) holds. This finishes the proof.

\end{proof}

\section{Equivalent conditions of a submodular function}
\label{app:L5}

The following lemma provides   equivalent conditions of a submodular function. 
The proof of the  lemma is similar to Proposition 2.1 of \cite{shaked1990parametric}
and Theorem 2.5 of \cite{mullerScarsini2001}, but the latter two results are formulated for functions on $\R^s$, and hence the results are not the same, which we discuss later.
In what follows, we denote by $\mathbf{e}_i$ the $i$th unit vector in $\mathbb{R}^{s}$  for $i\in [s]$.

\begin{lemma}\label{lemma:submodular}
	For a function $g:[0,1]^s\rightarrow\mathbb{R}$, the following conditions are equivalent.
	\begin{enumerate}[(i)]
		\item $g$ is componentwise concave and submodular on $[0,1]^s$;
		\item For all $\mathbf{x}_1, \mathbf{x}_2, \mathbf{x}_3, \mathbf{x}_4\in [0,1]^{s}$ such that $\mathbf{x}_1\leq \mathbf{x}_2\leq \mathbf{x}_4$, $\mathbf{x}_1\leq \mathbf{x}_3\leq \mathbf{x}_4$ and $\mathbf{x}_1+\mathbf{x}_4 = \mathbf{x}_2+\mathbf{x}_3$,
        \begin{equation}
            g(\mathbf{x}_1)+g(\mathbf{x}_4) \leq g(\mathbf{x}_2)+g(\mathbf{x}_3);\label{eq:submodular2}
        \end{equation}
		\item For all $\mathbf{x}\in [0,1]^{s}$, $\epsilon, \delta\geq 0$ and $i,j\in [s]$ such that 
		$\mathbf{x}+\varepsilon\mathbf{e}_i+\delta\mathbf{e}_j\in [0,1]^s$,
		\begin{equation}
			g(\mathbf{x}) + g(\mathbf{x}+\varepsilon\mathbf{e}_i+\delta\mathbf{e}_j)\leq g(\mathbf{x}+\varepsilon\mathbf{e}_i)+g(\mathbf{x}+\delta\mathbf{e}_j).\label{eq:submodular3}
		\end{equation}
	\end{enumerate}
     Any of (i)--(iii) implies
    \begin{enumerate}[(iv)]
    \item For all $\mathbf{x}\in [0,1]^s$, $\varepsilon,\delta\geq 0$ and $i,j\in [s]$ such that  $\mathbf{x}-\varepsilon\mathbf{e}_i,\mathbf{x}+\varepsilon\mathbf{e}_i+\delta\mathbf{e}_j\in [0,1]^s$,
        \begin{equation}
        g(\mathbf{x}-\varepsilon\mathbf{e}_i)+g(\mathbf{x}+\varepsilon\mathbf{e}_i+\delta\mathbf{e}_j)
        \leq g(\mathbf{x})+g(\mathbf{x}+\delta\mathbf{e}_j).\label{eq:submodular4}
        \end{equation}
    \end{enumerate} 
    Moreover, if $g$ is upper semi-continuous, then (i)--(iv) are equivalent.
\end{lemma}

 Theorem 2.5 of \cite{mullerScarsini2001} and Theorem 3.12.2 of \cite{muller2002comparison} state that   
(i)--(iv) in Lemma \ref{lemma:submodular}, when formulated for functions on $\R^s$, are equivalent. For functions $g$ on $[0,1]^s$, condition (iv) does not   imply the other three conditions when $g$ is not upper semi-continuous. For a  counterexample, consider the function $g$ given by $g(x_1,x_2)=\id_{\{x_1>0\}}\id_{\{x_2>0\}}$ with  $s=2$, which is not upper semi-continuous. It is straightforward to see that $g$ satisfies \eqref{eq:submodular4}, but $g(0,0)+g(1,1)-g(0,1)-g(1,0)=1>0$, contradicting \eqref{eq:submodular2}. Therefore, we need the upper semi-continuity of  $g$.

\begin{proof}[Proof of Lemma \ref{lemma:submodular}]
	(i)$\Rightarrow$(ii): Take vectors  $\mathbf{x}_1 = (x_{1,1}, \ldots, x_{1,s})$, $\mathbf{x}_2 = (x_{2,1}, \ldots, x_{2,s})$, and $\mathbf{y} = (y_1, \ldots, y_s) \in [0,1]^s$ satisfying  $\mathbf{x}_1 \leq \mathbf{x}_2$ and $\mathbf x_2+\mathbf y \in [0,1]^s$. For $j\in [s]$, we have 
	\begin{align*}
	 g(\mathbf{x}_2 + y_j \mathbf{e}_j) - g(\mathbf{x}_2) 
		&\leq  g(\mathbf{x}_2 + (x_{1,j}-x_{2,j}+y_j)  \mathbf{e}_j)
		- g(\mathbf{x}_2 + (x_{1,j}-x_{2,j})  \mathbf{e}_j )\\
		& \leq g(\mathbf{x}_1 + y_j \mathbf{e}_j) - g(\mathbf{x}_1),
	\end{align*}
	where the first inequality follows from componentwise concavity, and the second inequality follows from  submodularity. 
	
	Applying the above arguments iteratively for all $j \in [s]$, we get $g(\mathbf{x}_2 + \mathbf{y}) - g(\mathbf{x}_1 + \mathbf{y}) \leq g(\mathbf{x}_2) - g(\mathbf{x}_1)$. 
  This implies that \eqref{eq:submodular2} holds by taking $\mathbf{y}=\mathbf{x}_4-\mathbf{x}_2=\mathbf{x}_3-\mathbf{x}_1$.
	
	(ii)$\Rightarrow$(iii): This is an immediate special case of \eqref{eq:submodular2}.
	
	(iii)$\Rightarrow$(i): The componentwise concavity follows by taking $i=j$. Take $\mathbf{x} = (x_{1}, \ldots, x_{s})$, $\mathbf{y} = (y_1, \ldots, y_s)$ and $\mathbf{z} = (z_1, \ldots, z_s)\in [0,1]^s$ satisfying $y_iz_i=0$ for all $i \in [s]$ and $\mathbf x+ \mathbf y + \mathbf z \in [0,1]^s$.
  To establish submodularity on $[0,1]^s$, it suffices to show       $$g(\mathbf{x})+g(\mathbf{x}+\mathbf{y}+\mathbf{z})\leq g(\mathbf{x}+\mathbf{y})+g(\mathbf{x}+\mathbf{z}),$$
  for   any $\mathbf x,\mathbf y,\mathbf z $ above.
    Define a difference operator $\Delta_{\mathbf{y}}g(\mathbf{x}):=g(\mathbf{x}+\mathbf{y})-g(\mathbf{x})$. Then for $i,j$ with $y_i,z_j>0$, we have $\Delta_{z_j\mathbf{e}_j}g(\mathbf{x}+y_i\mathbf{e}_i)\leq \Delta_{z_j\mathbf{e}_j}g(\mathbf{x})$
    by \eqref{eq:submodular3}. Applying the above arguments iteratively for all $j \in [s]$,  we get   $$\Delta_{z_j\mathbf{e}_j}g(\mathbf{x}+\mathbf{y})\leq \cdots\leq \Delta_{z_j\mathbf{e}_j}g(\mathbf{x}+y_i\mathbf{e}_i)\leq \Delta_{z_j\mathbf{e}_j}g(\mathbf{x}),$$
    that is, $\Delta_{\mathbf{y}}g(\mathbf{x}+z_j\mathbf{e}_j)\leq \Delta_{\mathbf{y}}g(\mathbf{x})$. Similarly, repeatedly applying the above procedure for  $j \in [s]$, we get
    $$\Delta_{\mathbf{y}}g(\mathbf{x}+\mathbf{z})\leq \cdots\leq \Delta_{\mathbf{y}}g(\mathbf{x}+z_j\mathbf{e}_j)\leq \Delta_{\mathbf{y}}g(\mathbf{x}).$$
 Therefore, $g$ is submodular on $[0,1]^s$.

We have shown that   statements (i)--(iii) are equivalent.
The direction (ii)$\Rightarrow$(iv) is also an immediate consequence of \eqref{eq:submodular2}.
Next, we show (iv)$\Rightarrow$(iii) when $g$ is upper semi-continuous. Taking $\delta=0$ in \eqref{eq:submodular4}, we obtain that $g$ is componentwise concave.
Note that any  upper semi-continuous concave function on $[0,1]$ must be continuous. Hence,    $g$ is componentwise continuous. 
For $\mathbf{x},\mathbf{x}+\varepsilon\mathbf{e}_i+\delta\mathbf{e}_j\in [0,1]^s$, fixing some integer $n$, the inequality \eqref{eq:submodular4} implies   
    \begin{equation*}
        g\left(\mathbf{x}+\frac{k+1}{n}\varepsilon\mathbf{e}_i+\delta\mathbf{e}_j\right)
        -g\left(\mathbf{x}+\frac{k}{n}\varepsilon\mathbf{e}_i+\delta\mathbf{e}_j\right)
        \leq g\left(\mathbf{x}+\frac{k}{n}\varepsilon\mathbf{e}_i\right)
        -g\left(\mathbf{x}+\frac{k-1}{n}\varepsilon\mathbf{e}_i\right)
    \end{equation*}
    for all $k=1,\ldots,n-1$. Adding up these inequalities yields 
    \begin{equation*}
        g(\mathbf{x}+\varepsilon\mathbf{e}_i+\delta\mathbf{e}_j)
        -g\left(\mathbf{x}+\frac{1}{n}\varepsilon\mathbf{e}_i+\delta\mathbf{e}_j\right)
        \leq g\left(\mathbf{x}+\frac{n-1}{n}\varepsilon\mathbf{e}_i\right)
        -g(\mathbf{x}).
    \end{equation*}
    Since $g$ is componentwise continuous, we get that \eqref{eq:submodular3} holds by sending  $n\rightarrow\infty$.

\end{proof}

\section{Axiomatic foundation for rating criteria}\label{appx_axioms}

\subsection{Choquet rating criteria and the self-consistency axiom}

\cite{GKWW25} proposed \emph{self-consistency} as an axiom linking a rating criterion to a pricing rule that is compatible with the Modigliani--Miller principle. Specifically, a pricing functional $p:\L_{[0,1]}\to\mathbb{R}$ is said to satisfy the MM theorem if for all $0<K<M$, $L/M\in \L_{[0,1]}$,
$$
p\left(\frac{(L-K)_+}{M-K}\right)(M-K)
+
p\left(\frac{L\wedge K}{K}\right)K
=
p\left(\frac{L}{M}\right)M.
$$
This means that splitting a security into a senior tranche and a junior tranche does not change its total price. A rating criterion $\I$ satisfies self-consistency if there exists a pricing functional $p$ satisfying the MM theorem such that, for any $L_1,L_2\in\L_{[0,1]}$,
$$
\I(L_1)<\I(L_2)\quad \Longrightarrow\quad p(L_1)>p(L_2).
$$

\cite{GKWW25} showed that  
  Choquet rating criteria are the only class that satisfies self-consistency. The intuition behind this result can be explained as follows. It is well-known that Choquet risk measures are characterized by comonotonic additivity (\citealp{S89}). 
Different tranches of structured finance products yield losses  that are comonotonic  because they are all piece-wise linear increasing functions of the loss of the asset pool. The MM theorem implies that any pricing functional compatible with the rating criterion $\I$ must be additive for such comonotonic losses, and thus a weak notion of comonotonic additivity holds. \citet[Theorem 5]{GKWW25} showed that this property is sufficient to establish the Choquet representation of \cite{S89} for the pricing functional, which represents $\I$.
\cite{CMM15}  characterized Choquet risk measures using a different weak notion of  comonotonic additivity motivated by the put--call parity.

\subsection{SEU rating criteria and the sure-thing principle}
The sure-thing principle is defined as follows. For an event $B\in\mathcal F$ and random losses $X,Y,Z,Z'\in\mathcal L_A$, let $X_B Z$ denote the random loss that coincides with $X$ on $B$ and with $Z$ on $B^c$. The sure-thing principle requires
$$
X_B Z\succeq_\rho Y_B Z
\quad\Longleftrightarrow\quad
X_B Z'\succeq_\rho Y_B Z',
$$
or equivalently,
$$
\rho(X_B Z)\leq \rho(Y_B Z)
\quad\Longleftrightarrow\quad
\rho(X_B Z')\leq \rho(Y_B Z').
$$
Thus, once two products have the same losses outside $B$, their relative ranking should depend only on their losses on $B$; changing the common outside loss from $Z$ to $Z'$ should not reverse the ranking.

The SEU risk measure satisfies the sure-thing principle on the finite domain of $\rho$. To see this, consider an event $B \in \mathcal F$ and $X,Y,Z,Z' \in \mathcal L_A$. For $X_B Z$, we have
$$
\rho(X_B Z)
=
\int_B v_\omega(X(\omega))\, d\mu
+
\int_{B^c} v_\omega(Z(\omega))\, d\mu,
$$
and similarly
$$
\rho(Y_B Z)
=
\int_B v_\omega(Y(\omega))\, d\mu
+
\int_{B^c} v_\omega(Z(\omega))\, d\mu.
$$

Hence,
$$
\rho(X_B Z) - \rho(Y_B Z)
=
\int_B \big( v_\omega(X(\omega)) - v_\omega(Y(\omega)) \big)\, d\mu,
$$
so the term involving $Z$ cancels out. Therefore,
$$
\rho(X_B Z) \le \rho(Y_B Z)
\quad \Longleftrightarrow \quad
\rho(X_B Z') \le \rho(Y_B Z').
$$

Without the finiteness qualification of $\rho$, the unrestricted sure-thing principle may fail, because a common outside component with infinite value may turn two distinct finite comparisons into the equality $\infty=\infty$.

For an SEU rating criterion $\mathcal I=f\circ\rho$, since $f$ is a discrete increasing transformation, a weaker no-reversal implication is inherited: 
$$
\mathcal I(X_B Z)<\mathcal I(Y_B Z)
\quad\Longrightarrow\quad
\mathcal I(X_B Z')\leq \mathcal I(Y_B Z').
$$
Thus, the SEU representation provides an event-wise separability foundation for SEU rating criteria.

\section{A counterexample to cross-class uniqueness}
\label{app:counterexampleForUniqueness}

The following example shows that a rating criterion represented by both an MPH
risk measure and an SEU risk measure does not necessarily have a unique representing risk measure. 

\begin{example}\label{ex:nonuniqueness_MPH_SEU}

Let $A=[0,1]$. Fix a cut point
$c\in(0,1)$. First define $\rho_1:\mathcal L_A\to\mathbb R$ by $ \rho_1(X)=m_X,$
where $m_X\in[0,1]$ is the unique solution of $
    \mathbb E[(X-m_X)^3]=0.
$
We can verify that $\rho_1$ is monetary, positive homogeneous and monotone, thus it is an MPH risk measure. 
Next define
$
    v(x)=((x-c)^3+c^3)/(c^3+(1-c)^3),
    \ x\in[0,1].
$
Then $v$ is strictly increasing and satisfies
$
    v(0)=0,
    \
    v(1)=1.
$
Hence
$
    \rho_2(X)=\mathbb E[v(X)],
    \ X\in\mathcal L_A,
$
is an SEU risk measure.
We can verify that $\rho_1$ and $\rho_2$ are not equal. However, we can define a binary rating criterion $\I:\mathcal L_A\to[2]$ by
$$
    \I(X)=
    \begin{cases}
        1, & \rho_1(X)\leq c,\\
        2, & \rho_1(X)>c,
    \end{cases}
     =
    \begin{cases}
        1, & \rho_2(X)\leq v(c),\\
        2, & \rho_2(X) > v(c),
    \end{cases}
$$
where the last equation is due to 
$$
    \rho_1(X)\leq c
     \Longleftrightarrow
    \mathbb E[(X-c)^3]\leq 0 \Longleftrightarrow
    \rho_2(X)\leq v(c).
$$
\end{example}

\section{A necessary condition for pooling effect consistency}\label{appx:necessaryPE}

We provide one necessary condition for [$S$-PE] when $g$ is assumed to be lower semi-continuous only at all  points in $\{0,1\}^s$.
This condition is useful when verifying [$S$-PE] for some of the examples in Section \ref{sec:charac-2}.

\begin{proposition}
\label{prop:SPE2}
Let $\rho$ be an $S$-distortion risk measure with an $S$-distortion function $g$ that is lower semi-continuous at all points in $\{0,1\}^s$. If [$S$-PE] holds, then 
\begin{equation} \label{eq_specon}
g(x_1,\dots,x_s)\ge 
\int_0^1 g (\id_{\{ x_1>z\}},\dots,\id_{\{ x_s>z\}}) \d z,~~~~(x_1,\dots,x_s)\in (0,1)^s.
\end{equation}
In particular, $g(\mathbf x) \geq \min (\mathbf x)  $  for $ \mathbf x \in (0,1)^s. $
\end{proposition}

\begin{proof}[Proof of Proposition \ref{prop:SPE2}]
Take a point $\mathbf{x}=(x_1,\ldots,x_s)\in (0,1)^s$.
Define a sequence $L_1,L_2,\dots \in \mathcal L_{[0,1]}$ that is conditionally iid under each $\mathbb{P}_j$, $j\in [s]$ such that $L_1$ is distributed as Bernoulli($x_j$) under scenario $S_j$. Clearly, 
$$\rho(L_1) = g(\p_1(L_1=1),\dots,\p_s(L_1=1)) = g(x_1,\ldots,x_s).$$
Moreover, for $j\in [s]$, $L^{(\ell)}$ converges $\p_j$-almost surely to $x_j$ as $\ell\to\infty$.  This implies 
$$\p^S(L^{(\ell)}>z)\to (\p_1(x_1>z),\dots,\p_s(x_s>z)) =   (\id_{\{ x_1>z\}},\dots,\id_{\{ x_s>z\}}) $$
   for all $z\in [0,1]$ except on a set $B\subseteq [0,1]$ of Lebesgue measure $0$.
   For  $z \in [0,1]\setminus B$,
   lower semi-continuity of $g$ at points in $\{0,1\}^s$ implies 
$$\liminf_{\ell \to\infty} (g\circ \p^S)(L^{(\ell)}>z) \ge  g(\id_{\{ x_1>z\}},\dots,\id_{\{ x_s>z\}}).$$ 
By Fatou's Lemma,  \begin{align*} 
  \liminf_{\ell \to\infty}   \rho(L^{(\ell)}) 
   =\liminf_{\ell \to\infty}\int_0^1 (g\circ \p^S)(L^{(\ell)}>z)  \d z
  &\ge  \int_0^1 g (\id_{\{ x_1>z\}},\dots,\id_{\{ x_s>z\}}) \d z.
\end{align*} 
  From [$S$-PE], we have that $\rho(L_1) \ge \rho(L^{(\ell)})$ for every $\ell$, and hence \eqref{eq_specon} holds. 

The lower semi-continuity of $g$ at all points in $\{0,1\}^s$ is essential to the conclusion in Proposition \ref{prop:SPE2}, as illustrated by the following example.
\begin{example}
    Let $s=2$ and $g(x_1,x_2)=\id_{\{x_1=1\}}\id_{\{x_2=1\}}$.
    This function $g$ is not lower semi-continuous at $(1,1)$, and it satisfies [$S$-PE], as presented in Table \ref{tab:jump-exs}.

    Clearly, 
    $g(x,\dots,x)\ge x$ does not hold for $x\in (0,1)$, violating \eqref{eq_specon}. 
\end{example}

\end{proof}

\section{Catastrophe bond case study}\label{appx:CATbond}
We present an empirical illustration based on CAT bonds, which have emerged as effective alternatives to traditional reinsurance for
risk mitigation in recent years. The issuance of CAT bonds reached a record high of \$17.7 billion in 2024, and the outstanding market is near \$50 billion.\footnote{See the report at https://www.artemis.bm/news/record-cat-bond-issuance-of-17-7bn-in-2024-takes-outstanding-market-near-50bn-report (accessed September 2026).}

In particular, we consider lightning-specific indemnity CAT bonds, which are structured based on the aggregate lightning-related losses in one or more US states over a one-year period. The capital raised through these bonds is held in reserve to cover potential claims. If the contractually defined trigger--insured losses from lightning events exceeding a specified threshold--is met, the insurer can access the funds to offset payouts to policyholders. The bond's payout is capped and fully exhausted once losses surpass a predetermined upper limit.

For analysis, we use state-level lightning loss data in the US from 1960 to 2012, obtained from the Spatial Hazard Events and Losses Database for the United States (SHELDUS), maintained by the Center for Emergency Management and Homeland Security at  Arizona State University. The dataset includes direct property damage estimates associated with lightning events. To enable meaningful comparisons over time, all loss values are inflation-adjusted using the Consumer Price Index, with 2012 as the base year.

We select five states---Kansas, Michigan, Indiana, Minnesota, and Kentucky---based on the absence of statistically significant pairwise correlations in their annual lightning loss data, as measured by both Spearman and Pearson correlation coefficients. We assume that losses across the five states are  independent, as a reasonable approximation.\footnote{A similar assumption of independence across US states is made in a recent actuarial study on wildfire losses; see \cite{LS24}.}  To model the marginal loss distributions, we rescale the data to units of millions of dollars and apply a small shift of 0.01 (equivalent to \$10{,}000) to ensure all values are strictly positive and to improve distributional fit. For all five states, both the Kolmogorov-Smirnov test and the Anderson-Darling  test fail to reject the null hypothesis that the data follow a lognormal distribution. We also verified that the logarithms of the rescaled data exhibit no statistically significant pairwise correlations.

\begin{table}[t]

\centering
\small
\caption{Estimated parameters of the two-sided truncated lognormal distribution and risk measures by state\label{tab_cat}}
\resizebox{\textwidth}{!}{%
\begin{tabular}{@{}c@{\quad}|@{\quad}c@{\quad}|@{\quad}c@{\quad}|@{\quad}c@{\quad}|@{\quad}c@{\quad}|@{\quad}c@{\quad}|@{\quad}c@{\quad}|@{\quad}c@{\quad}|@{\quad}c@{\quad}|@{\quad}c@{}}
\hline 

state & \tabincell{c}{no.\\obs.} & $\mu$    & $\sigma$ & \tabincell{c}{attach.\\point} & \tabincell{c}{detach.\\point} & PD    & EL    & \tabincell{c}{MAXVAR\\($\gamma=0.3$)}  & \tabincell{c}{ES\\($p=0.9$)}  \\
    \hline
    Kansas & 51    & -0.69  & 1.03  & 1.88  & 7.42  & 0.10  & 0.025  & 0.2499  & 0.3047  \\
    Michigan & 51    & -0.51  & 1.48  & 4.03  & 24.89  & 0.10  & 0.025  & 0.2496  & 0.3092  \\
    Indiana & 51    & -1.02  & 1.67  & 3.07  & 22.44  & 0.10  & 0.025  & 0.2499  & 0.3109  \\
    Minnesota & 50    & -1.26  & 1.60  & 2.21  & 15.17  & 0.10  & 0.025  & 0.2498  & 0.3103  \\
    Kentucky & 47    & -2.04  & 1.65  & 1.07  & 7.65  & 0.10  & 0.025  & 0.2500  & 0.3107 
 
 \\ \hline
\end{tabular}
}

\end{table}

To construct the insurance payoff structure, the attachment point is set at the 90th percentile of the loss distribution for each selected state, which means the probability of first dollar loss (PD) is 0.1. Given that catastrophe bonds are typically rated as speculative-grade securities, such as B grade, we calibrate the upper limit (detachment point) such that the bond's expected loss rate (EL)  is 0.025,\footnote{Swiss Re maintains a B-rated composite for catastrophe bonds with expected losses ranging from 181 to 375 basis points when an external rating is not available; see \cite{S24}. We take the geometric mean of these two values, which is approximately 0.025.} after normalizing the loss to the $[0,1]$ interval by using the range between the attachment point and the detachment point. Since lightning-related losses are largely uncorrelated with financial market movements or economic performance, we assume that their distribution remains invariant across different economic scenarios. Table \ref{tab_cat} reports the number of observations, the estimated parameters $\mu$ and $\sigma$ of the lognormal distribution, along with the attachment and detachment points for two-sided truncation, and the corresponding values of PD, EL, MAXVAR ($\gamma = 0.3$), and ES ($p=0.9$) for each state. While PD and EL are predetermined to be  identical across the five states, the MAXVAR and ES values are also numerically close, indicating that the risk profiles of different states are highly similar after normalization.

\begin{figure}[t]
    \centering
    \caption{Pooling effect on CAT bond ratings as states are added sequentially}
    \label{fig_measures3}
    \begin{subfigure}[b]{0.475\textwidth}
        \centering
        \caption[]{}
        \includegraphics[width=\textwidth]{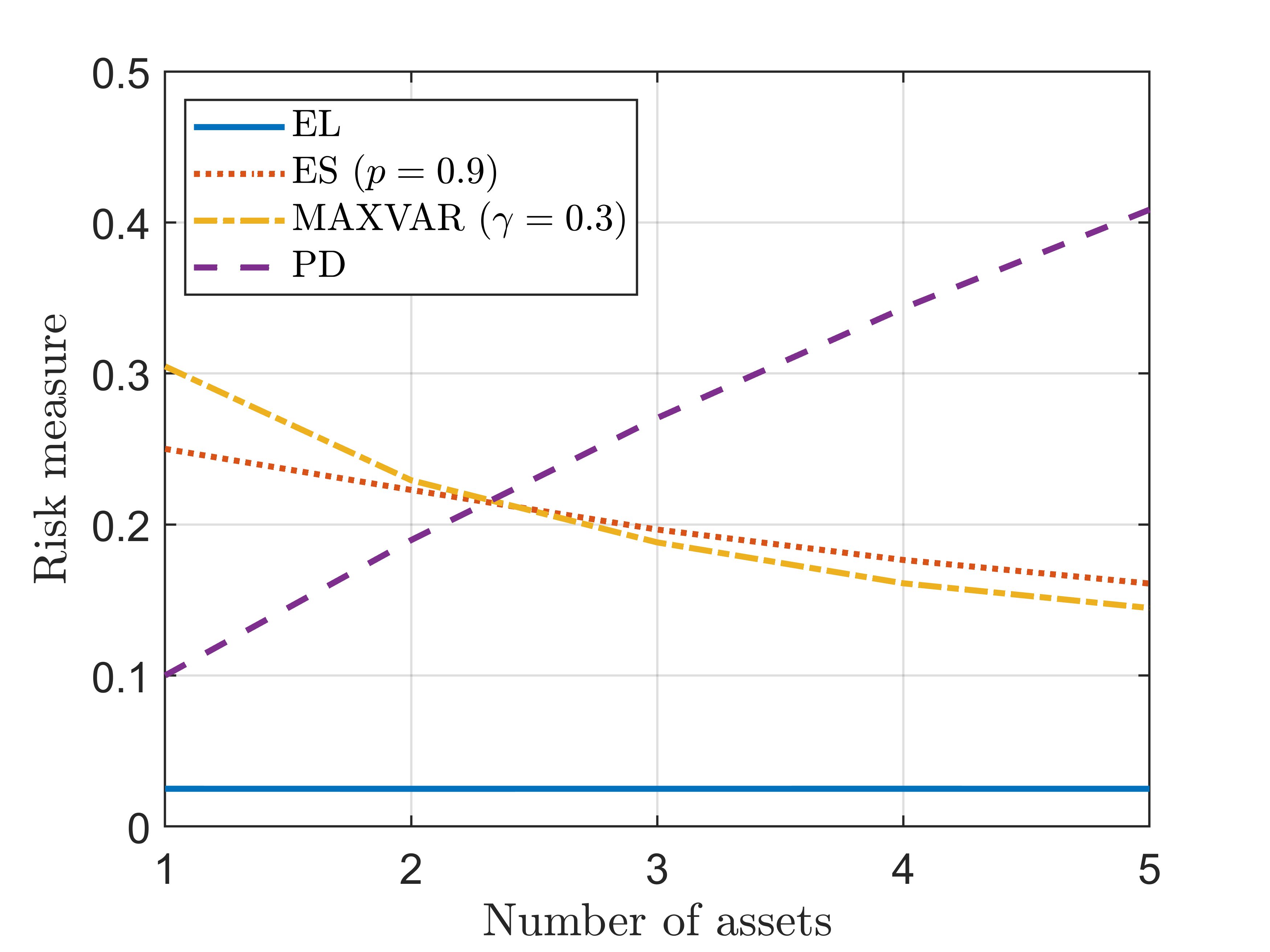}
        
        \label{fig:measure_cat}
    \end{subfigure}
    \hfill
    \begin{subfigure}[b]{0.475\textwidth}
        \centering
        \caption[]{}
        \includegraphics[width=\textwidth]{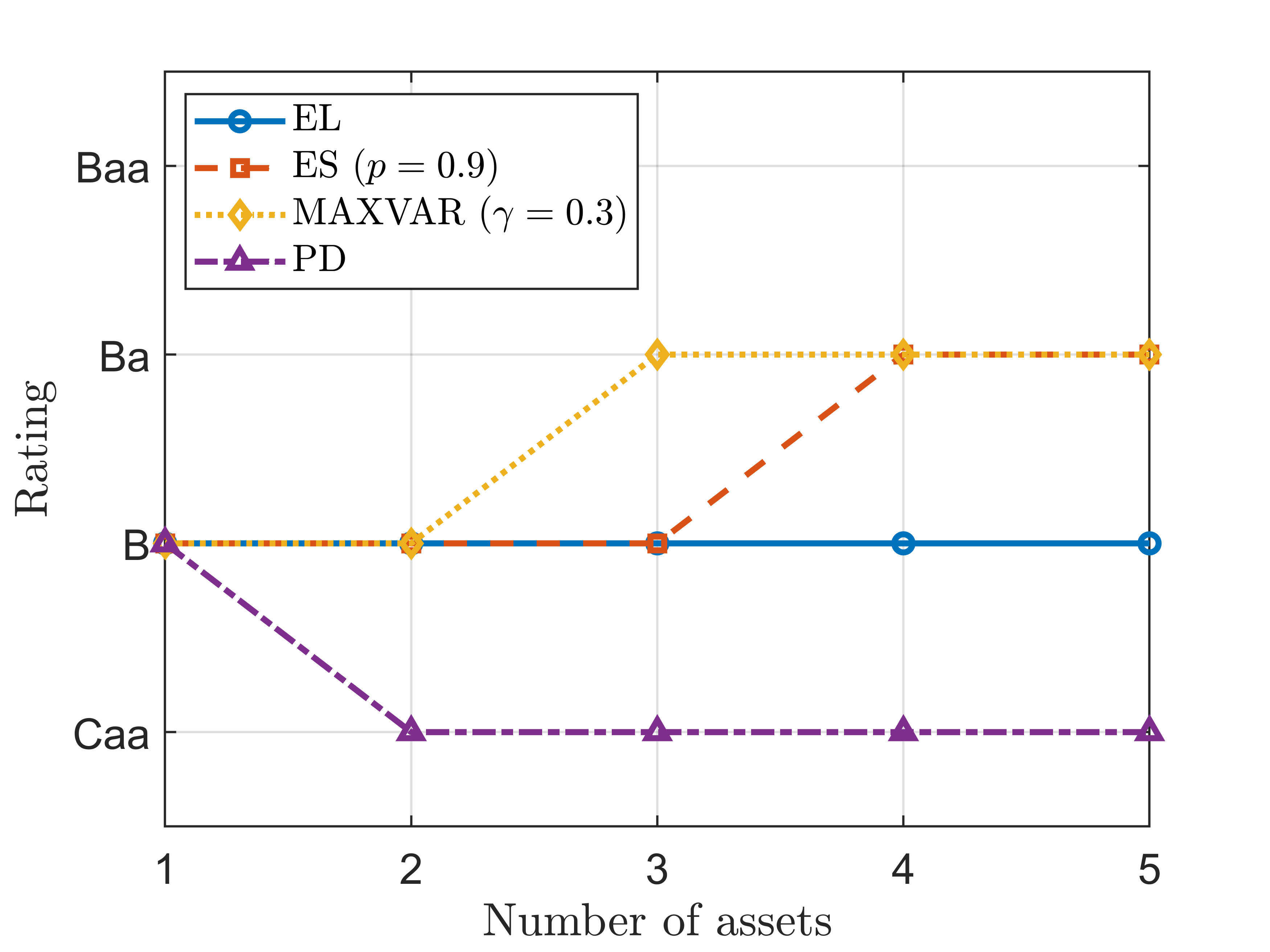}
        
        \label{fig:rating_cat}
    \end{subfigure}
\end{figure}

CAT bonds can be linked to insurance covering losses in a single state or across multiple states. Although the loss distributions across states are not iid, we can still invoke quasi-convexity to illustrate the effect of pooling in the context of diversification. Panel (a) of Figure \ref{fig_measures3} presents the results of risk measures for the EL, MAXVAR, ES, and PD criteria, by sequentially adding states to the portfolio in the following ordering: Kansas, Michigan, Indiana, Minnesota, and Kentucky. In this case study, we set $K=0$, since each insurance already has its own trigger (i.e., the attachment point). 
In this setting, EL remains flat with respect to asset pooling, indicating that expected losses are unaffected by diversification. In contrast, although the state-wise MAXVAR and ES values are similar, their pooled counterparts exhibit a clear decreasing trend relative to the theoretical upper bound given by the line $\rho = 0.25$ (for MAXVAR) and $\rho = 0.31$ (for ES), underscoring their sensitivity to diversification. Moreover, PD shows a consistently increasing pattern throughout, suggesting a different response to risk aggregation. These patterns remain unchanged under different orderings of the states, due to the similarity of their loss profiles as shown earlier in Table \ref{tab_cat}.

Panel (b) of Figure \ref{fig_measures3} presents the rating results, by using the same calibration method as introduced in Section \ref{sec:CLO}. The resulting thresholds for all rating criteria are provided in Table \ref{tab:mapping_catbond} in Appendix \ref{appx:mapping}. All evaluations initially assign a B rating for the given single state. When using the EL criterion, the rating remains unchanged after incorporating all additional states. Under the PD criterion, however, the rating declines to Caa upon the inclusion of the second state. In contrast, MAXVAR improves to Ba after the third state is added, whereas ES
improves to Ba after the fourth state is added. 
The above analysis shows the advantage of using a rating criterion,  other than EL, that satisfies the risk consistency properties studied in our paper.

\section{Numerical thresholds for rating categories }\label{appx:mapping}
\setcounter{table}{0}

This section reports tables for the calibrated thresholds underlying rating categories via different rating criteria, complementing the   results in 
Section \ref{sec:CLO} and Appendix \ref{appx:CATbond}.

Table \ref{tbl:mapping_clo} reports  the threshold values to rating collateralized loan obligations. A given rating is assigned once the value of the risk falls within the corresponding upper threshold. For the Average EL criterion, we use Moody's idealized default and loss rate table from the ``Rating Symbols and Definitions'' document (\url{https://ratings.moodys.com/rmc-documents/53954}, accessed September 2026). The thresholds for the other criteria are calibrated using the method described in Section \ref{sec:CLO}.

\begin{table}[t]
  \centering
  \small
  \caption{Thresholds for rating collateralized loan obligations; the numbers are the upper bounds on a particular rating category}
    \begin{tabular}{c|c|c|c|c|c|c}
    Rating & Average EL & Average ES & \tabincell{c}{Average\\ MAXVAR} & \tabincell{c}{Average\\ VaR}  & Max VaR  & Average PD \\
    \hline
    Aaa   & 0.000016  & 0.000102  & 0.000092  & 0.000023  & 0.000045  & 0.000101  \\
    Aa1   & 0.000171  & 0.001092  & 0.000978  & 0.000241  & 0.000482  & 0.001082  \\
    Aa2   & 0.000374  & 0.002389  & 0.002139  & 0.000528  & 0.001055  & 0.002366  \\
    Aa3   & 0.000781  & 0.004988  & 0.004467  & 0.001102  & 0.002202  & 0.004940  \\
    A1    & 0.001436  & 0.009172  & 0.008214  & 0.002026  & 0.004050  & 0.009084  \\
    A2    & 0.002569  & 0.016408  & 0.014694  & 0.003624  & 0.007245  & 0.016251  \\
    A3    & 0.004015  & 0.025644  & 0.022965  & 0.005665  & 0.011322  & 0.025398  \\
    Baa1  & 0.006050  & 0.038642  & 0.034605  & 0.008536  & 0.017061  & 0.038270  \\
    Baa2  & 0.008690  & 0.055504  & 0.049706  & 0.012260  & 0.024506  & 0.054970  \\
    Baa3  & 0.016775  & 0.107144  & 0.095952  & 0.023667  & 0.047306  & 0.106113  \\
    Ba1   & 0.029040  & 0.185482  & 0.166106  & 0.040971  & 0.081894  & 0.183697  \\
    Ba2   & 0.046255  & 0.295436  & 0.264574  & 0.065259  & 0.130441  & 0.292593  \\
    Ba3   & 0.065230  & 0.416631  & 0.373110  & 0.092030  & 0.183951  & 0.412623  \\
     B1   & 0.088660  & 0.566281  & 0.507127  & 0.125087  & 0.250024  & 0.560833  \\
     B2   & 0.113905  & 0.727523  & 0.651526  & 0.160704  & 0.321216  & 0.720524  \\
     B3   & 0.148775  & 0.950242  & 0.850980  & 0.209901  & 0.419551  & 0.941100  \\
    \end{tabular}
  \label{tbl:mapping_clo}
\end{table}

Table \ref{tab:mapping_catbond} reports the threshold values to rating catastrophe bonds. For the EL criterion, we adopt the thresholds reported in \cite{S24EC}, which provide upper and lower bounds for the BB (equivalent to Ba under Moody's rating system) and B categories. The thresholds for the other criteria are calibrated using the same method described in Section \ref{sec:CLO}.

\begin{table}[t]
  \centering
  \small
  \caption{Thresholds for rating catastrophe bonds; the numbers are the upper bounds on a particular rating category}
    \begin{tabular}{c|c|c|c|c}
    Rating & EL    & ES    & MAXVAR & PD \\
    \hline
    Baa   & 0.0016  & 0.0160  & 0.0195  & 0.0064  \\
    Ba    & 0.0181  & 0.1810  & 0.2207  & 0.0724  \\
    B     & 0.0375  & 0.3750  & 0.4572  & 0.1500  \\
    Caa   & 1.0000  & 1.0000  & 1.0000  & 1.0000  \\
    \end{tabular}
  \label{tab:mapping_catbond}
\end{table}

\end{document}